%% file: draft.tex
\documentclass[onecolumn,11pt]{IEEEtran}

\usepackage[dvipdfmx]{graphicx}
\usepackage{multirow}
\usepackage{amsmath}
\usepackage{amssymb}
\usepackage{theorem}
\usepackage{color}
\DeclareMathAlphabet{\bm}{OML}{cmm}{b}{it}
\theorembodyfont{\rmfamily}
\newtheorem{theorem}{Theorem}
\newtheorem{lemma}{Lemma}

\newtheorem{corollary}{Corollary}
\newtheorem{remark}{Remark}

\newtheorem{assumption}{Assumption}
\newtheorem{example}{Example}
\newcommand{\qed}{\hfill \IEEEQED}

\newcommand{\bol}[1]{\mathbf{#1}}
\newcommand{\rom}[1]{\mathrm{#1}}
\newcommand{\san}[1]{\mathsf{#1}}

\newcommand{\argmax}{\mathop{\rm argmax}\limits}
\newcommand{\argmin}{\mathop{\rm argmin}\limits}
\newcommand{\textchange}[1]{#1}

\newcommand{\inner}[2]{\langle #1 | #2 \rangle} 

\newcommand{\Pe}{\rom{P}_{\rom{e}}}
\newcommand{\barPe}{\bar{\rom{P}}_{\rom{e}}}
\newcommand{\Pse}{\rom{P}_{\rom{s}}}
\newcommand{\barPse}{\bar{\rom{P}}_{\rom{s}}}
\newcommand{\Pce}{\rom{P}_{\rom{c}}}
\newcommand{\barPce}{\bar{\rom{P}}_{\rom{c}}}

\allowdisplaybreaks[2]

\begin{document}

\title{Finite-Length Analyses for Source and Channel Coding on Markov Chains\thanks{Parts of this paper were presented at 51st Allerton conference
and 2014 Information Theory and Applications Workshop.}}

\author{Masahito~Hayashi~\IEEEmembership{Fellow,~IEEE}
\thanks{The first author is with the Graduate School of Mathematics, Nagoya University, Japan. He is also with the Center for
Quantum Technologies, National University of Singapore, Singapore, e-mail:masahito@math.nagoya-u.ac.jp}
and Shun~Watanabe~\IEEEmembership{Member,~IEEE}       
\thanks{The second author is with the Department of Computer and Information Sciences, Tokyo University of Agriculture and Technology, Japan, e-mail:shunwata@cc.tuat.ac.jp.
When the main part of this paper was done, he was with the Department
of Information Science and Intelligent Systems, 
University of Tokushima, Japan.}

\thanks{Manuscript received ; revised }}

\markboth{Journal of \LaTeX\ Class Files,~Vol.~6, No.~1, January~2007}%
{Shell \MakeLowercase{\textit{et al.}}: Bare Demo of IEEEtran.cls for Journals}

\maketitle
\begin{abstract}
We study finite-length bounds for source coding with side information for Markov sources
and channel coding for channels with conditional Markovian additive noise.
For this purpose, we propose two criteria for finite-length bounds. One is the asymptotic optimality and the other is the efficient computability of the bound.
Then, we derive finite-length upper and lower bounds for coding length in both settings
so that their computational complexity is efficient.
To discuss the first criterion, we derive the large deviation bounds, the moderate deviation bounds, and second order bounds for these two topics,
and show that these finite-length bounds achieves the asymptotic optimality in these senses.
For this discussion, we introduce several kinds of information measure for transition matrices.
\end{abstract}


\IEEEpeerreviewmaketitle

\input{./intro}

\input{./Multi-Entropy}
\input{./Multi-Source}

\input{./Channel}

\input{./conclusion}

\appendix
\input{./Appendix-Preparation}

\input{./Appendix-Multi-Entropy}

\input{./Appendix-Multi-Source}
\input{./Appendix-Channel}

\section*{Acknowledgment}

The authors would like to thank Prof.~Vincent Y.~F.~Tan for pointing out 
Remark \ref{remark:convergence-on-moderate-deviation}.
The authors are also grateful to Mr. Ryo Yaguchi for his helpful comments.
HM is partially supported by a MEXT Grant-in-Aid for
Scientific Research (A) No. 23246071. He is partially supported by the National Institute of Information and Communication Technology (NICT), Japan. 
The Centre for Quantum Technologies is funded by the Singapore Ministry of Education and the
National Research Foundation as part of the Research Centres of Excellence programme.


\bibliographystyle{../09-04-17-bibtex/IEEEtran}
\bibliography{../09-04-17-bibtex/reference.bib}


\end{document}

%% file: intro.tex
\section{Introduction}

Recently, finite-length analyses for coding problems are attracting a considerable attention 
\cite{polyanskiy:10}.
This paper focuses on finite-length analyses for the source coding with side-information for Markov sources and the channel coding for channels with conditional Markovian additive noise.
Although the main purpose of this paper is finite-length analyses,
this paper also develops a unified approach to investigate these topics including the asymptotic analyses.
Since this discussion spreads so many subtopics,
we explain them separately in the introduction. 

\subsection{Two criteria for finite-length bounds} \label{subsection:motivation}

For an explanation of the motivations of this paper,
we start with two criteria for finite-length bounds while the problems treated in this paper are not restricted to channel coding.
Until now, so many types of finite-length achievability bounds have been proposed. For example, Verd\'u and Han derived a
finite-length bound by using the information spectrum approach in order to derive the general formula \cite{verdu:94}
(see also \cite{han:book}), which we call the {\em information-spectrum bound}.
One of the authors and Nagaoka derived a bound (for the classical-quantum channel) by relating the error probability to
the binary hypothesis testing \cite[Remark 15]{hayashi:03} (see also \cite{wang:12}), which we call the {\em hypothesis testing bound}.
Polyanskiy {\em et.~al.} derived  the {\em RCU (random coding union) bound} and the {\em DT (dependence testing)
bound} \cite{polyanskiy:10}\footnote{A bound slightly looser (coefficients are worse) than the DT bound 
can be derived from the hypothesis testing bound of \cite{hayashi:03}.}. 
Also, Gallager's bound \cite{gallager:65} is known as an efficient bound to derive the exponential decreasing rate.

Here, we focus on two important criteria for finite-length bounds:
\begin{description}
\item[(C1)] Computational complexity for the bound, and
 
\item[(C2)] Asymptotic optimality for the bound.
\end{description}

First, we consider the first criterion, i.e., the computational complexity for the bound.
For the BSC, the computational complexity of the RCU bound is $O(n^2)$ and that of the DT bound is $O(n)$ \cite{polyanskiy:thesis}.
However, the computational complexities of these bounds is much larger for general DMCs or channels with memory.
It is known that the hypothesis testing bound can be described as a linear programming (eg.~see \cite{tomamichel:12, matthews:12}\footnote{In the 
case of quantum channel, the bound is described as a semi-definite programming.}), and can be efficiently computed under certain symmetry.
However, the number of variables in the linear programming grows exponentially in the block length, and it is 
difficult to compute in general.  
The computation of the information-spectrum bound depends on the evaluation of a tail probability.
The information-spectrum bound is less operational than the hypothesis testing bound 
in the sense of the hierarchy introduced in \cite{tomamichel:12}, and the computational complexity of
the former is much smaller than that of the latter. 
However the computation of a tail probability is still not so easy unless the channel is a DMC. 
For DMCs, computational complexity of Gallager's bound is $O(1)$ since the Gallager function is additive
quantity for DMCs. However, this is not the case if there is a memory\footnote{The Gallager bound for
finite states channels was considered in \cite[Section 5.9]{gallager:68}, but a closed form expression
for the exponent was not derived.}.
Consequently, there is no bound that is efficiently computable for the Markov chain so far.
The situation is the same for source coding with  side-ifnromation.

Next, let us consider the second criterion, i.e., asymptotic optimality.
So far, three kinds of asymptotic regimes have been studied in the information theory \cite{polyanskiy:10,hayashi:09,hayashi:08,altug:10,dakehe:09,tan:12b,kuzuoka:12}:
\begin{itemize}
\item The large deviation regime in which the error probability $\varepsilon$ asymptotically behaves like $e^{- n r}$ for some $r > 0$,

\item The moderate deviation regime in which $\varepsilon$ asymptotically behaves like $e^{- n^{1 - 2t} r}$ for some $r > 0$ and $t \in (0,1/2)$, and

\item The second order regime in which $\varepsilon$ is a constant.
\end{itemize}
We shall claim that a good finite-length bound should be asymptotically optimal at least one of the above mentioned three regimes.
In fact, the information spectrum bound, the hypothesis testing bound, and the DT bound are asymptotically optimal in the moderate deviation
regime and the second order regime; the Gallager bound is asymptotically optimal in the large deviation regime;
and the RCU bound is asymptotically optimal in all the regimes\footnote{The Gallager bound and the RCU bound are asymptotically optimal
in the large deviation regime only up to the critical rate.}. 
Recently, for DMC, Yang-Meng derived efficiently computable bound for
low density parity check (LDPC) codes \cite{YanMen:15}, which is asymptotically optimal in
the moderate deviation regime and the second order regime.

\subsection{Main Contribution for Finite-Length Analysis}

To derive finite-length achievability bounds on the problems,
we basically use the exponential type bounds\footnote{For channel coding, it corresponds to the Gallager bound.}.
In source coding with side-information, 
the exponential type upper bounds on error probability $\barPe(M_n)$ for a given message size 
$M_n$ are described by using conditional
R\'enyi entropies as follows (cf.~Lemma \ref{lemma:multi-source-tight-bound} and Lemma \ref{lemma:multi-source-loose-bound}):
\begin{eqnarray}
\barPe(M_n) \le 
\inf_{-\frac{1}{2} \le \theta \le 0} M_n^{\frac{\theta}{1+\theta}} e^{- \frac{\theta}{1+\theta} H_{1+\theta}^\uparrow(X^n|Y^n)}
\end{eqnarray}
and
\begin{eqnarray}
\barPe(M_n) \le \inf_{-1 \le \theta \le 0} M_n^{\theta} e^{-\theta H_{1+\theta}^\downarrow(X^n|Y^n)}.
\end{eqnarray}
Here, $H_{1+\theta}^\uparrow(X^n|Y^n)$ is the conditional R\'enyi entropy introduced by
Arimoto \cite{arimoto:75}, which we shall call {\em upper conditional R\'enyi entropy} (cf.~\eqref{eq:upper-conditional-renyi}).
On the other hand, $H_{1+\theta}^\downarrow(X^n|Y^n)$ is the conditional R\'enyi entropy introduced in \cite{hayashi:10},
which we shall call the {\em lower conditional R\'enyi entropy}. Although there are several other definitions of
conditional R\'enyi entropies, we will only use these two in this paper;
see \cite{teixeira:12, iwamoto:13} for extensive review on conditional R\'enyi entropies.

Although the above mentioned conditional R\'enyi entropies are additive for i.i.d. random variables,
they are not additive for Markov chains, which is a difficulty to derive finite-length bounds for Markov chains. 
In general, it is not easy to evaluate the conditional R\'enyi entropies for Markov chains. 
Thus, we consider two assumptions on transition matrices 
(see Assumption \ref{assumption-Y-marginal-markov} and Assumption \ref{assumption-memory-through-Y} of Section \ref{section:preparation-multi}). Without  Assumption \ref{assumption-Y-marginal-markov},
it should be noted that even the conditional entropy rate is difficult to be evaluated.
Under Assumption \ref{assumption-Y-marginal-markov}, we introduce  
the lower conditional R\'enyi entropy for transition matrices 
$H_{1+\theta}^{\downarrow,W}(X|Y)$ (cf.~\eqref{eq:definition-lower-conditional-renyi-markov}). 
Then, we evaluate
the lower conditional R\'nyi entropy for the Markov chain in terms of its transition matrix counterpart.
More specifically, we derive an approximation
\begin{eqnarray}
H_{1+\theta}^{\downarrow}(X^n|Y^n) = n H_{1+\theta}^{\downarrow,W}(X|Y) + O(1),
\end{eqnarray}
where an explicit form of $O(1)$ is also derived.
This evaluation gives finite-length bounds under Assumption \ref{assumption-Y-marginal-markov}.
Under more restrictive assumption, i.e., Assumption \ref{assumption-memory-through-Y}, we also 
introduce the upper conditional R\'enyi entropy for a transition matrix 
$H_{1+\theta}^{\uparrow,W}(X|Y)$ (cf.~\eqref{eq:definition-upper-conditional-renyi-markov}). 
Then, we evaluate 
the upper R\'enyi entropy for the Markov chain in terms of its transition matrix counterpart. 
More specifically, we derive an approximation
\begin{eqnarray}
H_{1+\theta}^{\uparrow}(X^n|Y^n) = n H_{1+\theta}^{\uparrow,W}(X|Y) + O(1),
\end{eqnarray}
where an explicit form of $O(1)$ is also derived.
This evaluation gives finite-length bounds that are tighter than those obtained under 
Assumption \ref{assumption-Y-marginal-markov}.

We also derive converse bounds by using the change of measure argument for Markov chains
developed by the authors in the accompanying paper on information geometry \cite{hayashi-watanabe:13, hayashi-watanabe:13b}.
For this purpose, we further introduce 
two-parameter conditional R\'enyi entropy and its transition matrix counterpart
(cf.~\eqref{eq:two-parameter-conditional-renyi} and \eqref{eq:definition-two-parameter-renyi-markov}).
This novel information measure includes the lower conditional R\'enyi entropy and the upper conditional
R\'enyi entropy as special cases.  
To clarify the relation among bounds based on these quantities,
we numerically calculate the upper and lower bounds for the optimal coding rate in source coding with Markovian source
as Figs. \ref{Fig:Comparison-single-source-fixed-epsilon} and \ref{Fig:Comparison-single-source-fixed-n}.
Thanks to the second criterion (C2), 
this calculation shows that our finite-length bounds are very close to the optimal value.
Although this numerical calculation contains the case with the huge size 
$n=1 \times 10^5$,
its calculation is not so difficult because their calculation complexity behaves as $O(1)$.
That is, this calculation shows the advantage of the first criterion (C1).

Here, we would like to remark on terminologies. There are a few ways to express 
exponential type bounds. In statistics or the large deviation theory, we usually use 
the cumulant generating function (CGF) to describe exponents. In information theory,
we use the Gallager function or the R\'enyi entropies. Although these three terminologies are 
essentially the same and are related by change of variables, 
the CGF and the Gallager function are convenient for some calculations since they have good properties such as convexity.
However, they are merely mathematical functions. On the other hand, the R\'enyi entropies
are information measures including Shannon's information measures as special cases. 
Thus, the R\'enyi entropies are intuitively familiar in the field of information theory.
The R\'enyi entropies also have an advantage that two types of bounds
(eg.~\eqref{eq:multi-source-ldp-converse-assumption-1} and \eqref{eq:multi-source-ldp-converse-assumption-2}) 
can be expressed in a unified manner. For these reasons, we state our main results in terms of
the R\'enyi entropies while we use the CGF and the Gallager function in the proofs. 
For readers' convenience, the relation between the R\'enyi entropies and corresponding CGFs are 
summarized in Appendices \ref{Appendix:preparation} and
and \ref{Appendix:preparetion-multi-terminal}.

\subsection{Main Contribution for Channel Coding}

It is known that there is an intimate relationship between channel coding
and source coding with side-information (eg.~\cite{wyner:74, csiszar:82, ahlswede:82}). 
In particular, for an additive channel, the error probability of channel coding by a linear code can be related to
the corresponding source coding problem with side information \cite{wyner:74}. Chen {\em et.~al.}~also showed that
the error probability of source coding with side-information by a linear encoder
can be related to the error probability of a dual channel coding problem and vice versa \cite{chen:09c} (see also \cite{hayashi:10b}). 
Since those dual channels can be regarded as 
additive channels conditioned by state-information, 
we call those channels {\em conditional additive channels}\footnote{In \cite{hayashi:10b}, we called those
channels {\em general additive channels} but we think "conditional" is more suitably describing the situation.}.
As a similar symmetric channel, a regular channel \cite{delsarte:82} is known.

In this paper, we mainly discuss a conditional additive channel, in which, the additive noise is operated subject to a distribution conditioned with an additional output information,
and propose a method to convert a regular channel into a conditional additive channel
so that our treatment covers regular channels.
Additionally, we show that the BPSK-AWGN channel is included in conditional additive channels.
Thus, by using aforementioned 
duality between channel coding and source coding with side-information, we can evaluate 
the error probability of channel coding for regular channels.

By the same reason as source coding with side-information, we assume two assumptions,
Assumption \ref{assumption-Y-marginal-markov} and Assumption \ref{assumption-memory-through-Y},
on the noise process of a conditional additive channel.
It should be noted that the Gilbert-Elliott channel \cite{gilbert:60,elliott:63} with state-information available at the receiver
can be regraded as a conditional additive channel such that the noise process is a Markov chain
satisfying both Assumption \ref{assumption-Y-marginal-markov} and Assumption \ref{assumption-memory-through-Y} (see Example \ref{example:gilbert-elliot}).
Thus, we believe that Assumption \ref{assumption-Y-marginal-markov} and Assumption \ref{assumption-memory-through-Y}
are quite reasonable assumptions.

\subsection{Asymptotic bounds and asymptotic optimality for finite-length bounds}

For asymptotic analyses of the large deviation and the moderate deviation regimes,
we derive the characterizations\footnote{For the large deviation regime, we only derive the characterizations
up to the critical rate.} by using our finite-length
achievability and converse bounds, which implies that our finite-length bounds are tight in 
the large deviation regime and the moderate deviation regime.
We also derive the second order rate. 
Although the second order rate can be derived by application of the central limit theorem to the 
information spectrum bound, the variance involves the limit with respect to the block length because of memory.
In this paper, we derive a single letter form of the variance by using the conditional R\'enyi entropy for transition 
matrices\footnote{An alternative way to derive a single letter characterization of the variance 
for the Markov chain was shown in \cite[Lemma 20]{tomamichel:13}. It should be also noted that a single letter characterization
can be derived by using the fundamental matrix \cite{kemeny-snell-book}. The single letter characterization of the variance in \cite[Section VII]{hayashi:08}
and \cite[Section III]{hayashi:09} has an error, which is corrected in this paper.}.

As we will see in Theorem \ref{theorem:source-coding-second-order}, Theorem \ref{theorem:multi-moderate-deviation},
Theorem \ref{theorem:source-coding-ldp-assumption-1}, Theorem \ref{theorem:source-coding-ldp-assumption-2},
Theorem \ref{theorem:channel-second-order}, Theorem \ref{theorem:channel-mdp}, Theorem \ref{theorem:large-deviation-source-coding-assumption-1},
and Theorem \ref{theorem:channel-ldp-2}, our asymptotic results have the same forms as  the counterparts of the i.i.d. case (cf.~\cite{gallager:65,polyanskiy:10,hayashi:09,hayashi:08,altug:10,dakehe:09})
when the information measures for distributions in the i.i.d. case are replaced by the information measures for transition matrices 
introduced in this paper.

To see the asymptotic optimality for finite-length bounds,
we summarize the relation between the asymptotic results and the finite-length bounds in Table \ref{table:summary:Asymptotic-results}.
In the table, the computational complexity of the finite-length bounds are also described.
"$\mbox{Solved}^*$" indicates that those problems are solved up to the critical rates.
"Ass.~1" and "Ass.~2" indicate that those problems are solved under 
Assumption \ref{assumption-Y-marginal-markov} or Assumption \ref{assumption-memory-through-Y}.
"$O(1)$" indicates that both the achievability part and the converse part of those asymptotic results are derived from our 
finite-length achievability bounds and converse bounds whose
computational complexities are $O(1)$. "Tail" indicates that both the achievability part and the converse part of those asymptotic results are derived from the
information-spectrum type achievability bounds and converse bounds whose computational complexities depend on the computational complexities
of tail probabilities.

Exact computations of tail probabilities are difficult in general though it may be feasible for a simple 
case such as an i.i.d.~case. One way to approximately compute tail probabilities is to use the Berry-Ess\'een
theorem \cite[Theorem 16.5.1]{feller:book} or its variant \cite{tikhomirov:80}. 
This direction of research is still continuing \cite{kontoyiannis:03,herve:12}, and an evaluation of
the constant was done in \cite{herve:12} though it is not clear how much tight it is. 
If we can derive a tight Berry-Ess\'een type bound for the Markov chain, we  can derive a finite-length bound
that is asymptotically tight in the second order regime.
However, the approximation errors
of Berry-Ess\'een type bounds converge only in the order of $1/\sqrt{n}$, and cannot be applied when $\varepsilon$ is rather small. 
Even in the cases such that exact computations of tail probabilities are possible, the information-spectrum type bounds
are looser than the exponential type bounds when $\varepsilon$ is rather small, and we need to use appropriate 
bounds depending on the size of $\varepsilon$. In fact, this observation was explicitly clarified in \cite{watanabe:13c}
for the random number generation with side-information. Consequently, we believe that our exponential type finite-length bounds are very useful.  
It should be also noted that, for source coding with side-information and channel coding for regular channels,
even the first order results have not been revealed as long as the authors know, and
they are clarified in this paper\footnote{General formulae for those problems were known \cite{verdu:94, han:book}, 
but single-letter expressions for Markov sources or channels were not clarified in the literature.}.

\begin{table}[htbp]
\begin{center}
\caption{Summary of asymptotic results and Finite-Length Bounds to Derive Asymptotic Results}
\label{table:summary:Asymptotic-results}
\begin{tabular}{|c|c|c|c|c|} \hline
 Problem & First Order & Large Deviation & Moderate Deviation & Second Order \\ \hline
 SC with SI & Solved (Ass.~1)  &$\mbox{Solved}^*$  (Ass.~2), $O(1)$ & Solved (Ass.~1), $O(1)$ & Solved (Ass.~1), Tail \\ \hline 
 CC for Conditional Additive Channels & Solved (Ass.~1) & $\mbox{Solved}^*$  (Ass.~2), $O(1)$ & Solved (Ass.~1), $O(1)$ & Solved (Ass.~1), Tail \\ \hline
\end{tabular}
\end{center}
\end{table}

\subsection{Related Works on Markov chains}

Since related works concerning the finite-length analysis has been reviewed in Section \ref{subsection:motivation},
we only review related works concerning the asymptotic analysis here.
There are some studies on Markov chains for the large deviation regime 
\cite{davisson:81, vasek:80, zhong:07b}. 
The derivation in \cite{davisson:81} uses the Markov type method. A drawback of 
this method is that it involves a term that stems from the number of types, which is not important for the asymptotic analysis but
is crucial for the finite-length analysis. Our achievability is derived by a similar approach as in \cite{vasek:80,zhong:07b},
i.e., the Perron-Frobenius theorem, but our derivation separates the single-shot part and the evaluation of the R\'enyi entropy,
and thus is more transparent. Also, the converse part of \cite{vasek:80,zhong:07b} is based 
on the Shannon-McMillan-Breiman limiting theorem 
and does not yield finite-length bounds. 
 
For the second order regime, 
Polyanskiy {\em et.~al.} studied the second order rate (dispersion) of the Gilbert-Elliott channel \cite{polyanskiy:11}. 
Tomamichel and Tan studied the second order rate of channel coding with state-information such that the 
state-information may be a general source, and derived a formula for the Markov chain as a special case \cite{tomamichel:13}.
Kontoyiannis studied the second order variable length source coding 
for the Markov chain \cite{kontoyiannis:97}. In \cite{kontoyiannis-verdu:14},
Kontoyiannis-Verd\'u derived the second order rate of lossless source coding under overflow probability criterion.

For channel coding of i.i.d. case, Scarlett {\em et. al.} derived a saddle-point approximation,
which unifies all the three regimes \cite{scarlett:13,scarlett:14}.
\subsection{Organization of Paper}

In Section \ref{section:preparation-multi}, we introduce information measures and their properties that will
be used in Section \ref{section:multi-source}
and Section \ref{section:channel}. 
Then, source coding with side-information and channel coding will be discussed in
Section \ref{section:multi-source}
and Section \ref{section:channel} respectively. 
As we mentioned above, we state our main result in terms of the
R\'enyi entropies, and we use the CGFs and the Gallager function in the proofs.
We explain how to cover the continuous case in Remarks \ref{R10-2-3} and \ref{R10-2-4}.
In Appendices \ref{Appendix:preparation}  
and \ref{Appendix:preparetion-multi-terminal}, the relation between the R\'enyi entropies and 
corresponding CGFs are summarized. The relation between the R\'enyi entropies and the Gallager function are explained 
as necessary. Proofs of some technical results are also shown in the rest of appendices.

\subsection{Notations}

For a set ${\cal X}$, the set of all distributions on ${\cal X}$ is denoted by ${\cal P}({\cal X})$.
The set of all sub-normalized non-negative functions on ${\cal X}$ is denoted by $\bar{{\cal P}}({\cal X})$.
The cumulative distribution function of the standard Gaussian random variable is denoted by
\begin{eqnarray}
\Phi(t ) = \int_{-\infty}^t \frac{1}{\sqrt{2\pi}} \exp\left[ - \frac{x^2}{2} \right] dx.
\end{eqnarray}
Throughout the paper, the base of the logarithm is $e$.

%% file: Multi-Entropy.tex
\section{Information Measures} \label{section:preparation-multi}

Since this paper discusses the second order optimality, 
we need to discuss the central limit theorem for the Markovian process.
For this purpose,
we usually employ advanced mathematical methods from probability theory.
For example, the paper \cite[Theorem 4]{Ben-Ari}
showed the Markov version of the central limit theorem by using a martingale
stopping technique.
Lalley \cite{Lalley} employed regular perturbation theory of operators on the infinite dimensional space \cite[Ch. 7, \#1, Ch. 4, \#3, and Ch. 3, \#5]{Kato}. 
The papers \cite{R5,R6}\cite[Lemma 1.5 of Chapter 1]{KLO}
employed the spectral measure
while it is hard to calculate the spectral measure in general even in the finite state case.
Further, the papers \cite{R5,CLT2,CLT3,CLT4} showed the central limit theorem
by using the asymptotic variance, 
but they did not give any computable expression of the asymptotic variance
without the infinite sum.
In summary, to derive the central limit theorem with the variance of computable form, these papers need to use very advanced mathematics beyond calculus and linear algebra.

To overcome this problem, we employ the method used in our recent paper \cite{hayashi-watanabe:13b}.
The paper \cite{hayashi-watanabe:13b} employed the method based on the cumulant generating function for transition matrices, 
which the Perron eigenvalue of a specific non-negative-entry matrix.
Since a Perron eigenvalue can be explained in the framework of linear algebra,
the method can be described with elementary mathematics.
To employ this method, we need to define the information measure 
in a way similar to the cumulant generating function for transition matrices.
That is, we define the information measures for transition matrices, 
e.g., the conditional R\'{e}nyi entropy for transition matrices, etc,
by using Perron eigenvalues. 

Fortunately, these information measures for transition matrices
are very useful even for large deviation type evaluation and finite-length bounds.
For example, our recent paper \cite{hayashi-watanabe:13b} derived finite-length bounds
for simple hypothesis testing for Markovian chain by using the cumulant generating function for transition matrices.
Therefore, using these information measures for transition matrices, 
this paper derives finite-length bounds for source coding and channel coding with Markov chains, and discusses their asymptotic bounds with large deviation, moderate deviation, and second order type.

Since they are natural extensions of information measures for single-shot setting,
we first review information measures for single-shot setting in Section \ref{section:multi-terminal-single-shot}.
Next, we introduce information measures for transition matrices in Section \ref{subsection:multi-terminal-measures-markov}.
Then, we show that information measures for Markov chains can be approximated by information measures for transition matrices generating 
those Markov chains in Section \ref{subsection-multi-terminal-information-measures-markov}.

\subsection{Information measures for Single-Shot Setting} \label{section:multi-terminal-single-shot}

In this section, we introduce conditional R\'enyi entropies for the single-shot setting. 
For more detailed review of conditional R\'enyi entropies, see \cite{iwamoto:13}.
For a correlated random variable $(X,Y)$ on ${\cal X} \times {\cal Y}$ with probability distribution $P_{XY}$ and a marginal distribution
$Q_Y$ on ${\cal Y}$, we introduce the conditional R\'enyi entropy of order $1+\theta$ relative to $Q_Y$ as
\begin{eqnarray}
H_{1+\theta}(P_{XY}|Q_Y) := - \frac{1}{\theta} \log \sum_{x,y} P_{XY}(x,y)^{1+\theta} Q_Y(y)^{-\theta},
\end{eqnarray}
where $\theta \in (-1,0) \cup (0,\infty)$. The conditional R\'enyi entropy of order $0$ relative to $Q_Y$ is defined by the limit with
respect to $\theta$.
When ${\cal  Y}$ is singleton, it is nothing but the ordinary R\'enyi entropy, and it is denoted by 
$H_{1+\theta}(X) = H_{1+\theta}(P_X)$ throughout the paper. 

One of important special cases of $H_{1+\theta}(P_{XY}|Q_Y)$ is the case with $Q_Y = P_Y$,
where $P_Y$ is the marginal of $P_{XY}$.
We shall call this special case the {\em lower conditional R\'enyi entropy} of order $1+\theta$ and denote\footnote{
This notation was first introduce in \cite{tomamichel:13b}.}
\begin{eqnarray} \label{eq:lower-conditional-renyi}
H_{1+\theta}^\downarrow(X|Y) &:=& H_{1+\theta}(P_{XY}|P_Y) \\
&=& - \frac{1}{\theta} \log \sum_{x,y} P_{XY}(x,y)^{1+\theta} P_Y(y)^{-\theta}.
\end{eqnarray}
We have the following property, which follows from the correspondence
between the conditional R\'enyi entropy and the cumulant generating function (cf.~Appendix \ref{Appendix:preparetion-multi-terminal}).
\begin{lemma} \label{lemma:property-down-conditional-renyi}
We have
\begin{eqnarray} \label{eq:down-conditional-renyi-theta-0}
\lim_{\theta \to 0} H_{1+\theta}^\downarrow(X|Y) = H(X|Y)
\end{eqnarray}
and
\begin{eqnarray}
\san{V}(X|Y) &:=& \mathrm{Var}\left[ \log \frac{1}{P_{X|Y}(X|Y)} \right] \\
&=& \lim_{\theta \to 0} \frac{2\left[ H(X|Y) - H_{1+\theta}^\downarrow(X|Y) \right]}{\theta} \label{eq:multi-single-shot-variance-1}.
\end{eqnarray}
\end{lemma}
\begin{proof}
\eqref{eq:down-conditional-renyi-theta-0} follows from the relation in \eqref{eq:relation-single-shot-downarrow-phi} and
the fact that the first-order derivative of cumulant generating function is the expectation. \eqref{eq:multi-single-shot-variance-1}
follows from \eqref{eq:relation-single-shot-downarrow-phi}, \eqref{eq:down-conditional-renyi-theta-0},
and \eqref{eq:variance-second-derivative-cgf}.
\end{proof}

The other important special cases of $H_{1+\theta}(P_{XY}|Q_Y)$ is the measure maximized over $Q_Y$.
We shall call this special case the {\em upper conditional R\'enyi entropy} of order $1+\theta$ and 
denote\footnote{For $-1 < \theta < 0$, \eqref{eq:optimal-choice-upper-conditional} can be proved by using the H\"older inequality,
and, for $0 < \theta$,  \eqref{eq:optimal-choice-upper-conditional} can be proved by using the reverse H\"older 
inequality \cite[Lemma 8]{hayashi:12d}.}
\begin{eqnarray} \label{eq:upper-conditional-renyi}
H_{1+\theta}^\uparrow(X|Y) &:=& \max_{Q_Y \in {\cal P}({\cal Y})} H_{1+\theta}(P_{XY}|Q_Y) \\
&=& H_{1+\theta}(P_{XY}|P_Y^{(1+\theta)}) \label{eq:optimal-choice-upper-conditional} \\
&=& - \frac{1+\theta}{\theta} \log \sum_y P_Y(y) \left[ \sum_x P_{X|Y}(x|y)^{1+\theta} \right]^{\frac{1}{1+\theta}},
\end{eqnarray}
where 
\begin{eqnarray} \label{eq:single-shot-optimal-conditioning-distribution}
P_Y^{(1+\theta)}(y) := 
\frac{\left[ \sum_x P_{XY}(x,y)^{1+\theta} \right]^{\frac{1}{1+\theta}}}{\sum_{y^\prime} \left[ \sum_x P_{XY}(x,y^\prime)^{1+\theta} \right]^{\frac{1}{1+\theta}}}. 
\end{eqnarray}
For this measure, we also have properties similar to Lemma \ref{lemma:property-down-conditional-renyi}.
This lemma will be proved in Appendix \ref{Appendix:lemma:property-upper-conditional-renyi-single-shot}.
\begin{lemma} \label{lemma:property-upper-conditional-renyi-single-shot}
We have
\begin{eqnarray} \label{eq:up-conditional-renyi-theta-0}
\lim_{\theta \to 0} H_{1+\theta}^\uparrow(X|Y) = H(X|Y)
\end{eqnarray}
and 
\begin{eqnarray}
\lim_{\theta \to 0} \frac{2\left[ H(X|Y) - H_{1+\theta}^\uparrow(X|Y) \right]}{\theta} 
= \san{V}(X|Y).
\label{eq:multi-single-shot-variance-2}
\end{eqnarray}
\end{lemma}

When we derive converse bounds, we need to consider the case such that 
the order of the R\'enyi entropy and the order of conditioning distribution defined in \eqref{eq:single-shot-optimal-conditioning-distribution} are different.
For this purpose, we introduce two-parameter conditional R\'enyi entropy:
\begin{eqnarray} \label{eq:two-parameter-conditional-renyi}
\lefteqn{ H_{1+\theta,1+\theta^\prime}(X|Y) } \\
&:=& H_{1+\theta}(P_{XY}|P_Y^{(1+\theta^\prime)}) \\
&=& - \frac{1}{\theta} \log \sum_y P_Y(y) \left[ \sum_x P_{X|Y}(x|y)^{1+\theta} \right] \left[\sum_x P_{X|Y}(x|y)^{1+\theta^\prime} \right]^{\frac{-\theta}{1+\theta^\prime}}
 + \frac{\theta^\prime}{1+\theta^\prime} H_{1+\theta^\prime}^\uparrow(X|Y).
\end{eqnarray}

Next, we investigate some properties of the measures defined above, 
which will be proved in Appendix \ref{Appendix:lemma:multi-terminal-single-shot-property}.
\begin{lemma} \label{lemma:multi-terminal-single-shot-property}
$\phantom{aaa}$ 
\begin{enumerate}
\item \label{item:multi-terminal-single-shot-property-1}
For fixed $Q_Y$, $\theta H_{1+\theta}(P_{XY}|Q_Y)$ is a concave function of $\theta$,
and it is strict concave iff. $\rom{Var}\left[ \log \frac{Q_Y(Y)}{P_{XY}(X,Y)} \right] > 0$.

\item \label{item:multi-terminal-single-shot-property-1-b}
For fixed $Q_Y$, $H_{1+\theta}(P_{XY}|Q_Y)$ is a monotonically decreasing\footnote{Technically, $H_{1+\theta}(P_{XY}|Q_Y)$ is always non-increasing and it is monotonically decreasing iff. strict concavity holds in Statement \ref{item:multi-terminal-single-shot-property-1}. Similar remarks are also applied for other information measures throughout the paper.} function of $\theta$.

\item \label{item:multi-terminal-single-shot-property-2}
The function $\theta H_{1+\theta}^\downarrow(X|Y)$ is a concave function of $\theta$, and it 
is strict concave iff. $\san{V}(X|Y) > 0$.

\item \label{item:multi-terminal-single-shot-property-2-b}
$H_{1+\theta}^\downarrow(X|Y)$ is a monotonically decreasing function of $\theta$.

\item \label{item:multi-terminal-single-shot-property-3}
The function $\theta H_{1+\theta}^\uparrow(X|Y)$ is a concave function of $\theta$, and it is strict concave 
iff. $\san{V}(X|Y) > 0$.

\item \label{item:multi-terminal-single-shot-property-3-b}
$H_{1+\theta}^\uparrow(X|Y)$ is a monotonically decreasing function of $\theta$.

\item \label{item:multi-terminal-single-shot-property-4}
For every $\theta \in (-1,0) \cup (0,\infty)$, we have $H_{1+\theta}^\downarrow(X|Y) \le H_{1+\theta}^\uparrow(X|Y)$.

\item \label{item:multi-terminal-single-shot-property-5}
For fixed $\theta^\prime$, the function $\theta H_{1+\theta,1+\theta^\prime}(X|Y)$ is a concave function of $\theta$,
and it is strict concave iff. $\san{V}(X|Y) > 0$.

\item \label{item:multi-terminal-single-shot-property-5-b}
For fixed $\theta^\prime$, $H_{1+\theta,1+\theta^\prime}(X|Y)$ is a monotonically decreasing function of $\theta$.

\item \label{item:multi-terminal-single-shot-property-6}
We have
\begin{eqnarray}
H_{1+\theta,1}(X|Y) = H_{1+\theta}^\downarrow(X|Y).
\end{eqnarray}

\item \label{item:multi-terminal-single-shot-property-7}
We have
\begin{eqnarray}
H_{1+\theta,1+\theta}(X|Y) = H_{1+\theta}^\uparrow(X|Y).
\end{eqnarray}

\item \label{item:multi-terminal-single-shot-property-8}
For every $\theta \in (-1,0) \cup (0,\infty)$, $H_{1+\theta,1+\theta^\prime}(X|Y)$ is maximized at $\theta^\prime = \theta$.


\end{enumerate}
\end{lemma}


We can also derive explicit forms of the conditional R\'enyi entropies of order $0$.
\begin{lemma} \label{lemma:limit-renyi-max-entropy}
We have
\begin{eqnarray}
\lim_{\theta \to -1} H_{1+\theta}(P_{XY}|Q_Y) &=& H_0(P_{XY}|Q_Y) \\
 &:=& \log \sum_y Q_Y(y) |\rom{supp}(P_{X|Y}(\cdot|y))|, \label{eq:limit-renyi-max-entropy-0} \\
\lim_{\theta \to -1} H_{1+\theta}^\uparrow(X|Y) &=&H_0^\uparrow(X|Y) \\
 &:=& \log \max_{y \in \rom{supp}(P_Y)} |\rom{supp}(P_{X|Y}(\cdot|y))|, 
  \label{eq:limit-renyi-max-entropy-1} \\
\lim_{\theta \to -1} H_{1+\theta}^\downarrow(X|Y) &=& H_0^\downarrow(X|Y) \\
&:=& \log \sum_y P_Y(y) |\rom{supp}(P_{X|Y}(\cdot|y))|. 
  \label{eq:limit-renyi-max-entropy-2}
\end{eqnarray}
\end{lemma}
\begin{proof}
See Appendix \ref{appendix:lemma:limit-renyi-max-entropy}.
\end{proof}

From Statement \ref{item:multi-terminal-single-shot-property-1} of Lemma \ref{lemma:multi-terminal-single-shot-property},
$\frac{d[\theta H_{1+\theta}(P_{XY}|Q_Y)]}{d\theta}$ is monotonically decreasing.
Thus, we can define the inverse function\footnote{Throughout the paper, the notations $\theta(a)$ and
$a(R )$ are reused for several inverse functions. Although the meanings of those notations are
obvious from the context, we occasionally put superscript $Q$, $\downarrow$ or $\uparrow$ to
emphasize that those inverse functions are induced from corresponding conditional R\'enyi entropies.
This definition is related to Legendre transform of the concave function $\theta \mapsto \theta H_{1+\theta}^\downarrow(X|Y)$.}  
$\theta(a) = \theta^Q(a)$ of $\frac{d[\theta H_{1+\theta}(P_{XY}|Q_Y)]}{d\theta}$ by
\begin{eqnarray} \label{eq:definition-inverse-theta-multi-one-shot}
\frac{d[\theta H_{1+\theta}(P_{XY}|Q_Y)]}{d\theta} \bigg|_{\theta = \theta(a)} = a
\end{eqnarray}
for $\underline{a} < a \le \overline{a}$, where
$\underline{a} = \underline{a}^Q := \lim_{\theta\to \infty} \frac{d[\theta H_{1+\theta}(P_{XY}|Q_Y)]}{d\theta}$
and $\overline{a} = \overline{a}^Q := \lim_{\theta\to -1} \frac{d[\theta H_{1+\theta}(P_{XY}|Q_Y)]}{d\theta}$. 
Let 
\begin{eqnarray}
R(a) = R^Q(a) := (1+\theta(a)) a - \theta(a) H_{1+\theta(a)}(P_{XY}|Q_Y).
\end{eqnarray}
Since
\begin{eqnarray}
R^\prime(a) = (1+\theta(a)),
\end{eqnarray}
$R(a)$ is a monotonic increasing function of $\underline{a} < a \le R(\overline{a})$.
Thus, we can define the inverse function $a(R ) = a^Q(R )$ of $R(a)$ by
\begin{eqnarray} \label{eq:definition-inverse-a-multi-one-shot}
(1+\theta(a(R ))) a(R ) - \theta(a(R )) H_{1+\theta(a(R ))}(P_{XY}|Q_Y) = R
\end{eqnarray}
for $R(\underline{a}) < R \le H_0(P_{XY}|Q_Y)$.

For $\theta H_{1+\theta}^\downarrow(X|Y)$,
by the same reason as above,
we can define the inverse functions
$\theta(a) = \theta^\downarrow(a)$ and $a(R ) = a^\downarrow(R )$ by
\begin{eqnarray} \label{eq:definition-inverse-theta-multi-one-shot-2}
\frac{d[\theta H_{1+\theta}^\downarrow(X|Y)]}{d\theta} \bigg|_{\theta = \theta(a)} = a
\end{eqnarray}
and 
\begin{eqnarray}
(1+\theta(a(R ))) a(R ) - \theta(a(R )) H_{1+\theta(a(R ))}^\downarrow(X|Y) = R,
\end{eqnarray}
for $R(\underline{a}) < R \le H_0^\downarrow(X|Y)$.
For $\theta H_{1+\theta}^\uparrow(X|Y)$,
we also introduce the inverse functions $\theta(a) = \theta^\uparrow(a)$ and $a(R ) = a^\uparrow(R )$ by
\begin{eqnarray} \label{eq:definition-rho-inverse-Gallager-one-shot}
\frac{d\theta H_{1+\theta}^\uparrow(X|Y)}{d\theta} \bigg|_{\theta = \theta(a)} = a
\end{eqnarray}
and 
\begin{eqnarray} \label{eq:definition-a-inverse-Gallager-one-shot}
(1+\theta(a(R ))) a(R ) - \theta(a(R )) H_{1+\theta(a(R ))}^\uparrow(X|Y) = R
\end{eqnarray}
for $R(\underline{a}) < R \le H_0^\uparrow(X|Y)$.

\begin{remark}\label{R10-2-3}
Here, we discuss the possibility for extension to the continuous case.
Since the entropy on the continuous diverges,
we cannot extend the information quantities to the case when ${\cal X}$ is continuous.
However, it is possible to extend these quantities
to the case when ${\cal Y}$ is continuous but ${\cal X}$ is a discrete finite set.
In this case, 
we prepare a general measure $\mu$ (like the Lebesgue measure) on ${\cal Y}$
and probability density function $p_Y$ and $q_Y$ such that 
the distributions $P_Y$ and $Q_Y$ are given as
$ p_Y(y) \mu(d y)$ and $ q_Y(y) \mu(d y)$, respectively.
Then, it is sufficient to replace $\sum$, $Q(y)$, and $P_{XY}(x,y)$
by $\int_{{\cal Y}} \mu(dy)$, $P_{X|Y}(x|y) p_Y(y)$,
and $q_Y(y)$, respectively.
Hence, in the $n$-independent and identical distributed case,
these information measures are given as $n$ times of the original information measures.

One might consider the information quantities for transition matrices given in the next subsection to this continuous case.
However, it is not so easy because it needs a continuous extension of the Perron eigenvalue.
\end{remark}

\subsection{Information Measures for Transition Matrix}  \label{subsection:multi-terminal-measures-markov}

Let $\{ W(x,y|x^\prime,y^\prime) \}_{((x,y),(x^\prime,y^\prime)) \in ({\cal X} \times {\cal Y})^2}$ be 
an ergodic and irrecucible transition matrix.
The purpose of this section is to introduce transition matrix
 counter parts of those measures in Section \ref{section:multi-terminal-single-shot}.
For this purpose, we first need to introduce some assumptions on transition matrices:
\begin{assumption}[Non-Hidden] \label{assumption-Y-marginal-markov}
We say that a transition matrix $W$ is {\em non-hidden} (with respect to ${\cal Y}$) if\footnote{The reason for the name ``non-hidden" is
the following. In general, the random variable $Y$ is subject to a hidden Markov process.
However, when the condition \eqref{eq:condition-non-hidden} holds, the random variable $Y$ is subject to a Markov process. Hence, we call the 
condition \eqref{eq:condition-non-hidden} non-hidden.} 
\begin{eqnarray} \label{eq:condition-non-hidden}
\sum_x W(x,y|x^\prime,y^\prime) = W(y|y^\prime)
\end{eqnarray}
for every $x^\prime \in {\cal X}$ and $y,y^\prime \in {\cal Y}$.
This condition is equivalent to
the existence of the following decomposition of 
$W(x,y|x^\prime,y^\prime)$;
\begin{eqnarray}
W(x,y|x^\prime,y^\prime) = W(y|y^\prime) W(x|x^\prime,y^\prime,y).
\label{10-2-1}
\end{eqnarray}
\end{assumption}
\begin{assumption}[Strongly Non-Hidden] \label{assumption-memory-through-Y}
We say that a transition matrix $W$ is {\em strongly non-hidden} (with respect to ${\cal Y}$)
if, for every $\theta \in (-1,\infty)$ and $y,y^\prime \in {\cal Y}$,
\begin{eqnarray} \label{eq:condition-assumption-2}
W_\theta(y|y^\prime) := \sum_x W(x,y|x^\prime,y^\prime)^{1+\theta} 
\end{eqnarray}
is well defined, i.e., the right hand side of \eqref{eq:condition-assumption-2} is
independent of $x^\prime$.
\end{assumption}
Assumption \ref{assumption-Y-marginal-markov} requires \eqref{eq:condition-assumption-2} to hold only
for $\theta = 0$, and thus Assumption \ref{assumption-memory-through-Y} implies Assumption \ref{assumption-Y-marginal-markov}.
However, Assumption \ref{assumption-memory-through-Y} is strictly stronger condition than 
Assumption \ref{assumption-Y-marginal-markov}. For example, let consider the case such that the transition matrix is a product form,
i.e., $W(x,y|x^\prime,y^\prime) = W(x|x^\prime) W(y|y^\prime)$. In this case, Assumption \ref{assumption-Y-marginal-markov} is
 obviously satisfied. However, Assumption \ref{assumption-memory-through-Y} is not satisfied in general.

\begin{remark}\label{R10-19}
Assumption \ref{assumption-memory-through-Y} has another expression as follows.
Assumption \ref{assumption-memory-through-Y} holds if and only if, for every $x^\prime \neq \tilde{x}^\prime$, there exists a permeation $\pi_{x^\prime ; \tilde{x}^\prime}$ on ${\cal X}$ such that
$W(x|x^\prime,y^\prime,y) = W(\pi_{x^\prime ; \tilde{x}^\prime}(x) | \tilde{x}^\prime,y^\prime,y)$.

Now, we fix an element $x_0 \in {\cal X}$,
and transform a sequence of random numbers $(X_1,Y_1,X_2,Y_2,\ldots, X_n,Y_n)$
to the sequence of random numbers $
(X_1',Y_1',X_2',Y_2',\ldots, X_n',Y_n')
:=(X_1,Y_1,\pi_{x_0 ; X_1}^{-1}(X_2),Y_2,\ldots, \pi_{x_0 ; X_1}^{-1}(X_n),Y_n)$.
Then, letting $W'(x|y^\prime,y) := W(x|x_0,y^\prime,y) $, we
have
$P_{X_{i}',Y_i'|X_{i-1}',Y_{i-1}'}=
W'(y_i'|y_{i-1}') W(x_i'|y_i',y_{i-1}')$.
That is, essentially, the transition matrix of this case can be written by 
the transition matrix $W(y_i'|y_{i-1}') W'(x_i'|y_i',y_{i-1}')$.
So, the transition matrix can be written by using the positive-entry matrix $W_{x_i'}(y_i'|y_{i-1}') := W(y_i'|y_{i-1}') W'(x_i'|y_i',y_{i-1}')$.

Since the part ``if'' is trivial, we show the part ``only if'' as follow.
By noting \eqref{10-2-1}, Assumption \ref{assumption-memory-through-Y} can be rephrased as 
\begin{eqnarray} \label{eq:equivalent-formulation-assumption-2}
\sum_x W(x|x^\prime,y^\prime,y)^{1+\theta}
\end{eqnarray}
does not depend on $x^\prime$ for every $\theta \in (-1,\infty)$. Furthermore, this condition can be rephrased as follows.
For $x^\prime \neq \tilde{x}^\prime$, if the largest values of $\{ W(x|x^\prime,y^\prime) \}_{x \in {\cal X}}$ and $\{ W(x|\tilde{x}^\prime,y^\prime) \}_{x \in {\cal X}}$
are different, say the former is larger, then $\sum_x W(x|x^\prime,y^\prime)^{1+\theta} > \sum_x W(x|\tilde{x}^\prime,y^\prime)^{1+\theta}$ for sufficiently large $\theta$,
which contradict the fact that \eqref{eq:equivalent-formulation-assumption-2} does not depend on $x^\prime$.
Thus, the largest values of $\{ W(x|x^\prime,y^\prime) \}_{x \in {\cal X}}$ and $\{ W(x|\tilde{x}^\prime,y^\prime) \}_{x \in {\cal X}}$ must coincide. 
By repeating this argument for the second largest value of $\{ W(x|x^\prime,y^\prime) \}_{x \in {\cal X}}$ and $\{ W(x|\tilde{x}^\prime,y^\prime) \}_{x \in {\cal X}}$
and so on, we find 
Assumption \ref{assumption-memory-through-Y} implies that for every $x^\prime \neq \tilde{x}^\prime$, there exists a permeation $\pi_{x^\prime ; \tilde{x}^\prime}$ on ${\cal X}$ such that
$W(x|x^\prime,y^\prime,y) = W(\pi_{x^\prime ; \tilde{x}^\prime}(x) | \tilde{x}^\prime,y^\prime,y)$.
\end{remark}

The followings are non-trivial examples satisfying Assumption \ref{assumption-Y-marginal-markov}
and Assumption \ref{assumption-memory-through-Y}.
\begin{example} \label{example:markov-Y-plus-additive-noise}
Suppose that ${\cal X} = {\cal Y}$ are a module. Let $P$ and $Q$ be transition matrices on ${\cal X}$. Then, the transition matrix 
given by
\begin{eqnarray} \label{eq:example-markov-plus-additive-noise}
W(x,y|x^\prime,y^\prime) = Q(y|y^\prime) P(x-y| x^\prime - y^\prime)
\end{eqnarray}
satisfies Assumption \ref{assumption-Y-marginal-markov}. Furthermore, if transition matrix $P(z|z^\prime)$
can be written as 
\begin{eqnarray} \label{eq:symmetric-additive-noise}
P(z|z^\prime) = P_Z(\pi_{z^\prime}(z)) 
\end{eqnarray}
for permutation $\pi_{z^\prime}$ and a distribution $P_Z$ on ${\cal X}$, then
transition matrix $W$ defined by \eqref{eq:example-markov-plus-additive-noise} satisfies 
 Assumption \ref{assumption-memory-through-Y} as well.
\end{example}
\begin{example}
Suppose that ${\cal X}$ is a module, and $W$ is (strongly) non-hidden with respect to ${\cal Y}$. Let $Q$ be a transition
matrix on ${\cal Z} = {\cal X}$. Then, the transition matrix given by
\begin{eqnarray}
V(x, y,z|x^\prime,y^\prime,z^\prime) = W(x-z, y| x^\prime - z^\prime,y) Q(z|z^\prime)
\end{eqnarray}
is (strongly) non-hidden with respect to ${\cal Y} \times {\cal Z}$.
\end{example}
The following is also an example satisfying Assumption \ref{assumption-memory-through-Y},
 which describes a noise process of an important class of channels with memory (cf.~Example \ref{example:gilbert-elliot}).
\begin{example} \label{example:gilbert-elliot-noise}
Let ${\cal X} = {\cal Y} = \{0,1\}$. Then, let 
\begin{eqnarray}
W(y |y^\prime) = \left\{
\begin{array}{ll}
1 - q_{y^\prime} & \mbox{if } y = y^\prime \\
q_{y^\prime} & \mbox{if } y \neq y^\prime
\end{array}
\right.
\end{eqnarray}
for some $0 < q_0,q_1 < 1$, and let
\begin{eqnarray}
W(x|x^\prime,y^\prime,y) = \left\{
\begin{array}{ll}
1 - p_y & \mbox{if } x = 0 \\
p_y & \mbox{if } x = 1
\end{array}
\right.
\end{eqnarray}
for some $0 < p_0,p_1 < 1$. 
By choosing $\pi_{x';\tilde{x}'}$ to be the identity, 
this transition matrix satisfies 
the condition given in Remark \ref{R10-19}, that is equivalent to 
Assumption \ref{assumption-memory-through-Y}.
\end{example}


First, we introduce information measures under Assumption \ref{assumption-Y-marginal-markov}. 
In order to define a transition matrix counterpart of \eqref{eq:lower-conditional-renyi}, let us introduce the following tilted matrix:
\begin{eqnarray}
\tilde{W}_\theta(x,y|x^\prime,y^\prime) := W(x,y|x^\prime,y^\prime)^{1+\theta} W(y|y^\prime)^{-\theta}.
\end{eqnarray}
Here, we should notice that the tilted matrix $\tilde{W}_\theta$ is not normalized, i.e., is not a transition matrix.
Let $\lambda_\theta$ be the Perron-Frobenius eigenvalue of $\tilde{W}_\theta$
and $\tilde{P}_{\theta,XY}$ be its normalized eigenvector. 
Then, we define the lower conditional R\'enyi entropy for $W$ by
\begin{eqnarray} \label{eq:definition-lower-conditional-renyi-markov}
H_{1+\theta}^{\downarrow,W}(X|Y) := - \frac{1}{\theta} \log \lambda_\theta,
\end{eqnarray}
where $\theta \in (-1,0) \cup (0,\infty)$. For $\theta = 0$, we define the lower conditional R\'enyi entropy for $W$ by 
\begin{eqnarray}
H^W(X|Y) &=& H_1^{\downarrow,W}(X|Y) \\
&:=& \lim_{\theta \to 0} H_{1+\theta}^{\downarrow,W}(X|Y), \label{eq:lower-conditional-renyi-markov-theta-0}
\end{eqnarray}
and we just call it the conditional entropy for $W$. In fact, the definition of $H^W(X|Y)$ above coincide with
\begin{eqnarray}
- \sum_{x^\prime, y^\prime} P_{0,XY}(x^\prime,y^\prime) \sum_{x,y} W(x,y|x^\prime,y^\prime) \log \frac{W(x,y|x^\prime,y^\prime)}{W(y|y^\prime)},
\end{eqnarray}
where $P_{0,XY}$ is the stationary distribution of $W$ (cf.~\cite[Eq.~(30)]{HayWat16}).
For $\theta=-1$, $H_0^{\downarrow,W}(X|Y)$ is also defined by taking the limit.
When ${\cal Y}$ is singleton, the R\'enyi entropy $H_{1+\theta}^W(X)$ for $W$ is defined
as a special case of $H_{1+\theta}^{\downarrow,W}(X|Y)$.

As a counterpart of \eqref{eq:multi-single-shot-variance-1}, we also define\footnote{Since the limiting expression 
in \eqref{eq:lower-conditional-renyi-markov-theta-0-derivative}
coincides with the second derivative of the CGF (cf.~\eqref{eq:relation-variance-cgf-transition-matrix}), 
and since the second derivative of the CGF exists (cf.~\cite[Appendix D]{hayashi-watanabe:13}), 
the variance in \eqref{eq:lower-conditional-renyi-markov-theta-0-derivative} is well defined.} 
\begin{eqnarray}
\san{V}^W(X|Y) := 
\lim_{\theta \to 0} \frac{2\left[ H^W(X|Y) - H_{1+\theta}^{\downarrow,W}(X|Y) \right]}{\theta}.
\label{eq:lower-conditional-renyi-markov-theta-0-derivative}
\end{eqnarray}

\begin{remark}
When transition matrix $W$ satisfies Assumption \ref{assumption-memory-through-Y}, $H_{1+\theta}^{\downarrow,W}(X|Y)$
can be written as 
\begin{eqnarray}
H_{1+\theta}^{\downarrow,W}(X|Y) = - \frac{1}{\theta} \log \lambda_\theta^\prime,
\end{eqnarray}
where $\lambda_\theta^\prime$ is the Perron-Frobenius eigenvalue of
$W_\theta(y|y^\prime) W(y|y^\prime)^{-\theta}$.
In fact, for the left Perro-Frobenius eigenvector $\hat{Q}_\theta$ of $W_\theta(y|y^\prime) W(y|y^\prime)^{-\theta}$, we have
\begin{eqnarray}
\sum_{x,y} \hat{Q}_\theta(y) W(x,y|x^\prime,y^\prime)^{1+\theta} W(y|y^\prime)^{-\theta}
= \lambda_\theta^\prime Q_\theta(y^\prime),
\end{eqnarray}
which implies that $\lambda_\theta^\prime$ is the Perron-Frobenius eigenvalue of $\tilde{W}_\theta$.
Consequently, we can evaluate $H_{1+\theta}^{\downarrow,W}(X|Y)$ by calculating the Perron-Frobenius 
eigenvalue of $|{\cal Y}| \times |{\cal Y}|$ matrix instead of $|{\cal X}| |{\cal Y}| \times |{\cal X}| |{\cal Y}|$ matrix
when $W$ satisfies Assumption \ref{assumption-memory-through-Y}.
\end{remark}

Next, we introduce information measures under Assumption \ref{assumption-memory-through-Y}.
In order to define a transition matrix counterpart of \eqref{eq:upper-conditional-renyi}, let us introduce the 
following $|{\cal Y}| \times |{\cal Y}|$ matrix:
\begin{eqnarray}
K_\theta(y|y^\prime) := W_\theta(y|y^\prime)^{\frac{1}{1+\theta}},
\end{eqnarray}
where $W_\theta$ is defined by \eqref{eq:condition-assumption-2}.
Let $\kappa_\theta$ be the Perron-Frobenius eigenvalue of $K_\theta$.
Then, we define the upper conditional R\'enyi entropy for $W$ by
\begin{eqnarray} \label{eq:definition-upper-conditional-renyi-markov}
H_{1+\theta}^{\uparrow,W}(X|Y) := - \frac{1+\theta}{\theta} \log \kappa_\theta,
\end{eqnarray}
where $\theta \in (-1,0)\cup (0,\infty)$. 
For $\theta = -1$ and $\theta = 0$, $H_{1+\theta}^{\uparrow,W}(X|Y)$ is defined by taking the limit.
We have the following properties, which will be proved in 
Appendix \ref{appendix:lemma:properties-upper-conditional-renyi-transition-matrix}.
\begin{lemma} \label{lemma:properties-upper-conditional-renyi-transition-matrix}
We have
\begin{eqnarray}
\lim_{\theta \to 0} H_{1+\theta}^{\uparrow,W}(X|Y) = H^W(X|Y)
\end{eqnarray}
and
\begin{eqnarray} \label{eq:upper-conditional-renyi-transition-variance}
\lim_{\theta \to 0} \frac{2\left[ H^W(X|Y) - H_{1+\theta}^{\uparrow,W}(X|Y) \right]}{\theta}
= \san{V}^W(X|Y).
\end{eqnarray}
\end{lemma}

Now, let us introduce a transition matrix counterpart of \eqref{eq:two-parameter-conditional-renyi}. For this purpose, we introduce 
the following $|{\cal Y}| \times |{\cal Y}|$ matrix:
\begin{eqnarray}
N_{\theta,\theta^\prime}(y|y^\prime) := W_\theta(y|y^\prime)
W_{\theta^\prime}(y|y^\prime)^{\frac{-\theta}{1+\theta^\prime}}.
\end{eqnarray}
Let $\nu_{\theta,\theta^\prime}$ be the Perron-Frobenius eigenvalue of $N_{\theta,\theta^\prime}$.
Then, we define the two-parameter conditional R\'enyi entropy by
\begin{eqnarray} \label{eq:definition-two-parameter-renyi-markov}
H_{1+\theta,1+\theta^\prime}^W(X|Y) 
 := -\frac{1}{\theta} \log \nu_{\theta,\theta^\prime} 
 + \frac{\theta^\prime}{1+\theta^\prime} H_{1+\theta^\prime}^{\uparrow,W}(X|Y).
\end{eqnarray}

\begin{remark}
Although we defined $H_{1+\theta}^{\downarrow,W}(X|Y)$ and $H_{1+\theta}^{\uparrow,W}(X|Y)$
by \eqref{eq:definition-lower-conditional-renyi-markov} and \eqref{eq:definition-upper-conditional-renyi-markov} respectively,
we can alternatively define these measures in the same spirit as the single-shot setting 
by introducing a transition matrix counterpart of $H_{1+\theta}(P_{XY}|Q_Y)$ as follows.
For the marginal $W(y|y^\prime)$ of $W(x,y|x^\prime,y^\prime)$,
let ${\cal Y}^2_W := \{(y,y^\prime) : W(y|y^\prime) > 0\}$. For another transition matrix $V$ on ${\cal Y}$,
we define ${\cal Y}_V^2$ in a similar manner. For $V$ satisfying ${\cal Y}_W^2 \subset {\cal Y}_V^2$, 
we define\footnote{Although we can also define $H_{1+\theta}^{W|V}(X|Y)$ even if ${\cal Y}_W^2 \subset {\cal Y}_V^2$
is not satisfied (see \cite{hayashi-watanabe:13} for the detail), for our purpose of defining $H_{1+\theta}^{\downarrow,W}(X|Y)$ and 
$H_{1+\theta}^{\uparrow,W}(X|Y)$, other cases are irrelevant.}
\begin{eqnarray}
H_{1+\theta}^{W|V}(X|Y) := - \frac{1}{\theta} \log \lambda_{\theta}^{W|V}
\end{eqnarray}
for $\theta \in (-1,0) \cup (0,\infty)$, where $\lambda_\theta^{W|V}$
is the Perron-Frobenius eigenvalue of 
\begin{eqnarray}
W(x,y|x^\prime,y^\prime)^{1+\theta} V(y|y^\prime)^{-\theta}.
\end{eqnarray}
By using this measure, we obviously have
\begin{eqnarray}
H_{1+\theta}^{\downarrow,W}(X|Y) = H_{1+\theta}^{W|W}(X|Y).
\end{eqnarray}
Furthermore, under Assumption \ref{assumption-memory-through-Y}, we can show that 
\begin{eqnarray} \label{eq:alternative-definition-of-upper-conditional-W}
H_{1+\theta}^{\uparrow,W}(X|Y) = \max_V H_{1+\theta}^{W|V}(X|Y)
\end{eqnarray}
holds (see Appendix \ref{appendix:eq:alternative-definition-of-upper-conditional-W} for the proof), 
where the maximum is taken over all transition matrices 
satisfying ${\cal Y}_W^2 \subset {\cal Y}_V^2$.
\end{remark}

Next, we investigate some properties of the information measures introduced in this section.
The following lemma is proved in Appendix \ref{appendix:lemma:multi-terminal-markov-property}.

\begin{lemma} \label{lemma:multi-terminal-markov-property}
$\phantom{a}$
\begin{enumerate}
\item \label{item:multi-terminal-markov-property-1} 
The function $\theta H_{1+\theta}^{\downarrow,W}(X|Y)$ is a concave function of $\theta$, and it is strict concave iff. 
$\san{V}^W(X|Y) > 0$.

\item \label{item:multi-terminal-markov-property-1-b}
$H_{1+\theta}^{\downarrow,W}(X|Y)$ is a monotonically decreasing function of $\theta$. 

\item \label{item:multi-terminal-markov-property-2} 
The function $\theta H_{1+\theta}^{\uparrow,W}(X|Y)$ is a concave function of $\theta$, and it is strict concave iff.
$\san{V}^W(X|Y) > 0$.

\item \label{item:multi-terminal-markov-property-2-b}
$H_{1+\theta}^{\uparrow,W}(X|Y)$ is a monotonically decreasing function of $\theta$. 

\item \label{item:multi-terminal-markov-property-3} 
For every $\theta \in (-1,0) \cup (0,\infty)$, we have 
$ H_{1+\theta}^{\downarrow,W}(X|Y) \le H_{1+\theta}^{\uparrow,W}(X|Y)$.

\item \label{item:multi-terminal-markov-property-4} 
For fixed $\theta^\prime$, the function $\theta H_{1+\theta,1+\theta^\prime}^W(X|Y)$ is a concave function of $\theta$,
and it is strict concave iff. $\san{V}^W(X|Y) > 0$.

\item \label{item:multi-terminal-markov-property-4-b} 
For fixed $\theta^\prime$, $H_{1+\theta,1+\theta^\prime}^W(X|Y)$ is a monotonically decreasing function of $\theta$.

\item \label{item:multi-terminal-markov-property-5} 
We have
\begin{eqnarray}
H_{1+\theta,1}^W(X|Y) = H_{1+\theta}^{\downarrow,W}(X|Y). 
\end{eqnarray}

\item \label{item:multi-terminal-markov-property-6} 
We have
\begin{eqnarray}
H_{1+\theta,1+\theta}^W(X|Y) = H_{1+\theta}^{\uparrow,W}(X|Y). 
\end{eqnarray}

\item \label{item:multi-terminal-markov-property-7} 
For every $\theta \in (-1,0) \cup (0,\infty)$, 
$H_{1+\theta,1+\theta^\prime}^W(X|Y)$ is maximized at $\theta^\prime = \theta$, i.e.,
\begin{align} \label{eq:derivative-second-argument}
\frac{d[H_{1+\theta,1+\theta^\prime}^W(X|Y)]}{d\theta^\prime} \bigg|_{\theta^\prime = \theta} = 0.
\end{align}


\end{enumerate}
\end{lemma}

From Statement \ref{item:multi-terminal-markov-property-1} of Lemma \ref{lemma:multi-terminal-markov-property},
$\frac{d [\theta H_{1+\theta}^{\downarrow,W}(X|Y)]}{d\theta}$ is monotonically decreasing.
Thus, we can define the inverse function $\theta(a) = \theta^\downarrow(a)$ of 
$\frac{d [\theta H_{1+\theta}^{\downarrow,W}(X|Y)]}{d\theta}$ by
\begin{eqnarray} \label{eq:definition-theta-inverse-multi-markov}
\frac{d [\theta H_{1+\theta}^{\downarrow,W}(X|Y)]}{d\theta} \bigg|_{\theta = \theta(a)} = a
\end{eqnarray}
for $\underline{a} < a \le \overline{a}$, where 
$\underline{a} := \lim_{\theta \to \infty} \frac{d [\theta H_{1+\theta}^{\downarrow,W}(X|Y)]}{d\theta}$
and $\overline{a} := \lim_{\theta \to -1} \frac{d [\theta H_{1+\theta}^{\downarrow,W}(X|Y)]}{d\theta}$.
Let 
\begin{eqnarray}
R(a) := (1+\theta(a)) a - \theta(a) H_{1+\theta(a)}^{\downarrow,W}(X|Y).
\end{eqnarray}
Since 
\begin{eqnarray}
R^\prime(a) = (1+\theta(a)),
\end{eqnarray}
$R(a)$ is a monotonic increasing function of $\underline{a} < a < R(\overline{a})$. Thus, we can define 
the inverse function $a(R ) = a^\downarrow(R )$ of $R(a)$ by
\begin{eqnarray} \label{eq:definition-a-inverse-multi-markov}
(1+\theta(a(R ))) a(R ) - \theta(a(R )) H_{1+\theta(a(R ))}^{\downarrow,W}(X|Y) = R
\end{eqnarray}
for $R(\underline{a}) < R < H_0^{\downarrow,W}(X|Y)$, where 
$H_0^{\downarrow,W}(X|Y) := \lim_{\theta \to -1} H_{1+\theta}^{\downarrow,W}(X|Y)$.

For $\theta H_{1+\theta}^{\uparrow,W}(X|Y)$, by the same reason, we can define the inverse 
function $\theta(a) = \theta^\uparrow(a)$ by
\begin{eqnarray} \label{eq:definition-theta-inverse-markov-optimal-Q}
\frac{d [\theta H_{1+\theta,1+\theta(a)}^W(X|Y)]}{d \theta} \bigg|_{\theta = \theta(a)}
= \frac{d [\theta H_{1+\theta}^{\uparrow,W}(X|Y)]}{d\theta} \bigg|_{\theta = \theta(a)} = a,
\end{eqnarray}
and the inverse function $a(R ) = a^\uparrow(R )$ of 
\begin{eqnarray} \label{eq:definition-R-a-optimal-q-markov}
R(a) := (1+\theta(a)) a - \theta(a) H_{1+\theta(a)}^{\uparrow,W}(X|Y) 
\end{eqnarray}
by
\begin{eqnarray} \label{eq:definition-a-inverse-markov-optimal-Q}
(1+\theta(a(R ))) a(R ) - \theta(a(R )) H_{1+\theta(a(R ))}^{\uparrow,W}(X|Y) = R,
\end{eqnarray}
for $R(\underline{a}) < R < H_0^{\uparrow,W}(X|Y)$, where 
$H_0^{\uparrow,W}(X|Y) := \lim_{\theta \to -1} H_{1+\theta}^{\uparrow,W}(X|Y)$.
Here, the first equality in \eqref{eq:definition-theta-inverse-markov-optimal-Q} follows from \eqref{eq:derivative-second-argument}.

Since $\theta \mapsto \theta H_{1+\theta}^{\downarrow,W}(X|Y)$ is concave, 
and $-1 \le \theta^\downarrow(R ) \le 0$ for $H^W(X|Y) \le R \le H_0^{\downarrow,W}(X|Y)$, we can prove the following.
\begin{lemma} \label{lemma:legendra-transform}
The function $\theta(R)$ defined in \eqref{eq:definition-theta-inverse-multi-markov} satisfies
\begin{align}
\sup_{-1 \le \theta \le 0} [-\theta R + \theta H_{1+\theta}^{\downarrow,W}(X|Y)] 
 = - \theta(R ) R + \theta(R ) H_{1+\theta(R )}^{\downarrow,W}(X|Y) 
\end{align}
for $H^W(X|Y) \le R \le H_0^{\downarrow,W}(X|Y)$.
\end{lemma}

Furthermore, we can show the following.
\begin{lemma} \label{lemma:legendra-transform-2}
The function $\theta(a(R))$ defined by \eqref{eq:definition-a-inverse-multi-markov} satisfies
\begin{align} \label{eq:legendra-transform-2-lower}
\sup_{-1 \le \theta \le 0} \frac{- \theta R + \theta H_{1+\theta}^{\downarrow,W}(X|Y)}{1+\theta}
 = - \theta(a(R )) a(R ) + \theta(a(R )) H_{1+\theta(a(R ))}^{\downarrow,W}(X|Y)
\end{align}
for $H^W(X|Y) \le R \le H_0^{\downarrow,W}(X|Y)$, and the function $\theta(a(R))$ defined in \eqref{eq:definition-a-inverse-markov-optimal-Q} satisfies
\begin{align} \label{eq:legendra-transform-2-upper}
\sup_{-1 \le \theta \le 0} \frac{- \theta R + \theta H_{1+\theta}^{\uparrow,W}(X|Y)}{1+\theta}
 = - \theta(a(R )) a(R ) + \theta(a( R)) H_{1+\theta(a(R ))}^{\uparrow,W}(X|Y)
\end{align}
for $H^W(X|Y) \le R \le H_0^{\uparrow,W}(X|Y)$.
\end{lemma}
\begin{proof}
See Appendix \ref{subsection:proof-lemma:legendra-transform-2}.
\end{proof}

\begin{remark}
As we can find from \eqref{eq:lower-conditional-renyi-markov-theta-0}, \eqref{eq:lower-conditional-renyi-markov-theta-0-derivative},
and Lemma \ref{lemma:properties-upper-conditional-renyi-transition-matrix}, both the conditional R\'enyi entropies expand as
\begin{eqnarray}
H_{1+\theta}^{\downarrow,W}(X|Y) &=& H^W(X|Y) - \frac{1}{2} \san{V}^W(X|Y) \theta + o(\theta), \\
H_{1+\theta}^{\uparrow,W}(X|Y) &=& H^W(X|Y) - \frac{1}{2} \san{V}^W(X|Y) \theta + o(\theta)
\end{eqnarray}
around $\theta = 0$.
Thus, the difference of these measures significantly appear only when
$|\theta|$ is rather large. For the transition matrix of Example \ref{example:gilbert-elliot-noise} with
$q_0 = q_1 = 0.1$, $p_0 = 0.1$, and $p_1 = 0.4$, we plotted the values of the information measures 
in Fig.~\ref{Fig:Comparison-Renyi-Entropies}. Although the values at $\theta = -1$ coincide in 
Fig.~\ref{Fig:Comparison-Renyi-Entropies}, note that the values at $\theta = -1$ may differ  in general.

In Example \ref{example:markov-Y-plus-additive-noise}, 
we have mentioned that transition matrix $W$ in \eqref{eq:example-markov-plus-additive-noise} 
satisfies Assumtption \ref{assumption-memory-through-Y}
when transition matrix $P$ is given by \eqref{eq:symmetric-additive-noise}.
In this case, we can find that 
\begin{eqnarray}
H_{1+\theta}^{\uparrow,W}(X|Y) &=& H_{1+\theta}^{\downarrow,W}(X|Y) \\
&=& H_{1+\theta}(P_Z),
\end{eqnarray}
i.e., the two kinds of conditional R\'enyi entropies coincide.
\end{remark}
\begin{figure}[t]
\centering
\includegraphics[width=0.5\linewidth]{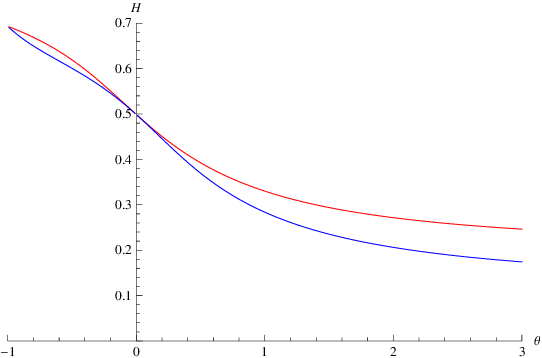}
\caption{A comparison of $H_{1+\theta}^{\uparrow,W}(X|Y)$ (red curve)
and $H_{1+\theta}^{\downarrow,W}(X|Y)$ (blue curve)
for the transition matrix of Example \ref{example:gilbert-elliot-noise} with
$q_0 = q_1 = 0.1$, $p_0 = 0.1$, and $p_1 = 0.4$.
The horizontal axis is $\theta$, and 
the vertical axis is the values of the information measures (nats).
}
\label{Fig:Comparison-Renyi-Entropies}
\end{figure}

Now, let's consider asymptotic behavior of $H_{1+\theta}^{\downarrow,W}(X|Y)$ around $\theta = 0$.
When $\theta(a)$ is close to $0$, we have
\begin{align} \label{eq:expansion-theta-a}
\theta(a) H_{1+\theta(a)}^{\downarrow,W}(X|Y)
 = \theta(a) H^W(X|Y) - \frac{1}{2} \san{V}^W(X|Y) \theta(a)^2 + o(\theta(a)^2).
\end{align}
Taking the derivative, \eqref{eq:definition-theta-inverse-multi-markov} implies that 
\begin{align} \label{eq:expansion-a}
a = H^W(X|Y) - \san{V}^W(X|Y) \theta(a) + o(\theta(a)).
\end{align}
Hence, when $R$ is close to $H^W(X|Y)$, we have
\begin{align}
R 
&= (1+\theta(a(R )) a(R ) - \theta(a(R )) H_{1+\theta(a(R ))}^{\downarrow,W}(X|Y) \\
&= H^W(X|Y) - \left(1 + \frac{\theta(a(R ))}{2}\right) \theta(a(R )) \san{V}^W(X|Y) + o(\theta(a(R )), 
\end{align}
i.e., 
\begin{align}
\theta(a(R )) = \frac{-R + H^W(X|Y)}{\san{V}^W(X|Y)} + o\left(  \frac{R - H^W(X|Y)}{\san{V}^W(X|Y)} \right).
\end{align}
Furthermore, \eqref{eq:expansion-theta-a} and \eqref{eq:expansion-a} imply
\begin{align}
& - \theta(a(R )) a(R ) + \theta(a(R )) H_{1+\theta(a(R ))}^{\downarrow,W}(X|Y) \\
&= \san{V}^W(X|Y) \frac{\theta(a(R ))^2}{2} + o(\theta(a(R ))^2) \\
&= \frac{\san{V}^W(X|Y)}{2} \left( \frac{R - H^W(X|Y)}{\san{V}^W(X|Y)} \right)^2 
 + o\left( \left( \frac{R - H^W(X|Y)}{\san{V}^W(X|Y)} \right)^2\right).
  \label{eq:expansion-exponent-around-H}
\end{align}

\subsection{Information Measures for Markov Chain} 
\label{subsection-multi-terminal-information-measures-markov}

Let $(\mathbf{X},\mathbf{Y})$ be the Markov chain induced by transition matrix $W$ and some initial distribution $P_{X_1Y_1}$.
Now, we show how information measures introduced in Section \ref{subsection:multi-terminal-measures-markov} are 
related to the conditional R\'enyi entropy rates.
First, we introduce the following lemma, which
gives finite upper and lower bounds on the lower conditional R\'enyi entropy.

\begin{lemma} \label{lemma:mult-terminal-finite-evaluation-down-conditional-renyi}
Suppose that transition matrix $W$ satisfies Assumption \ref{assumption-Y-marginal-markov}.
Let $v_\theta$ be the eigenvector of $W_\theta^T$ with respect to the Perron-Frobenius eigenvalue $\lambda_\theta$
such that $\min_{x,y} v_\theta(x,y) = 1$.\footnote{Since the eigenvector corresponding to the 
Perron-Frobenius eigenvalue for an irreducible non-negative matrix has always strictly positive
entries \cite[Theorem 8.4.4, p.~508]{horn-johnson}, we can choose the eigenvector $v_\theta$ satisfying this condition.} Let $w_\theta(x,y) := P_{X_1 Y_1}(x,y)^{1+\theta} P_{Y_1}(y)^{-\theta}$. Then, 
for every $n\ge 1$, we have
\begin{eqnarray}
(n-1) \theta H_{1+\theta}^{\downarrow,W}(X|Y) + \underline{\delta}(\theta)
 \le \theta H_{1+\theta}^\downarrow(X^n|Y^n) 
 \le (n-1) \theta H_{1+\theta}^{\downarrow,W}(X|Y) + \overline{\delta}(\theta), 
 \end{eqnarray}
 where 
 \begin{eqnarray}
 \overline{\delta}(\theta) &:=& - \log \langle v_\theta | w_\theta \rangle + \log \max_{x,y} v_\theta(x,y),
    \label{eq:definition-overline-delta} \\
 \underline{\delta}(\theta) &:=& - \log \langle v_\theta | w_\theta \rangle, \label{eq:definition-underline-delta} 
 \end{eqnarray}
 and $\langle v_\theta | w_\theta \rangle$ is defined as $\sum_{x,y} v_\theta(x,y) w_\theta(x,y)$.
\end{lemma}
\begin{proof}
It follows from \eqref{eq:relation-down-conditional-renyi-moment-transition-matrix} and 
Lemma \ref{lemma:finite-evaluation-of-cgf}.
\end{proof}

From Lemma \ref{lemma:mult-terminal-finite-evaluation-down-conditional-renyi}, we have the following.

\begin{theorem} \label{theorem:asymptotic-down-conditional-renyi}
Suppose that transition matrix $W$ satisfies Assumption \ref{assumption-Y-marginal-markov}.
For any initial distribution, we have
\begin{eqnarray}
\lim_{n\to\infty} \frac{1}{n} H_{1+\theta}^\downarrow(X^n|Y^n) &=& H_{1+\theta}^{\downarrow,W}(X|Y), \\
\lim_{n\to\infty} \frac{1}{n} H(X^n|Y^n) &=& H^W(X|Y).
\end{eqnarray}
\end{theorem}

We also have the following asymptotic evaluation of the variance,
which follows from Lemma \ref{lemma:appendix-variance-limit} in Appendix \ref{Appendix:preparation}. 

\begin{theorem} \label{theorem:multi-markov-variance}
Suppose that transition matrix $W$ satisfies Assumption \ref{assumption-Y-marginal-markov}.
For any initial distribution, we have
\begin{eqnarray}
 \lim_{n\to\infty} \frac{1}{n} \san{V}(X^n|Y^n) = \san{V}^W(X|Y).
\end{eqnarray}
\end{theorem}

Theorem \ref{theorem:multi-markov-variance} is practically important since the limit of the variance can
be described by a single letter characterized quantity. A method to calculate $\san{V}^W(X|Y)$ can be 
found in \cite{hayashi-watanabe:13b}.

Next, we show the lemma that gives finite upper and lower bound on the upper conditional R\'enyi entropy in terms of
the upper conditional R\'enyi entropy for the transition matrix.

\begin{lemma} \label{lemma:multi-terminal-finite-evaluation-upper-conditional-renyi}
Suppose that transition matrix $W$ satisfies Assumption \ref{assumption-memory-through-Y}. 
Let $v_\theta$ be the eigenvector of
$K_\theta^T$ with respect to the Perro-Frobenius eigenvalue $\kappa_\theta$ such that $\min_y v_\theta(y) = 1$.
Let $w_\theta$ be the $|{\cal Y}|$-dimensional vector defined by
\begin{eqnarray}
w_\theta(y) := \left[ \sum_x P_{X_1 Y_1}(x,y)^{1+\theta} \right]^{\frac{1}{1+\theta}}.
\end{eqnarray}
Then, we have
\begin{eqnarray}
(n-1) \frac{\theta}{1+\theta} H_{1+\theta}^{\uparrow,W}(X|Y) + \underline{\xi}(\theta)
 \le \frac{\theta}{1+\theta}H_{1+\theta}^\uparrow(X^n|Y^n)
 \le (n-1) \frac{\theta}{1+\theta} H_{1+\theta}^{\uparrow,W}(X|Y) + \overline{\xi}(\theta),
\end{eqnarray}
where 
\begin{eqnarray}
\overline{\xi}(\theta) &:=& - \log \langle v_\theta | w_\theta \rangle + \log \max_y v_\theta(y), \label{eq:definition-overline-xi} \\
\underline{\xi}(\theta) &:=& - \log \langle v_\theta | w_\theta \rangle. \label{eq:definition-underline-xi}
\end{eqnarray}
\end{lemma}
\begin{proof}
See Appendix \ref{appendix:proof-lemma:multi-terminal-finite-evaluation-upper-conditional-renyi}.
\end{proof}

From Lemma \ref{lemma:multi-terminal-finite-evaluation-upper-conditional-renyi}, we have the following.

\begin{theorem} \label{theorem:asymptotic-up-conditional-renyi}
Suppose that transition matrix $W$ satisfies Assumption \ref{assumption-memory-through-Y}. 
For any initial distribution, we have
\begin{eqnarray}
\lim_{n\to\infty} \frac{1}{n} H_{1+\theta}^\uparrow(X^n|Y^n) &=& H_{1+\theta}^{\uparrow,W}(X|Y).
 \label{eq:markov-theta-entropy-up-asymptotic} 
\end{eqnarray}
\end{theorem}

Finally, we show the lemma that gives finite upper and lower bounds on the two-parameter conditional 
R\'enyi entropy in terms of the two-parameter conditional R\'enyi entropy for the transition matrix. 

\begin{lemma} \label{lemma:multi-terminal-finite-evaluation-two-parameter-conditional-renyi}
Suppose that transition matrix $W$ satisfies Assumption \ref{assumption-memory-through-Y}. 
Let $v_{\theta,\theta^\prime}$ be the eigenvector of
$N_{\theta,\theta^\prime}^T$ with respect to the Perro-Frobenius eigenvalue $\nu_{\theta,\theta^\prime}$ such that 
$\min_y v_{\theta,\theta^\prime}(y) = 1$.
Let $w_{\theta,\theta^\prime}$ be the $|{\cal Y}|$-dimensional vector defined by
\begin{eqnarray}
w_{\theta,\theta^\prime}(y) := \left[ \sum_x P_{X_1 Y_1}(x,y)^{1+\theta} \right] 
 \left[ \sum_x P_{X_1 Y_1}(x,y)^{1+\theta^\prime} \right]^{\frac{-\theta}{1+\theta^\prime}}.
\end{eqnarray}
Then, we have
\begin{eqnarray}
(n-1) \theta H_{1+\theta,1+\theta^\prime}^W(X|Y) + \underline{\zeta}(\theta,\theta^\prime)
 \le \theta H_{1+\theta,1+\theta^\prime}(X^n|Y^n)
 \le (n-1) \theta H_{1+\theta,1+\theta^\prime}^W(X|Y) + \overline{\zeta}(\theta,\theta^\prime),
\end{eqnarray}
where 
\begin{eqnarray}
\overline{\zeta}(\theta,\theta^\prime) &:=& - \log \langle v_{\theta,\theta^\prime} | w_{\theta,\theta^\prime} \rangle
   + \log \max_y v_{\theta,\theta^\prime}(y) + \theta \overline{\xi}(\theta^\prime), \label{eq:definition-overline-zeta} \\
\underline{\zeta}(\theta,\theta^\prime) &:=& 
 - \log \langle v_{\theta,\theta^\prime} | w_{\theta,\theta^\prime} \rangle + \theta \underline{\xi}(\theta^\prime) 
  \label{eq:definition-underline-zeta}
\end{eqnarray}
for $\theta > 0$ and
\begin{eqnarray}
\overline{\zeta}(\theta,\theta^\prime) &:=& - \log \langle v_{\theta,\theta^\prime} | w_{\theta,\theta^\prime} \rangle
   + \log \max_y v_{\theta,\theta^\prime}(y) + \theta \underline{\xi}(\theta^\prime), \label{eq:definition-overline-zeta-2} \\
\underline{\zeta}(\theta,\theta^\prime) &:=& 
 - \log \langle v_{\theta,\theta^\prime} | w_{\theta,\theta^\prime} \rangle + \theta \overline{\xi}(\theta^\prime)
 \label{eq:definition-underline-zeta-2}
\end{eqnarray}
for $\theta < 0$
\end{lemma}
\begin{proof}
We can write
\begin{eqnarray}
\lefteqn{ \theta H_{1+\theta,1+\theta^\prime}(X^n|Y^n) } \\
 &=& - \log \sum_{y^n} \left[ \sum_{x^n} P_{X^nY^n}(x^n,y^n)^{1+\theta} \right] 
  \left[ \sum_{x^n} P_{X^nY^n}(x^n,y^n)^{1+\theta^\prime} \right]^{\frac{-\theta}{1+\theta^\prime}}
   + \frac{\theta \theta^\prime}{1+\theta^\prime} H_{1+\theta^\prime}^\uparrow(X^n|Y^n).
\end{eqnarray}
The second term is evaluated by Lemma \ref{lemma:multi-terminal-finite-evaluation-upper-conditional-renyi}.
The first term can be evaluated almost the same manner as Lemma \ref{lemma:multi-terminal-finite-evaluation-upper-conditional-renyi}.
\end{proof}

From Lemma \ref{lemma:multi-terminal-finite-evaluation-two-parameter-conditional-renyi}, we have the following.

\begin{theorem} \label{theorem:asymptotic-two-parameter-renyi-markov}
Suppose that transition matrix $W$ satisfies Assumption \ref{assumption-memory-through-Y}. 
For any initial distribution, we have
\begin{eqnarray}
\lim_{n\to\infty} \frac{1}{n} H_{1+\theta,1+\theta^\prime}(X^n|Y^n) 
 = H_{1+\theta,1+\theta^\prime}^W(X|Y).
\end{eqnarray}
\end{theorem}


%% file: Multi-Source.tex
\section{Source Coding with Full Side-Information} \label{section:multi-source}

In this section, we investigate the source coding with side-information.
We start this section by showing the problem setting in Section \ref{subsection:multi-source-problem-formulation}.
Then, we review and introduce some single-shot bounds in Section \ref{subsection:multi-source-one-shot}.
We derive finite-length bounds for the Markov chain in Section \ref{subsection:multi-source-finite-markov}.
Then, in Sections \ref{subsection:multi-source-large-deviation} and \ref{subsection:multi-source-mdp},
we show the asymptotic characterization for the large deviation regime and the moderate deviation regime
by using those finite-length bounds. We also derive the second order rate in Section \ref{subsection:multi-source-second-order}.


The results shown in this section are summarized in Table \ref{table:summary:multi-source}.
The checkmarks $\checkmark$ indicate that the tight asymptotic bounds (large deviation, moderate deviation, and second order)
 can be obtained from those bounds.
The marks $\checkmark^*$ indicate that the large deviation bound can be derived up to the critical rate. 
The computational complexity "Tail" indicates that the computational complexities of those bounds 
depend on the computational complexities of tail probabilites.
It should be noted that Theorem \ref{theorem:multi-source-finite-markov-converse-assumptiotn-1}
is derived from a special case ($Q_{Y} = P_Y$) of Theorem \ref{theorem:multi-source-converse}.
The asymptotically optimal choice is $Q_Y = P_Y^{(1+\theta)}$, which corresponds to Corollary \ref{corollary:multi-source-converse}.
Under Assumption \ref{assumption-Y-marginal-markov}, we can derive the bound of the Markov case only for that
special choice of $Q_Y$, while under Assumption \ref{assumption-memory-through-Y}, we can derive the bound of
the Markov case for the optimal choice of $Q_Y$.

\begin{table}[htbp]
\begin{center}
\caption{Summary of the bounds for source coding with full side-information.}
\label{table:summary:multi-source}
{\renewcommand{\arraystretch}{1.4}
\begin{tabular}{|c|c|c|c|c|c|c|c|} \hline
  \multirow{2}{*}{Ach./Conv.}  
  & \multirow{2}{*}{Markov} & \multirow{2}{*}{Single Shot} &  \multirow{2}{*}{$\Pse$/$\barPse$} &  \multirow{2}{*}{Complexity} 
  & Large & Moderate & Second \\
  & & & & & Deviation & Deviation & Order \\ \hline
\multirow{3}{*}{Achievability} 
 & Theorem \ref{theorem:multi-source-finite-marko-direct-assumptiton-1} (Ass.~1) & Lemma \ref{lemma:multi-source-loose-bound} & $\barPse$ & $O(1)$ &    & \checkmark &  \\ \cline{2-8}
 & Theorem \ref{theorem:multi-source-finite-marko-direct-assumptiton-2} (Ass.~2) & Lemma \ref{lemma:multi-source-tight-bound}  & $\barPse$ & $O(1)$ & $\checkmark^*$ & \checkmark &  \\ \cline{2-8}
 & \multicolumn{2}{|c|}{Lemma \ref{lemma:multi-source-han-direct}}   &  $\barPse$ & Tail &  & \checkmark & \checkmark \\ 
 \hline
\multirow{3}{*}{Converse} 
 & Theorem \ref{theorem:multi-source-finite-markov-converse-assumptiotn-1} (Ass.~1)  & (Theorem \ref{theorem:multi-source-converse})  & $\Pse$ & $O(1)$ &  & \checkmark &  \\ \cline{2-8} 
 & Theorem \ref{theorem:multi-source-finite-marko-converse-assumptiton-2} (Ass.~2)  & Corollary \ref{corollary:multi-source-converse}  & $\Pse$ & $O(1)$ & $\checkmark^*$ & \checkmark &  \\ \cline{2-8} 
 & \multicolumn{2}{|c|}{Lemma \ref{lemma:multi-source-han-converse}}   &  $\Pse$ & Tail &  & \checkmark & \checkmark \\ 
\hline 
\end{tabular}}
\end{center}
\end{table}

\subsection{Problem Formulation} \label{subsection:multi-source-problem-formulation}

A code $\Psi = (\san{e},\san{d})$ consists of one encoder $\san{e}:{\cal X} \to \{1,\ldots,M\}$ and one decoder
$\san{d}:\{1,\ldots,M\} \times {\cal Y} \to {\cal X}$. The decoding error probability is defined by
\begin{eqnarray}
\Pse[\Psi] &=& \Pse[\Psi|P_{XY}] \\
&:=& \Pr\{ X \neq \san{d}(\san{e}(X),Y) \}.
\end{eqnarray}
For notational convenience, we introduce the infimum of error probabilities under the condition that the message size is $M$:
\begin{eqnarray}
\Pse(M) &=& \Pse(M|P_{XY}) \\
&:=& \inf_{\Psi} \Pse[\Psi]
\label{eq:definition-Pe-bar-source-side-info-3}.
\end{eqnarray}
For theoretical simplicity, we focus on a randomized choice of our encoder.
For this purpose, we employ a randomized hash function $F$ from ${\cal X}$ to $\{1, \ldots, M\}$.
A randomized hash function $F$ 
is called two-universal hash when $\Pr\{ F(x) = F(x^\prime) \} \le \frac{1}{M}$
for any distinctive $x$ and $x^\prime$ \cite{wegman:81}; 
the so-called bin coding \cite{cover} is an example of two-universal hash function.
In the following, we denote the set of two-universal hash functions by 
${\cal F}$.
Given an encoder $f$ as a function from ${\cal X}$ to $\{1, \ldots, M\}$,
we define the decoder $\san{d}_f$ as the optimal decoder
by $\argmin_{\san{d}} \Pse[(f,\san{d})]$. Then, we denote the code $(f,\san{d}_f)$ by $\Psi(f)$.
Then, we bound the error probability $\Pse[\Psi(F)]$ averaged
over the random function $F$ by only using the property of two-universality. 
In order to consider the worst case of such schemes, we introduce the following quantity:
\begin{eqnarray} \label{eq:definition-Pe-bar-source-side-info}
\barPse(M) &=& \barPse(M|P_{XY}) \\
&:=& \sup_{F \in {\cal F}}  \mathbb{E}_F[\Pse[\Psi(F)]], \label{eq:definition-Pe-bar-source-side-info-2}.
\end{eqnarray}

When we consider $n$-fold extension, the source code and related quantities are denoted with the superscript $(n)$.
For example,
the quantities in \eqref{eq:definition-Pe-bar-source-side-info-3} and \eqref{eq:definition-Pe-bar-source-side-info-2} are written to be
$\Pse^{(n)}(M)$ and $\barPse^{(n)}(M)$, respectively.
Instead of evaluating them, we are often interested in evaluating 
\begin{eqnarray}
M(n,\varepsilon) &:=& \inf\{ M_n : \Pse^{(n)}(M_n) \le \varepsilon \},  \\
\bar{M}(n,\varepsilon) &:=& \inf\{ M_n : \barPse^{(n)}(M_n) \le \varepsilon \} 
\end{eqnarray}
for given $0 \le \varepsilon < 1$.

\subsection{Single Shot Bounds} \label{subsection:multi-source-one-shot}

In this section, we review existing single shot bounds
and also show novel converse bounds. 
For the information meaures used below, see Section \ref{section:preparation-multi}.

\textchange{By using the standard argument on information-spectrum approach, we have the following achievability bound.}
\begin{lemma}[Lemma 7.2.1 of \cite{han:book}] \label{lemma:multi-source-han-direct}
The following bound holds:
\begin{eqnarray}
\barPse(M) \le \inf_{\gamma \ge 0} \left[ P_{XY}\left\{ \log \frac{1}{P_{X|Y}(x|y)} > \gamma \right\} + \frac{e^\gamma}{M} \right].
\end{eqnarray}
\end{lemma}

\textchange{Although Lemma \ref{lemma:multi-source-han-direct} is useful for the second-order regime, it is known to be
not tight in the large deviation regime.
By using the large deviation technique of Gallager, we have the following exponential type achievability bound.}
\begin{lemma}[\cite{gallager:76}] \label{lemma:multi-source-tight-bound}
The following bound holds:\footnote{Note that the Gallager function and the upper conditional R\'enyi entropy 
are related by \eqref{eq:relation-upper-renyi-gallager}.}
\begin{eqnarray} \label{eq:gallager-bound}
\barPse(M) \le \inf_{-\frac{1}{2} \le \theta \le 0} M^{\frac{\theta}{1+\theta}} e^{- \frac{\theta}{1+\theta} H_{1+\theta}^\uparrow(X|Y)}.
\end{eqnarray}
\end{lemma}

\textchange{Although Lemma \ref{lemma:multi-source-tight-bound} is known to be tight in the large deviation regime
for i.i.d. sources,  $H_{1+\theta}^\uparrow(X|Y)$ for Markov chains can only be evaluated under the strongly non-hidden assumption. 
For this reason, even though the following bound is looser than Lemma \ref{lemma:multi-source-tight-bound},
it is useful to have another bound in terms of $H_{1+\theta}^\downarrow(X|Y)$, which 
can be evaluated for Markov chains under the non-hidden assumption.}
\begin{lemma} \label{lemma:multi-source-loose-bound}
The following bound holds:
\begin{eqnarray}
\barPse(M) \le \inf_{-1 \le \theta \le 0} M^{\theta} e^{-\theta H_{1+\theta}^\downarrow(X|Y)}. 
\end{eqnarray}
\end{lemma}
\begin{proof}
To derive this bound, we change variable in \eqref{eq:gallager-bound} as 
$\theta = \frac{\theta^\prime}{1-\theta^\prime}$. Then, $-1 \le \theta^\prime \le 0$, and we have
\begin{eqnarray*}
M^{\theta^\prime} e^{-\theta^\prime H_{\frac{1}{1-\theta^\prime}}^\uparrow(X|Y)} \le 
 M^{\theta^\prime} e^{- \theta^\prime H_{1+\theta^\prime}^\downarrow(X|Y)},
\end{eqnarray*}
where we used Lemma \ref{lemma:appendix-upper-conditional-renyi-upper-and-lower} 
in Appendix \ref{Appendix:lemma:property-upper-conditional-renyi-single-shot}.
\end{proof}

When ${\cal Y}$ is singleton, we have the following bound, which is tighter than Lemma \ref{lemma:multi-source-tight-bound}.
\begin{lemma}[(2.39) \cite{hayashi-book:06}] \label{lemma:tight-bound-singleton-Y}
The following bound holds
\begin{eqnarray}
\Pe(M) \le \inf_{-1 < \theta \le 0} M^{\frac{\theta}{1+\theta}} e^{- \frac{\theta}{1+\theta} H_{1+\theta}(X)}.
\end{eqnarray}
\end{lemma}

\textchange{For converse part, we first have the following bound, which is very close to the operational definition of 
source coding with side-information.}
\begin{lemma}[\cite{renner:05d}] \label{lemma:multi-source-one-shot-hypo-converse}
Let $\{ \Omega_y \}_{y \in {\cal Y}}$ be a family of subsets $\Omega_y \subset {\cal X}$, and let 
$\Omega = \cup_{y\in{\cal Y}} \Omega_y \times \{ y \}$. Then, for any $Q_Y \in {\cal P}({\cal Y})$, the following bound holds:
\begin{eqnarray}
\Pse(M) \ge \min_{\{ \Omega_y \}}\left\{ P_{XY}(\Omega^c) : \sum_y Q_Y(y) |\Omega_y| \le M  \right\}. 
\end{eqnarray}
\end{lemma}

\textchange{Since Lemma \ref{lemma:multi-source-one-shot-hypo-converse} is close to the operational definition, it is not 
easy to evaluate Lemma \ref{lemma:multi-source-one-shot-hypo-converse}. Thus, we derive another bound by loosening Lemma \ref{lemma:multi-source-one-shot-hypo-converse},
which is more tractable for evaluation.}
Slightly weakening Lemma \ref{lemma:multi-source-one-shot-hypo-converse}, we have the following.
\begin{lemma} [\cite{han:book,hayashi:03}] \label{lemma:multi-source-han-converse}
For any $Q_Y \in {\cal P}({\cal Y})$, we have\footnote{In fact, a special case
for $Q_Y = P_Y$ correspond to Lemma 7.2.2 of \cite{han:book}. A bound that involve $Q_Y$ was 
introduced in \cite{hayashi:03} for channel coding, and it can be regarded as a source coding counterpart of that result.}
\begin{eqnarray}
\Pse(M) \ge \sup_{\gamma \ge 0} \left[ P_{XY}\left\{ \log \frac{Q_Y(y)}{P_{XY}(x,y)} > \gamma  \right\} - \frac{M}{e^{\gamma}} \right].
\end{eqnarray}
\end{lemma}

\textchange{By using the change-of-measure argument,
we can also derive the following converse bound.}
\begin{theorem} \label{theorem:multi-source-converse}
For any $Q_Y \in {\cal P}({\cal Y})$, we have
\begin{eqnarray}
\lefteqn{ - \log \Pse(M) } \\
&\le& \inf_{s > 0 \atop \tilde{\theta} \in \mathbb{R}, \vartheta \ge 0}
\bigg[ (1+s) \tilde{\theta}\left\{ H_{1+\tilde{\theta}}(P_{XY}|Q_Y) - H_{1+(1+s)\tilde{\theta}}(P_{XY}|Q_Y) \right\} \\
 && - (1+s) \log \left( 1- 2 e^{- \frac{- \vartheta R + (\tilde{\theta} + \vartheta(1+\tilde{\theta})) 
  H_{\tilde{\theta} + \vartheta(1+\tilde{\theta})}(P_{XY}|Q_Y) - (1+\vartheta)\tilde{\theta} H_{1+\tilde{\theta}}(P_{XY}|Q_Y)}{1+\vartheta}} \right) \bigg] / s
  \label{eq:bound-multi-source-exponential-converse-1} \\
&\le& \inf_{s > 0 \atop -1 < \tilde{\theta} < \theta(a(R )) }
\bigg[
 (1+s) \tilde{\theta}\left\{ H_{1+\tilde{\theta}}(P_{XY}|Q_Y) - H_{1+(1+s)\tilde{\theta}}(P_{XY}|Q_Y) \right\} \\
&&  - (1+s) \log \left( 1 - 2 e^{(\theta(a(R )) - \tilde{\theta}) a(R ) - \theta(a(R )) H_{1+\theta(a(R ))}(P_{XY}|Q_Y) + \tilde{\theta} H_{1+\tilde{\theta}}(P_{XY}|Q_Y)} \right)
\bigg] / s, 
 \label{eq:bound-multi-source-exponential-converse-2}
\end{eqnarray}
where $R = \log M$, 
and $\theta(a) = \theta^Q(a)$ and $a(R ) = a^Q(R )$ are the inverse functions defined in 
\eqref{eq:definition-inverse-theta-multi-one-shot} 
and \eqref{eq:definition-inverse-a-multi-one-shot} respectively.
\end{theorem}
\begin{proof}
See Appendix \ref{appendix:theorem:multi-source-converse}.
\end{proof}

In particular, by taking $Q_Y = P_Y^{(1+\theta(a(R )))}$ in Theorem \ref{theorem:multi-source-converse}, we have the following.
\begin{corollary} \label{corollary:multi-source-converse}
We have
\begin{eqnarray}
\lefteqn{ - \log \Pse(M) } \\ 
&\le& \inf_{s > 0 \atop -1 < \tilde{\theta} < \theta(a(R )) }
\bigg[
 (1+s) \tilde{\theta} \left\{ H_{1+\tilde{\theta},1+\theta(a(R ))}(X|Y) - H_{1+(1+s)\tilde{\theta},1+\theta(a(R ))}(X|Y) \right\} \\
&&  - (1+s) \log \left( 1 - 2 e^{(\theta(a(R )) - \tilde{\theta}) a(R ) - \theta(a(R )) H_{1+\theta(a(R ))}^\uparrow(X|Y)
 + \tilde{\theta} H_{1+\tilde{\theta}, 1+\theta(a(R ))}(X|Y) } \right)
\bigg] / s, 
 \label{eq:bound-multi-source-exponential-Gallager-converse-2}
\end{eqnarray}
where 
$\theta(a) = \theta^\uparrow(a)$ and $a(R ) = a^\uparrow(R )$ 
are the inverse functions defined in \eqref{eq:definition-rho-inverse-Gallager-one-shot}
and \eqref{eq:definition-a-inverse-Gallager-one-shot}.
\end{corollary}

\begin{remark}\label{R10-2-4}
Here, it is better to discuss the possibility for extension to the continuous case.
As explained in Remark \ref{R10-2-3},
we can define the information quantities
to the case when ${\cal Y}$ is continuous but ${\cal X}$ is a discrete finite set.
The discussions in this subsection still hold even in this continuous case.
In particular, in the $n$-i.i.d. extension case with this continuous setting,
Lemma \ref{lemma:multi-source-tight-bound} and Corollary \ref{corollary:multi-source-converse}
hold when the information measures are replaced by 
$n$ times of the single-shot information measures.
\end{remark}

\subsection{Finite-Length Bounds for Markov Source} \label{subsection:multi-source-finite-markov}

\textchange{In this subsection, we derive several finite-length bounds for Markovian source with a computable
form. Unfortunately, it is not easy to evaluate how tight those bounds are only with their formula. 
Their tightness will be discussed by considering the asymptotic limit in the remaining subsections of this section.
Since we assume the irreducibility for the transition matrix describing the Markovian chain, the following bound
hold with any initial distribution. }

\textchange{To derive a lower bounds on $- \log \barPse(M_n)$ in terms of the R\'enyi entropy of transition matrix, 
we substitute the formula for the R\'enyi entropy given in Lemma \ref{lemma:mult-terminal-finite-evaluation-down-conditional-renyi}
into Lemma \ref{lemma:multi-source-loose-bound}. Then, we can derive the following achievability bound.}

\begin{theorem}[Direct, Ass. 1] \label{theorem:multi-source-finite-marko-direct-assumptiton-1}
Suppose that transition matrix $W$ satisfies Assumption \ref{assumption-Y-marginal-markov}.
Let $R := \frac{1}{n} \log M_n$. Then, for every $n \ge 1$, we have
\begin{eqnarray}
- \log \barPse^{(n)}(M_n) \ge 
\sup_{-1 \le \theta \le 0} \left[- \theta n R + (n-1) \theta H_{1+\theta}^{\downarrow,W}(X|Y) + \underline{\delta}(\theta) \right],
\end{eqnarray}
where $\underline{\delta}(\theta)$ is given by \eqref{eq:definition-underline-delta}.
\end{theorem}

When ${\cal Y}$ is singleton, from Lemma \ref{lemma:tight-bound-singleton-Y} and a special case of Lemma \ref{lemma:mult-terminal-finite-evaluation-down-conditional-renyi}, 
we have the following achievability bound.
\begin{theorem}[Direct, Singleton] \label{theorem:tight-finite-bound-singleton-Y}
Let $R := \frac{1}{n} \log M_n$. Then, for every $n \ge 1$, we have
\begin{eqnarray}
- \log \Pe^{(n)}(M_n) \ge \sup_{-1 < \theta \le 0} \frac{- n \theta R + (n-1) \theta H_{1+\theta}^W(X) + \underline{\delta}(\theta)}{1+\theta}.
\end{eqnarray}
\end{theorem}

\textchange{To derive an upper bound on $- \log \Pse(M_n)$ in terms of the R\'enyi entropy of transition matrix, we substitute 
the formula for the R\'enyi entropy given in Lemma \ref{lemma:mult-terminal-finite-evaluation-down-conditional-renyi}
to Theorem \ref{theorem:multi-source-converse}. Then, we have the following converse bound.}

\begin{theorem}[Converse, Ass. 1] \label{theorem:multi-source-finite-markov-converse-assumptiotn-1}
Suppose that transition matrix $W$ satisfies Assumption \ref{assumption-Y-marginal-markov}.
Let $R := \frac{1}{n} \log M_n$. For any $H^W(X|Y) < R < H_0^{\downarrow,W}(X|Y)$, we have
\begin{eqnarray}
\lefteqn{ - \log \Pse^{(n)}(M_n) } \\ 
&\le&
 \inf_{s > 0 \atop -1 < \tilde{\theta} < \theta(a(R )) } 
 \bigg[ (n-1) (1+s) \tilde{\theta} \left\{ H_{1+\tilde{\theta}}^{\downarrow,W}(X|Y) - H_{1+(1+s)\tilde{\theta}}^{\downarrow,W}(X|Y) \right\} 
  + \delta_1 \\
  && - (1+s) \log \left(1 - 2e^{(n-1)[ (\theta(a(R )) - \tilde{\theta}) a(R ) - \theta(a(R )) H_{1+\theta(a(R ))}^{\downarrow,W}(X|Y)
   + \tilde{\theta} H_{1+\tilde{\theta}}^{\downarrow,W}(X|Y) ] + \delta_2} \right) \bigg] / s,
\end{eqnarray}
where 
$\theta(a) = \theta^\downarrow(a)$ and $a(R ) = a^\downarrow(R )$ are the inverse functions defined by 
\eqref{eq:definition-theta-inverse-multi-markov} 
and \eqref{eq:definition-a-inverse-multi-markov} respectively, 
\begin{eqnarray}
\delta_1 &:=& (1+s) \overline{\delta}(\tilde{\theta}) - \underline{\delta}((1+s)\tilde{\theta}), \\
\delta_2 &:=& \frac{(\theta(a(R )) - \tilde{\theta}) R - (1+\tilde{\theta}) \underline{\delta}(\theta(a(R ))) + (1+\theta(a(R ))) \overline{\delta}(\tilde{\theta})}{1 + \theta(a(R ))},
\end{eqnarray}
and $\overline{\delta}(\cdot)$ and $\underline{\delta}(\cdot)$ are given by \eqref{eq:definition-overline-delta}
and \eqref{eq:definition-underline-delta}, respectively.
\end{theorem}
\begin{proof}
We first use \eqref{eq:bound-multi-source-exponential-converse-1} of Theorem \ref{theorem:multi-source-converse} 
for $Q_{Y^n} = P_{Y^n}$ and Lemma \ref{lemma:mult-terminal-finite-evaluation-down-conditional-renyi}.
Then, we restrict the range of $\tilde{\theta}$ 
as $-1 < \tilde{\theta} < \theta(a(R ))$ and set $\vartheta = \frac{\theta(a(R )) - \tilde{\theta}}{1+\tilde{\theta}}$.
Then, we have the assertion of the theorem.
\end{proof}

Next, we derive tighter bounds under Assumption \ref{assumption-memory-through-Y}.
\textchange{To derive a lower bound on $- \log \barPse(M_n)$ in terms of the R\'enyi entropy of
transition matrix, we substitute the formula for the R\'enyi entropy in Lemma \ref{lemma:multi-terminal-finite-evaluation-upper-conditional-renyi}
to Lemma \ref{lemma:multi-source-tight-bound}. Then, we have the following achievability bound.}

\begin{theorem}[Direct, Ass. 2] \label{theorem:multi-source-finite-marko-direct-assumptiton-2}
Suppose that transition matrix $W$ satisfies Assumption \ref{assumption-memory-through-Y}. 
Let $R := \frac{1}{n} \log M_n$. Then we have
\begin{eqnarray}
- \log \barPse^{(n)}(M_n) \ge 
\sup_{-\frac{1}{2} \le \theta \le 0}  \frac{-\theta nR + (n-1) \theta H_{1+\theta}^{\uparrow,W}(X|Y) }{1+\theta} + \underline{\xi}(\theta),
\end{eqnarray}
where $\underline{\xi}(\theta)$ is given by \eqref{eq:definition-underline-xi}.
\end{theorem}

\textchange{Finally, to derive an upper bound on $- \log \Pse(M_n)$ in terms of the R\'enyi entropy for transition matrix, 
we substitute the formula for the R\'enyi entropy in Lemma \ref{lemma:multi-terminal-finite-evaluation-two-parameter-conditional-renyi}
to Theorem \ref{theorem:multi-source-converse} for $Q_{Y^n} = P_{Y^n}^{(1+\theta(a(R )))}$. Then, we can derive the following converse bound.}

\begin{theorem}[Converse, Ass. 2] \label{theorem:multi-source-finite-marko-converse-assumptiton-2}
Suppose that transition matrix $W$ satisfies Assumption \ref{assumption-memory-through-Y}. 
Let $R := \frac{1}{n} \log M_n$. For any $H^W(X|Y) < R < H_0^{\uparrow,W}(X|Y)$, we have
\begin{eqnarray}
\lefteqn{ - \log \Pse^{(n)}(M_n) } \\
&\le& \inf_{s > 0 \atop -1 < \tilde{\theta} < \theta(a(R )) } \bigg[
 (n-1)(1+s) \tilde{\theta} \left\{ H_{1+\tilde{\theta},1+\theta(a(R ))}^W(X|Y) - H_{1+(1+s)\tilde{\theta},1+\theta(a(R ))}^W(X|Y) \right\} + \delta_1 \\
 && - (1+s)\log \left( 1 - 2e^{(n-1)[(\theta(a(R ))-\tilde{\theta}) a(R ) - \theta(a(R )) H_{1+\theta(a(R ))}^{\uparrow,W}(X|Y) 
  + \tilde{\theta} H_{1+\tilde{\theta},1+\theta(a(R ))}^W(X|Y) ] + \delta_2} \right)
\bigg] / s,
\end{eqnarray}
where 
$\theta(a) = \theta^\uparrow(a)$ and $a(R ) = a^\uparrow(R )$ 
are the inverse functions defined by \eqref{eq:definition-theta-inverse-markov-optimal-Q} 
and \eqref{eq:definition-a-inverse-markov-optimal-Q} respectively, 
\begin{eqnarray}
\delta_1 &:=& (1+s) \overline{\zeta}(\tilde{\theta}, \theta(a(R ))) - \underline{\zeta}((1+s)\tilde{\theta},\theta(a(R ))) , \\
\delta_2 &:=& \frac{(\theta(a(R )) - \tilde{\theta}) R - (1+\tilde{\theta}) \underline{\zeta}(\theta(a(R )), \theta(a(R ))) + (1+\theta(a(R ))) \overline{\zeta}(\tilde{\theta},\theta(a(R )))}{1+\theta(a(R ))},
\end{eqnarray}
and $\overline{\zeta}(\cdot,\cdot)$ and $\underline{\zeta}(\cdot,\cdot)$ 
are given by \eqref{eq:definition-overline-zeta}-\eqref{eq:definition-underline-zeta-2}.
\end{theorem}
\begin{proof}
We first use \eqref{eq:bound-multi-source-exponential-converse-1} of 
Theorem \ref{theorem:multi-source-converse} for $Q_{Y^n} = P_{Y^n}^{(1+\theta(a(R )))}$
and Lemma \ref{lemma:multi-terminal-finite-evaluation-two-parameter-conditional-renyi}.
Then, we restrict the range of $\tilde{\theta}$ 
as $-1 < \tilde{\theta} < \theta(a(R ))$ and set $\vartheta = \frac{\theta(a(R )) - \tilde{\theta}}{1+\tilde{\theta}}$.
Then, we have the assertion of the theorem.
\end{proof}

\subsection{Second Order} \label{subsection:multi-source-second-order}

By applying the central limit theorem
to Lemma \ref{lemma:multi-source-han-direct} 
(cf.~\cite[Theorem 27.4, Example 27.6]{billingsley-book})
and Lemma \ref{lemma:multi-source-han-converse}
for $Q_Y = P_Y$,
and by using Theorem \ref{theorem:multi-markov-variance},
 we have the following.
\begin{theorem} \label{theorem:source-coding-second-order}
Suppose that transition matrix $W$ on ${\cal X}\times {\cal Y}$ satisfies Assumption \ref{assumption-Y-marginal-markov}.
For arbitrary $\varepsilon \in (0,1)$, we have
\begin{eqnarray}
M(n,\varepsilon)= \bar{M}(n,\varepsilon)+o(\sqrt{n})
= n H^W(X|Y)+ \sqrt{\san{V}^W(X|Y)}\sqrt{n}+o(\sqrt{n}).
\end{eqnarray}
\end{theorem} 
\begin{proof}
The central limit theorem for Markovian process cf.~\cite[Theorem 27.4, Example 27.6]{billingsley-book}
guarantees that the random variable $(-\log P_{X^n|Y^n}(X^n|Y^n) - n H^W(X|Y)) / \sqrt{n}$ 
asymptotically obeys the normal distribution with average $0$ and the variance $\san{V}^W(X|Y)$,
where we use Theorem \ref{theorem:multi-markov-variance} to show that the limit of the variance is given by
$\san{V}^W(X|Y)$. Let $R = \sqrt{\san{V}^W(X|Y)} \Phi^{-1}(1-\varepsilon)$.
Substituting $M = e^{n H^W(X|Y) + \sqrt{n}R}$ and $\gamma = n H^W(X|Y) + \sqrt{n}R - n^{\frac{1}{4}}$ in
Lemma \ref{lemma:multi-source-han-direct}, we have 
\begin{eqnarray} \label{eq:second-order-proof-1}
\lim_{n \to \infty} \barPse^{(n)}\left( e^{n H^W(X|Y) + \sqrt{n}R} \right) \le \varepsilon.
\end{eqnarray}
On the other hand, substituting $M = e^{n H^W(X|Y) + \sqrt{n}R}$ and $\gamma = n H^W(X|Y) + \sqrt{n}R + n^{\frac{1}{4}}$ in
Lemma \ref{lemma:multi-source-han-converse} for $Q_Y = P_Y$, we have
\begin{eqnarray} \label{eq:second-order-proof-2}
\lim_{n \to \infty} \Pse^{(n)}\left( e^{n H^W(X|Y) + \sqrt{n}R} \right) \ge \varepsilon.
\end{eqnarray}
Combining \eqref{eq:second-order-proof-1} and \eqref{eq:second-order-proof-2}, we have the statement of the theorem. 
\end{proof}

From the above theorem,
the (first-order) compression limit of source coding with side-information for  a Markov source 
under Assumption \ref{assumption-Y-marginal-markov} is given by\footnote{Although the compression limit of source coding with side-information for a Markov chain is known more generally \cite{cover:75}, 
we need Assumption \ref{assumption-Y-marginal-markov} to get a single letter characterization.}
\begin{eqnarray}
 \lim_{n\to \infty} \frac{1}{n} \log M(n,\varepsilon) 
&=&  \lim_{n\to \infty} \frac{1}{n} \log \bar{M}(n,\varepsilon) \\
&=& H^W(X|Y)
\end{eqnarray}
for any $\varepsilon \in (0,1)$.
In the next subsections, we consider the asymptotic behavior of the error probability when the rate
is larger than the compression limit $H^W(X|Y)$ in the moderate deviation regime
and the large deviation regime, respectively.

\subsection{Moderate Deviation} \label{subsection:multi-source-mdp}

From Theorem \ref{theorem:multi-source-finite-marko-direct-assumptiton-1} and
Theorem \ref{theorem:multi-source-finite-markov-converse-assumptiotn-1}, we have the following.
\begin{theorem} \label{theorem:multi-moderate-deviation}
Suppose that transition matrix $W$ satisfies Assumption \ref{assumption-Y-marginal-markov}.
For arbitrary $t \in (0,1/2)$ and $\delta > 0$, we have 
\begin{eqnarray}
\lim_{n\to\infty} - \frac{1}{n^{1-2t}} \log \Pse^{(n)}\left(e^{n H^W(X|Y) + n^{1-t}\delta} \right)
&=& \lim_{n\to\infty} - \frac{1}{n^{1-2t}} \log \barPse^{(n)}\left(e^{n H^W(X|Y) + n^{1-t}\delta} \right) \\
&=& \frac{\delta^2}{2 \san{V}^W(X|Y)}.
\end{eqnarray}
\end{theorem}
\begin{proof}
We apply Theorem \ref{theorem:multi-source-finite-marko-direct-assumptiton-1} and
Theorem \ref{theorem:multi-source-finite-markov-converse-assumptiotn-1} to the case with
$R = H^W(X|Y) + n^{-t} \delta$, i.e., $\theta(a(R )) = -n^{-1} \frac{\delta}{\san{V}^W(X|Y)} + o(n^{-t})$.
For the achievability part,
from \eqref{eq:expansion-exponent-around-H} and Theorem \ref{theorem:multi-source-finite-marko-direct-assumptiton-1},
we have
\begin{align}
- \log \Pse^{(n)}\left(M_n \right)
 &\ge \sup_{-1 \le \theta \le 0} \left[ - \theta n R + (n-1) \theta H_{1+\theta}^{\downarrow,W}(X|Y) \right] 
  + \inf_{-1\le\theta\le0} \underline{\delta}(\theta) \\
 &\ge n^{1-2t} \frac{\delta^2}{2 \san{V}^W(X|Y)} + o(n^{1-2t}).
\end{align}
To prove the converse part, we fix arbitrary $s > 0$ and choose $\tilde{\theta}$ to be $-n^{-t} \frac{\delta}{\san{V}^W(X|Y)} + n^{-2t}$.
Then, Theorem \ref{theorem:multi-source-finite-markov-converse-assumptiotn-1} implies that 
\begin{align}
\limsup_{n\to\infty} - \frac{1}{n^{1-2t}} \log \Pse(M_n)
&\le \limsup_{n\to\infty} n^{2t} \frac{1+s}{s} \tilde{\theta}
 \left\{ H_{1+\tilde{\theta}}^{\downarrow,W}(X|Y) - H_{1+(1+s)\tilde{\theta}}^{\downarrow,W}(X|Y) \right\} \\
&= \limsup_{n\to\infty} n^{2t} \frac{1+s}{s} s \tilde{\theta}^2 \frac{d H_{1+\theta}^{\downarrow,W}(X|Y)}{d\theta} \bigg|_{\theta=\tilde{\theta}} \\
&= (1+s) \frac{\delta^2}{2 \san{V}^W(X|Y)}.
\end{align}
\end{proof}

\begin{remark} \label{remark:convergence-on-moderate-deviation}
In the literatures \cite{altug:10,dembo-zeitouni-book}, the moderate deviation results are stated for
$\epsilon_n$ such that $\epsilon_n \to 0$ and $n \epsilon_n^2 \to \infty$ instead of $n^{-t}$ for $t \in (0,1/2)$.
Although the former is slightly more general than the latter, we employ the latter formulation in 
Theorem \ref{theorem:multi-moderate-deviation} 
since the order of convergence is clearer. In fact, $n^{-t}$ in Theorem \ref{theorem:multi-moderate-deviation} can 
be replaced by general $\epsilon_n$ without modifying the argument of the proof.
\end{remark}

\subsection{Large Deviation} \label{subsection:multi-source-large-deviation}

From Theorem \ref{theorem:multi-source-finite-marko-direct-assumptiton-1} and
Theorem \ref{theorem:multi-source-finite-markov-converse-assumptiotn-1}, we have the following.
\begin{theorem} \label{theorem:source-coding-ldp-assumption-1}
Suppose that transition matrix $W$ satisfies Assumption \ref{assumption-Y-marginal-markov}.
For $H^W(X|Y) < R$, we have
\begin{eqnarray} \label{eq:multi-source-large-deviation-achievability}
\liminf_{n\to\infty} - \frac{1}{n} 
\log \barPse^{(n)}\left(e^{nR}\right) 
\ge \sup_{- 1 \le \theta \le 0} [-\theta R + \theta H_{1+\theta}^{\downarrow,W}(X|Y)].
\end{eqnarray}
On the other hand, for $H^W(X|Y) < R < H_0^{\downarrow,W}(X|Y)$, we have
\begin{eqnarray} 
\limsup_{n\to\infty} - \frac{1}{n} \log \Pse^{(n)}\left(e^{nR}\right) 
&\le& - \theta(a(R )) a(R ) + \theta(a(R )) H_{1+\theta(a(R ))}^{\downarrow,W}(X|Y) \label{eq:multi-source-ldp-converse-assumption-1} \\
&=& \sup_{-1 < \theta \le 0} \frac{-\theta R + \theta H_{1+\theta}^{\downarrow,W}(X|Y)}{1+\theta}.
 \label{eq:multi-source-ldp-converse-assumption-1-2}
\end{eqnarray}
\end{theorem} 
\begin{proof}
The achievability bound \eqref{eq:multi-source-large-deviation-achievability} follows from 
Theorem \ref{theorem:multi-source-finite-marko-direct-assumptiton-1}. 
The converse part \eqref{eq:multi-source-ldp-converse-assumption-1} is proved 
from Theorem \ref{theorem:multi-source-finite-markov-converse-assumptiotn-1} as follows.
We first fix $s > 0$ and $-1 < \tilde{\theta} < \theta(a(R ))$. Then, Theorem \ref{theorem:multi-source-finite-markov-converse-assumptiotn-1} implies 
\begin{align}
\limsup_{n\to \infty} - \frac{1}{n} \log \Pse^{(n)}\left(e^{nR}\right)
 \le \frac{1+s}{s} \tilde{\theta} \left\{ H_{1+\tilde{\theta}}^{\downarrow,W}(X|Y) -  H_{1+(1+s)\tilde{\theta}}^{\downarrow,W}(X|Y)\right\}.
\end{align}
By taking the limit $s \to 0$ and $\tilde{\theta} \to \theta(a(R ))$, we have
\begin{align}
& \frac{1+s}{s} \tilde{\theta} \left\{ H_{1+\tilde{\theta}}^{\downarrow,W}(X|Y) -  H_{1+(1+s)\tilde{\theta}}^{\downarrow,W}(X|Y) \right\} \\
&=  \frac{1}{s}\left(\tilde{\theta} H_{1+\tilde{\theta}}^{\downarrow,W}(X|Y)
- (1+s) \tilde{\theta} H_{1+(1+s)\tilde{\theta}}^{\downarrow,W}(X|Y)  \right) 
 + \tilde{\theta} H_{1+\tilde{\theta}}^{\downarrow,W}(X|Y) \\
&\to - \tilde{\theta} \frac{d[\theta H_{1+\theta}^{\downarrow,W}(X|Y)]}{d\theta} \bigg|_{\theta=\tilde{\theta}} 
 + \tilde{\theta} H_{1+\tilde{\theta}}^{\downarrow,W}(X|Y) ~~~(\mbox{as } s \to 0) \\
&\to - \theta(a(R )) \frac{d[\theta H_{1+\theta}^{\downarrow,W}(X|Y)]}{d\theta} \bigg|_{\theta=\theta(a(R ))}
 + \theta(a(R )) H_{1+\theta(a(R ))}^{\downarrow,W}(X|Y)~~~(\mbox{as } \tilde{\theta} \to \theta(a(R ))) \\
&= - \theta(a(R )) a(R )   + \theta(a(R )) H_{1+\theta(a(R ))}^{\downarrow,W}(X|Y).
\end{align}
Thus, \eqref{eq:multi-source-ldp-converse-assumption-1} is proved. 
The alternative expression \eqref{eq:multi-source-ldp-converse-assumption-1-2} is derived via Lemma \ref{lemma:legendra-transform-2}.
\end{proof}

Under Assumption \ref{assumption-memory-through-Y}, from 
Theorem \ref{theorem:multi-source-finite-marko-direct-assumptiton-2} and 
Theorem \ref{theorem:multi-source-finite-marko-converse-assumptiton-2}, we have the following tighter bound.
\begin{theorem} \label{theorem:source-coding-ldp-assumption-2}
Suppose that transition matrix $W$ satisfies Assumption \ref{assumption-memory-through-Y}. 
For $H^W(X|Y) < R$, we have
\begin{eqnarray} \label{eq:multi-source-ldp-assumption-2-direct}
\liminf_{n\to\infty} 
- \frac{1}{n} \log \barPse^{(n)}\left(e^{nR}\right) 
 \ge \sup_{-\frac{1}{2} \le \theta \le 0} \frac{- \theta R + \theta H_{1+\theta}^{\uparrow,W}(X|Y)}{1+\theta}.
\end{eqnarray}
On the other hand, for $H^W(X|Y) < R < H_0^{\uparrow,W}(X|Y)$, we have
\begin{eqnarray} \label{eq:multi-source-ldp-assumption-2-converse-statement}
\limsup_{n\to\infty} - \frac{1}{n} \log \Pse^{(n)}\left(e^{nR}\right) 
 &\le& - \theta(a(R )) a(R ) + \theta(a(R )) H_{1+\theta(a(R ))}^{\uparrow,W}(X|Y) \label{eq:multi-source-ldp-converse-assumption-2} \\
 &=& \sup_{-1 < \theta \le 0} \frac{- \theta R + \theta H_{1+\theta}^{\uparrow,W}(X|Y)}{1+\theta}.
  \label{eq:multi-source-ldp-assumption-2-converse-statement-2}
\end{eqnarray}
\end{theorem}
\begin{proof}
The achievability bound \eqref{eq:multi-source-ldp-assumption-2-direct} follows from 
Theorem \ref{theorem:multi-source-finite-marko-direct-assumptiton-2}. 
The converse part \eqref{eq:multi-source-ldp-assumption-2-converse-statement} is proved 
from Theorem \ref{theorem:multi-source-finite-marko-converse-assumptiton-2} as follows.
We first fix $s > 0$ and $-1 < \tilde{\theta} < \theta(a(R ))$. Then, Theorem \ref{theorem:multi-source-finite-marko-converse-assumptiton-2} implies 
\begin{align}
\limsup_{n\to \infty} - \frac{1}{n} \log \Pse^{(n)}\left(e^{nR}\right)
 \le \frac{1+s}{s} \tilde{\theta} \left\{ H_{1+\tilde{\theta},1+\theta(a(R ))}^W(X|Y) 
  -  H_{1+(1+s)\tilde{\theta}, 1+\theta(a(R ))}^W(X|Y)\right\}.
\end{align}
By taking the limit $s \to 0$ and $\tilde{\theta} \to \theta(a(R ))$, we have
\begin{align}
& \frac{1+s}{s} \tilde{\theta} \left\{ H_{1+\tilde{\theta},1+\theta(a(R ))}^W(X|Y) 
 -  H_{1+(1+s)\tilde{\theta},1+\theta(a(R ))}^W(X|Y) \right\} \\
&=  \frac{1}{s}\left(\tilde{\theta} H_{1+\tilde{\theta},1+\theta(a(R ))}^W(X|Y)
- (1+s) \tilde{\theta} H_{1+(1+s)\tilde{\theta}, 1+\theta(a(R ))}^W(X|Y)  \right) 
 + \tilde{\theta} H_{1+\tilde{\theta},1+\theta(a(R ))}^W(X|Y) \\
&\to - \tilde{\theta} \frac{d[\theta H_{1+\theta,1+\theta(a(R ))}^W(X|Y)]}{d\theta} \bigg|_{\theta=\tilde{\theta}} 
 + \tilde{\theta} H_{1+\tilde{\theta},1+\theta(a(R ))}^W(X|Y) ~~~(\mbox{as } s \to 0) \\
&\to - \theta(a(R )) \frac{d[\theta H_{1+\theta,1+\theta(a(R ))}^W(X|Y)]}{d\theta} \bigg|_{\theta=\theta(a(R ))}
 + \theta(a(R )) H_{1+\theta(a(R ))}^{\uparrow,W}(X|Y)~~~(\mbox{as } \tilde{\theta} \to \theta(a(R ))) \\
&= - \theta(a(R )) a(R )   + \theta(a(R )) H_{1+\theta(a(R ))}^{\uparrow,W}(X|Y).
\end{align}
Thus, \eqref{eq:multi-source-ldp-assumption-2-converse-statement} is proved. 
The alternative expression \eqref{eq:multi-source-ldp-assumption-2-converse-statement-2} is derived via Lemma \ref{lemma:legendra-transform-2}.
\end{proof}

\begin{remark}
For $R \le R_{\rom{cr}}$, where (cf.~\eqref{eq:definition-R-a-optimal-q-markov} for the definition of $R(a)$)
\begin{eqnarray}
R_{\rom{cr}} := R\left( \frac{d[\theta H_{1+\theta}^{\uparrow,W}(X|Y)]}{d\theta} \bigg|_{\theta = -\frac{1}{2}} \right)
\end{eqnarray}
is the critical rate, we can rewrite the lower bound in \eqref{eq:multi-source-ldp-assumption-2-direct} as 
(cf.~Lemma \ref{lemma:legendra-transform-2})
\begin{eqnarray}
\sup_{- \frac{1}{2} \le \theta \le 0} \frac{- \theta R + \theta H_{1+\theta}^{\uparrow,W}(X|Y)}{1+\theta}
= - \theta(a(R )) a(R ) + \theta(a(R )) H_{1+\theta(a(R ))}^{\uparrow,W}(X|Y).
\end{eqnarray}
Thus, the lower bound and the upper bound coincide up to the critical rate.
\end{remark}

\begin{remark} \label{remark:singleton}
When ${\cal Y}$ is singleton, from Theorem \ref{theorem:tight-finite-bound-singleton-Y} and 
a special case of \eqref{eq:multi-source-ldp-converse-assumption-1},
we can derive 
\begin{eqnarray}
\liminf_{n\to\infty} - \frac{1}{n} \log \barPse^{(n)}\left(e^{nR}\right) 
\ge \sup_{- 1 \le \theta \le 0} \frac{-\theta R + \theta H_{1+\theta}^{W}(X)}{1+\theta}
\end{eqnarray}
and 
\begin{eqnarray}
\limsup_{n\to\infty} - \frac{1}{n} \log \Pse^{(n)}\left(e^{nR}\right) 
&\le& \sup_{-1 < \theta \le 0} \frac{-\theta R + \theta H_{1+\theta}^{W}(X)}{1+\theta}.
\end{eqnarray}
for $H^W(X) < R < H_0^W(X)$. Thus, we can recover the results in \cite{davisson:81, vasek:80} by our approach.
\end{remark}

\subsection{Numerical Example} \label{subsection:single-source-numerical}

\begin{figure}[t]
\centering
\includegraphics[width=0.5\linewidth]{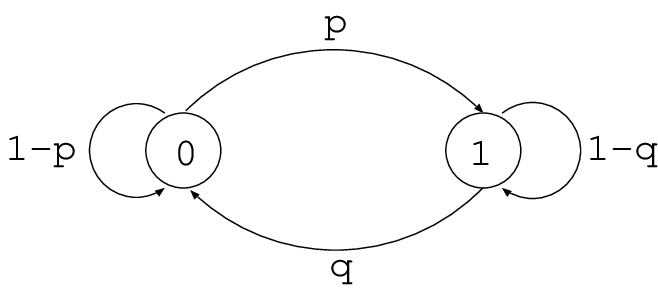}
\caption{The description of the transition matrix.}
\label{Fig:transition}
\end{figure}

In this section, 
to demonstrate the advantage of our finite-length bound,
we numerically evaluate the achievability bound in Theorem \ref{theorem:tight-finite-bound-singleton-Y}
and a special case of the converse bound in Theorem \ref{theorem:multi-source-finite-markov-converse-assumptiotn-1}
for singleton ${\cal Y}$.
Thanks to the criterion (C2), our numerical calculation shows that 
our upper finite-length bounds 
is very close to 
our lower finite-length bounds when the size $n$ is sufficiently large.
Thanks to the criterion (C1), we could calculate both bounds with huge size $n=1\times 10^5$
because the calculation complexity behaves as $O(1)$.

We consider a binary transition matrix $W$ given by Fig.~\ref{Fig:transition}, i.e.,
\begin{eqnarray}
W = \left[
\begin{array}{cc}
1-p & q \\
p & 1-q
\end{array}
\right].
\end{eqnarray}
In this case, the stationary distribution is
\begin{eqnarray}
\tilde{P}(0) &=& \frac{q}{p+q}, \\
\tilde{P}(1) &=& \frac{p}{p+q}.
\end{eqnarray}
The entropy is
\begin{eqnarray}
H^W(X) = \frac{q}{p+q} h(p ) + \frac{p}{p+q} h(q ),
\end{eqnarray}
where $h(\cdot)$ is the binary entropy function.
The tilted transition matrix is 
\begin{eqnarray}
W_\theta = \left[
\begin{array}{cc}
(1-p)^{1+\theta} & q^{1+\theta} \\
p^{1+\theta} & (1-q)^{1+\theta}
\end{array}
\right].
\end{eqnarray}
The Perron-Frobenius eigenvalue is
\begin{eqnarray}
\lambda_\theta = \frac{(1-p)^{1+\theta} + (1-q)^{1+\theta} + \sqrt{\{(1-p)^{1+\theta} - (1-q)^{1+\theta} \}^2 + 4 p^{1+\theta} q^{1+\theta}} }{2}
\end{eqnarray}
and its normalized eigenvector is
\begin{eqnarray}
\tilde{P}_\theta(0) &=& \frac{q^{1+\theta}}{\lambda_\theta - (1-p)^{1+\theta} + q^{1+\theta}}, \\
\tilde{P}_\theta(1) &=& \frac{\lambda_\theta - (1-p)^{1+\theta}}{\lambda_\theta - (1-p)^{1+\theta} + q^{1+\theta}}.
\end{eqnarray}
The normalized eigenvector of $W_\rho^T$ is also given by
\begin{eqnarray}
\hat{P}_\theta(0) &=& \frac{p^{1+\theta}}{\lambda_\theta - (1-p)^{1+\theta} + p^{1+\theta}}, \\
\hat{P}_\theta(1) &=& \frac{\lambda_\theta - (1-p)^{1+\theta}}{\lambda_\theta - (1-p)^{1+\theta} + p^{1+\theta}}.
\end{eqnarray}
From these calculations, we can evaluate the bounds in Theorem \ref{theorem:tight-finite-bound-singleton-Y}
and  Theorem \ref{theorem:multi-source-finite-markov-converse-assumptiotn-1}. For $p = 0.1$, $q=0.2$, the bounds are plotted 
in Fig.~\ref{Fig:Comparison-single-source-fixed-epsilon} for fixed error probability $\varepsilon = 10^{-3}$.
Although there is a gap between the achievability bound and the converse bound for rather small $n$, the gap is less than
approximately $5$\% of the entropy rate for $n$ larger than $10000$. 
We also plotted the bounds in Fig.~\ref{Fig:Comparison-single-source-fixed-n} for fixed block length $n=10000$
and varying $\varepsilon$. The gap between the achievability bound and the converse bound remains 
approximately $5$\% of the entropy rate even for $\varepsilon$ as small as $10^{-10}$.

\begin{figure}[t]
\centering
\includegraphics[width=0.5\linewidth]{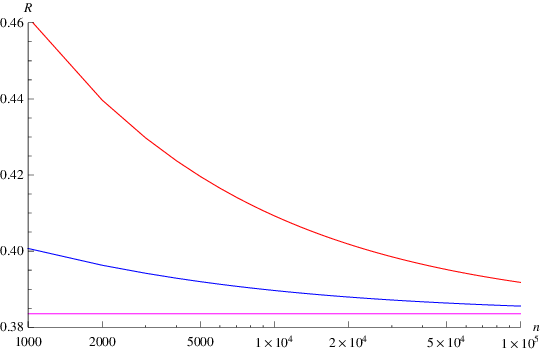}
\caption{A comparison of the bounds for $p=0.1$, $q=0.2$, and $\varepsilon = 10^{-3}$.
The horizontal axis is the block length $n$
and the vertical axis is the rate $R$ (nats).
The red curve is the achievability bound in Theorem \ref{theorem:tight-finite-bound-singleton-Y}.
The blue curve is the converse bound in Theorem \ref{theorem:multi-source-finite-markov-converse-assumptiotn-1}.
The purple line is the entropy $H^W(X)$.
}
\label{Fig:Comparison-single-source-fixed-epsilon}
\end{figure}
\begin{figure}[t]
\centering
\includegraphics[width=0.5\linewidth]{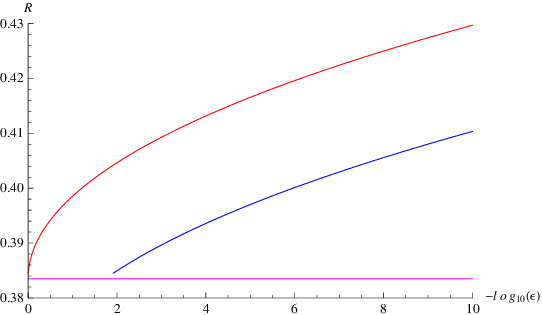}
\caption{A comparison of the bounds for $p=0.1$, $q=0.2$, and $n=10000$.
The horizontal axis is $-\log_{10}(\varepsilon)$, and 
the vertical axis is the rate $R$ (nats).
The red curve is the achievability bound in Theorem \ref{theorem:tight-finite-bound-singleton-Y}.
The blue curve is the converse bound in Theorem \ref{theorem:multi-source-finite-markov-converse-assumptiotn-1}.
The purple line is the entropy $H^W(X)$.
}
\label{Fig:Comparison-single-source-fixed-n}
\end{figure}

%% file: Channel.tex
\section{Channel Coding} \label{section:channel}

In this section, we investigate the channel coding with 
a {\it conditional additive channel}.
The former part of this section discusses general properties of 
the channel coding with a {\it conditional additive channel}.
The latter part of this section discusses 
properties of the channel coding
when the conditional additive noise of the channel is Markovian.
We start this section by showing the problem setting in Section \ref{subsection:channel-problem-formulation} by introducing a {\it conditional additive channel}.
Section \ref{subsection:conversion-regular-g-additive} gives a canonical method to convert a regular channel to a conditional additive channel.
Section \ref{subsection:conversion-bpsk-awgn} gives a method to convert a BPSK-AWGN channel to a conditional additive channel.
Then, we show some single-shot achievability bounds in Section \ref{subsection:channel-single-achievability},
and single-shot converse bounds in Section \ref{subsection:channel-converse}.

As the latter part, 
we derive finite-length bounds for the Markov noise channel in Section \ref{subsection:channel-finite-markov}.
Then, in Sections \ref{subsection:channel-ldp} and \ref{subsection:channel-mdp},
we show the asymptotic characterization for the large deviation regime and the moderate deviation regime
by using those finite-length bounds. We also derive the second order rate in Section \ref{subsection:multi-random-second-order}.


The results shown in this section for the Markovian conditional additive noise are summarized in Table \ref{table:summary:channel}.
The checkmarks $\checkmark$ indicate that the tight asymptotic bounds (large deviation, moderate deviation, and second order)
 can be obtained from those bounds.
The marks $\checkmark^*$ indicate that the large deviation bound can be derived up to the critical rate. 
The computational complexity ``Tail" indicates that the computational complexities of those bounds 
depend on the computational complexities of tail probabilities.
It should be noted that Theorem \ref{theorem:channel-finite-markov-converse-assumptiotn-1}
is derived from a special case ($Q_{Y} = P_Y$) of Theorem \ref{theorem:channel-exponential-converse}.
The asymptotically optimal choice is $Q_Y = P_Y^{(1+\theta)}$.
Under Assumption \ref{assumption-Y-marginal-markov}, we can derive the bound of the Markov case only for that
special choice of $Q_Y$, while under Assumption \ref{assumption-memory-through-Y}, we can derive the bound of
the Markov case for the optimal choice of $Q_Y$.
Furthermore, Theorem \ref{theorem:channel-finite-markov-converse-assumptiotn-1} is not asymptotically tight in
the large deviation regime in general, but it is tight if ${\cal Y}$ is singleton, i.e., the channel is additive.
It should be also noted that Theorem \ref{theorem:channel-finite-marko-converse-assumptiton-2} does not imply
Theorem \ref{theorem:channel-finite-markov-converse-assumptiotn-1} even for the additive channel case since
Assumption \ref{assumption-memory-through-Y} restricts the structure of transition matrices even when 
${\cal Y}$ is singleton.

\begin{table}[htbp]
\begin{center}
\caption{Summary of the finite-length bounds for channel coding.}
\label{table:summary:channel}
{\renewcommand{\arraystretch}{1.4}
\begin{tabular}{|c|c|c|c|c|c|c|c|} \hline
  \multirow{2}{*}{Ach./Conv.}  
  & \multirow{2}{*}{Markov} & \multirow{2}{*}{Single Shot} &  \multirow{2}{*}{$\Pce$/$\barPce$} &  \multirow{2}{*}{Complexity} 
  & Large & Moderate & Second \\
  & & & & & Deviation & Deviation & Order \\ \hline
\multirow{4}{*}{Achievability} 
 & Theorem \ref{theorem:channel-finite-marko-direct-assumptiton-1} (Ass.~1) & Lemma \ref{lemma:channel-loose-bound} & $\barPce$ & $O(1)$ &    & \checkmark &  \\ \cline{2-8}
 & Theorem \ref{theorem:channel-finite-marko-direct-assumptiton-2} (Ass.~2) & Lemma \ref{lemma:channel-gallager-bound}  & $\barPce$ & $O(1)$ & $\checkmark^*$ & \checkmark &  \\ \cline{2-8}
 & Theorem \ref{theorem:channel-finite-markov-additive} (Additive) & Lemma \ref{lemma:channel-additive-channel-bound}  & $\barPce$ & $O(1)$ & $\checkmark^*$ & \checkmark &  \\ \cline{2-8}
 & \multicolumn{2}{|c|}{Lemma \ref{lemma:channel-spectrum-direct}}   &  $\barPce$ & Tail &  & \checkmark & \checkmark \\ 
 \hline
\multirow{4}{*}{Converse} 
 & Theorem \ref{theorem:channel-finite-markov-converse-assumptiotn-1} (Ass.~1)  & (Theorem \ref{theorem:channel-exponential-converse})  & $\Pce$ & $O(1)$ &  & \checkmark &  \\ \cline{2-8} 
 & Theorem \ref{theorem:channel-finite-marko-converse-assumptiton-2} (Ass.~2)  & Theorem \ref{theorem:channel-exponential-converse}  & $\Pce$ & $O(1)$ & $\checkmark^*$ & \checkmark &  \\ \cline{2-8} 
 & Theorem \ref{theorem:channel-finite-markov-converse-assumptiotn-1} (Additive)  & (Theorem \ref{theorem:channel-exponential-converse})  & $\Pce$ & $O(1)$ & $\checkmark^*$ & \checkmark &  \\ \cline{2-8} 
 & \multicolumn{2}{|c|}{Lemma \ref{lemma:channel-converse-general-additive-specialized}}   &  $\Pce$ & Tail &  & \checkmark & \checkmark \\ 
\hline 
\end{tabular}}
\end{center}
\end{table}

\subsection{Formulation for conditional additive channel} \label{subsection:channel-problem-formulation}
\subsubsection{Single-shot case}
We first present the problem formulation by the single shot setting.
For a channel $P_{B|A}(b|a)$ with input alphabet ${\cal A}$
and output alphabet ${\cal B}$, a channel code $\Psi = (\san{e},\san{d})$ consists of one encoder $\san{e}: \{1,\ldots, M\} \to {\cal A}$ and
one decoder $\san{d}:{\cal B} \to \{1,\ldots,M\}$. The average decoding error probability is defined by
\begin{eqnarray}
\Pce[\Psi] := \sum_{m=1}^M \frac{1}{M} P_{B|A}(\{b:\san{d}(b) \neq m \}|\san{e}(m)).\label{10-20-1}
\end{eqnarray}
For notational convenience, we introduce the error probability under the condition that the message 
size is $M$:
\begin{eqnarray}
\Pce(M) := \inf_{\Psi} \Pce[\Psi].\label{10-20-2}
\end{eqnarray}

Assume that the input alphabet ${\cal A}$ is the same set as the output alphabet ${\cal B}$ and they equals an additive group ${\cal X}$.
When the transition matrix $P_{B|A}(b|a)$ is given as $P_X(b-a)$
by using a distribution $P_X$ on ${\cal X}$,
the channel is called additive.
 
To extend the concept of additive channel,
we consider the case when 
the input alphabet ${\cal A}$ is an additive group ${\cal X}$
and the output alphabet ${\cal B}$ is the product set ${\cal X}\times {\cal Y}$.
When the transition matrix $P_{B|A}(x,y|a)$ is given as $P_{XY}(x-a,y)$
by using a distribution $P_{XY}$ on ${\cal X}\times {\cal Y}$,
the channel is called conditional additive.
In this paper, we are exclusively interested in the conditional additive channel.
As explained in Subsection \ref{subsection:conversion-regular-g-additive},
a channel is a conditional additive channel
if and only if 
it is a regular channel in the sense of \cite{delsarte:82}.
When we need to explicitly express the underlying distribution of the noise,
we denote the average decoding error probability
by $\Pce[\Psi|P_{XY}]$.

\subsubsection{$n$-fold extension}
When we consider $n$-fold extension, the channel code is denoted with subscript $n$ such as $\Psi_n = (\san{e}_n,\san{d}_n)$. 
The error probabilities given in \eqref{10-20-1} and \eqref{10-20-2}
are written with the superscript $(n)$
as $\Pce^{(n)}[\Psi_n]$ and $\Pce^{(n)}(M_n)$, respectively.
Instead of evaluating the error 
probability $\Pce^{(n)}(M_n)$ for given $M_n$,
we are also interested in evaluating 
\begin{eqnarray}
M(n,\varepsilon) := \sup\left\{ M_n : \Pce^{(n)}(M_n) \le \varepsilon \right\}
\end{eqnarray}
for given $0 \le \varepsilon \le 1$.

When the channel is given as a conditional distribution,
the channel is given by
\begin{eqnarray}
P_{B^n|A^n}(x^n,y^n|a^n) = P_{X^n Y^n}(x^n-a^n, y^n),\label{10-2-2}
\end{eqnarray}
where 
$P_{X^nY^n}$ is a noise distribution on ${\cal X}^n\times {\cal Y}^n$. 

For the code construction, we investigate the linear code.
For an $(n,k)$ linear code ${\cal C}_n \subset {\cal A}^n$, there exists a parity check matrix
$f_n:{\cal A}^n \to {\cal A}^{n-k}$ such that the kernel of $f_n$ is ${\cal C}_n$.
That is, given a parity check matrix $f_n:{\cal A}^n \to {\cal A}^{n-k}$,
we define the encoder $I_{\rom{Ker}(f_n)}:{\cal C}_n \to {\cal A}^n$
as the imbedding of the kernel $\rom{Ker}(f_n)$.
Then, using the decoder $\san{d}_{f_n} := \argmin_{\san{d}} \Pce[(I_{\rom{Ker}(f_n)},\san{d})]$,
we define $\Psi(f_n) = (I_{\rom{Ker}(f_n)}, \san{d}_{f_n})$.

Here, we employ a randomized choice of a parity check matrix.
In particular, instead of a two-universal hash function,
we focus on liner two-universal hash functions, 
because the linearity is required in the above relation with source coding.
So, denoting the set of linear two-universal hash functions from 
${\cal A}^n$ to ${\cal A}^{n-k}$ by ${\cal F}_l$,
we introduce the quantity:
\begin{eqnarray}
\barPce(n,k) := 
\sup_{F_n \in {\cal F}_l}  \mathbb{E}_{F_n}
\left[ \Pce^{(n)} [\Psi(F_n)] \right].
\end{eqnarray}
Taking the infimum over all linear codes associated with $F_n$
(cf.~\eqref{eq:definition-Pe-bar-source-side-info}),
we obviously have
\begin{eqnarray}
\Pce^{(n)}(|{\cal A}|^k) \le \barPce(n,k).
\end{eqnarray}
When we consider the error probability for conditionally additive channels, we use notation 
$\barPce(n,k|P_{XY})$ so that the underlying distribution of the noise is explicit.
We are also interested in characterizing 
\begin{eqnarray}
k(n,\varepsilon) := \sup\left\{ k : \barPce(n,k) \le \varepsilon \right\}
\end{eqnarray}
for given $0 \le \varepsilon \le 1$.

\subsection{Conversion from regular channel to conditional additive channel} \label{subsection:conversion-regular-g-additive}
This subsection shows that a channel is a regular channel 
in the sense of \cite{delsarte:82} if and only if it is conditional additive.
Then, we see that a binary erasure symmetric channel is an example of a regular channel.

We assume that the input alphabet ${\cal A}$ has an additive group structure.
Let $P_{\tilde{X}}$ be a distribution on the output alphabet ${\cal B}$. 
Let $\pi_a$ be a representation 
of the group ${\cal A}$ on ${\cal B}$, and let $G = \{ \pi_a : a \in {\cal A} \}$.
A regular channel \cite{delsarte:82} is defined by
\begin{eqnarray}
P_{B|A}(b|a) = P_{\tilde{X}}(\pi_a(b)).
\end{eqnarray}
The group action induces orbit 
\begin{eqnarray}
\rom{Orb}(b) := \{ \pi_a(b) : a \in {\cal A}\}.
\end{eqnarray}
The set of all orbits constitute a disjoint partition of ${\cal B}$.
A set of representatives of the orbits is denoted by $\bar{{\cal B}}$,
and let $\varpi:{\cal B} \to \bar{{\cal B}}$ be the map to the representatives. 

\begin{example}[Binary Erasure Symmetric Channel] \label{example:BES}
Let ${\cal A} = \{0,1\}$, ${\cal B} = \{ 0,1,? \}$, and 
\begin{eqnarray}
P_{\tilde{X}}(b) = \left\{ \begin{array}{ll}
1 - p - p^\prime & \mbox{if } b = 0 \\
p & \mbox{if } b = 1 \\
p^\prime & \mbox{if } b = ?
\end{array} \right..
\end{eqnarray} 
Then, let
\begin{eqnarray}
\pi_0 = \left[ \begin{array}{ccc}
0 & 1 & ? \\
0 & 1 & ?
\end{array} \right],~~
\pi_1 = \left[ \begin{array}{ccc}
0 & 1 & ? \\
1 & 0 & ?
\end{array} \right].
\end{eqnarray}
The channel defined in this way is a regular channel (see Fig.~\ref{Fig:BSE-Channel}).
In this case, there are two orbits: $\{0,1\}$ and $\{ ? \}$.
\end{example}

\begin{figure}[t]
\centering
\includegraphics[width=0.3\linewidth]{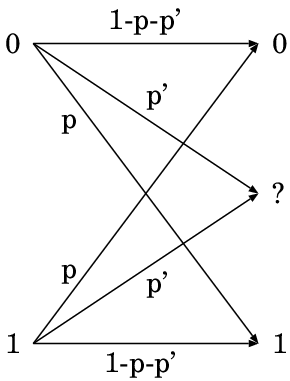}
\caption{The binary erasure symmetric channel.}
\label{Fig:BSE-Channel}
\end{figure}

Let ${\cal B} = {\cal X} \times {\cal Y}$ and $P_{\tilde{X}} = P_{XY}$ for some joint distribution on ${\cal X} \times {\cal Y}$. 
Now, we consider a conditional additive channel, whose 
transition matrix $P_{B|A}(x,y|a)$ is given as $P_{XY}(x-a,y)$.
When the group action is given by $\pi_a(x,y) = (x-a,y)$, 
the above conditional additive channel
given as a regular channel.
In this case, there are $|{\cal Y}|$ orbits
and the size of each orbit is $|{\cal X}|$ respectively.
This fact shows that 
any conditional additive channel is written as a regular channel.

Conversely, we show that any regular channel is written as a conditional additive channel.
For this purpose, we convert a regular channel to a conditional additive channel as follows.


We first explain the construction for single shot channel. 
For random variable $\tilde{X} \sim P_{\tilde{X}}$, let ${\cal Y} = \bar{{\cal B}}$ and $Y = \varpi(\tilde{X})$ be the random variable
describing the representatives of the orbits. For each orbit $\rom{Orb}(y)$, fix an element $0_y \in \rom{Orb}(y)$.
Then, let 
\begin{eqnarray}
\rom{Stb}(0_y) := \{ a \in {\cal A} : \pi_a(0_y) = 0_y\}
\end{eqnarray}
be the stabilizer subgroup of $0_y$\footnote{Since ${\cal A}$ is an Abelian group, the stabilizer
group actually does not depend on the choice $0_y \in \rom{Orb}(y)$.}. Let 
$\overline{{\cal A} / \rom{Stb}(0_y)}$ be a set of coset representatives of the coset ${\cal A} / \rom{Stb}(0_y)$, and let
\begin{eqnarray}
\vartheta_y: {\cal A} \to \overline{{\cal A} / \rom{Stb}(0_y)}
\end{eqnarray}
be the map to the coset representatives. Then, we can define the bijective map
\begin{eqnarray}
\iota_y: \rom{Orb}(y) \ni \pi_{-\bar{a}}(0_y) \mapsto \bar{a} \in \overline{{\cal A} / \rom{Stb}(0_y)}.
\end{eqnarray}
Let ${\cal X} = {\cal A}$ and $P_{X|Y}(\cdot|y)$ be the distribution on ${\cal A}$ defined by
\begin{eqnarray} \label{eq:definition-of-x-given-y-general-additive}
P_{X|Y}(x|y) := \frac{P_{\tilde{X}}(\iota_y^{-1}(\vartheta_y(x)))}{P_{\tilde{X}}(\rom{Orb}(y))} \frac{1}{|\rom{Stb}(0_y)|}.
\end{eqnarray}
When the output from the real channel is $b$, the output from the virtual channel is
defined by 
\begin{eqnarray} \label{eq:definition-virtual-channel}
(\iota_{\varpi(b)}(b) + A^\prime, \varpi(b))
\end{eqnarray}
where $A^\prime$ is randomly chosen from $\rom{Stb}(0_{\varpi(b)})$.
\begin{theorem}
The virtual channel defined by \eqref{eq:definition-virtual-channel} is the conditional additive channel 
such that the output is given by $(a+X,Y)$ for $(X,Y) \sim P_{XY}$, where $P_{XY}$ is defined from $P_Y$ and 
$P_{X|Y}$ of \eqref{eq:definition-of-x-given-y-general-additive}.
\end{theorem}
\begin{proof}
When the input to the real channel is $a$, note that the output can be written as $\pi_{-a}(\tilde{X})$, where $\tilde{X} \sim P_{\tilde{X}}$.
By noting that $Y = \varpi(\pi_{-a}(\tilde{X})) \sim P_Y$, the output of the virtual channel is written as
\begin{eqnarray}
(\iota_{Y}(\pi_{-a}(\tilde{X})) + A^\prime, Y) 
&=& (\iota_Y(\pi_{-\vartheta_Y(a)}(\tilde{X})) + A^\prime, Y) \\
&=& (\vartheta_Y(\vartheta_Y(a) + \iota_Y(\tilde{X})) + A^\prime, Y) \label{eq:regular-conversion-proof-1} \\
&=& (a + \iota_Y(\tilde{X}) + A^{\prime\prime}, Y), \label{eq:regular-conversion-proof-2}
\end{eqnarray}
where \eqref{eq:regular-conversion-proof-1} follows from the fact that 
\begin{eqnarray}
\pi_{-\vartheta_Y(a)}(\tilde{X}) 
&=& \pi_{-\vartheta_Y(a)}( \pi_{-\iota_Y(\tilde{X})}(0_Y)) \\
&=& \pi_{- \vartheta_Y(a) - \iota_Y(\tilde{X})}(0_Y),
\end{eqnarray}
and we set $A^{\prime\prime} = \vartheta_Y(\vartheta_Y(a) + \iota_Y(\tilde{X})) - a - \iota_Y(\tilde{X}) + A^\prime$ in \eqref{eq:regular-conversion-proof-2}.
Since $A^{\prime\prime}$ is the uniform random variable on $\rom{Stb}(0_Y)$, the joint distribution of 
$(\iota_Y(\tilde{X}) + A^{\prime\prime}, Y)$ is $P_{XY}$. Thus, we have the statement of the theorem. 
\end{proof}

\begin{example}[Binary Erasure Symmetric Channel Revisited]
We convert the regular channel of Example \ref{example:BES} to a conditional additive channel.
Let us  label the orbit $\{ 0,1\}$ as $y=0$ and $\{ ? \}$ as $y=1$.
Let $0_0 = 0$ and $0_? = ?$. Then, $\rom{Stb}(0_0) = \{0\}$ and $\rom{Stb}(0_1) = \{0,1\}$.
The map $\vartheta_0$ is the identity map, and $\vartheta_1$ is the trivial map given by $\vartheta_1(a) = 0$.
The map $\iota_y$ is given by $\iota_0(0) = 0$, $\iota_0(1) = 1$, and $\iota_1(?) = 0$. The distribution $P_Y$
is given by $P_Y(0) = 1 - p^\prime$ and $P_Y(1) = p^\prime$. The conditional distribution $P_{X|Y}$ is given by
\begin{eqnarray}
P_{X|Y}(x|0) = \left\{
\begin{array}{ll}
\frac{1-p-p^\prime}{1-p^\prime} & \mbox{if } x = 0 \\
\frac{p}{1-p^\prime} & \mbox{if } x = 0
\end{array}
\right.
\end{eqnarray} 
and $P_{X|Y}(0|1) = P_{X|Y}(1|1) = \frac{1}{2}$.
\end{example}

When we consider $n$th extension, a channel is given by
\begin{eqnarray}
P_{B^n|A^n}(b^n|a^n) = P_{\tilde{X}^n}(\pi_{a^n}(b^n)),
\end{eqnarray}
where $n$th extension of the group action is defined by 
$\pi_{a^n}(b^n) = (\pi_{a_1}(b_1),\ldots,\pi_{a_n}(b_n))$.

Similarly, for $n$-fold extension, we can also construct the virtual conditional additive channel.
More precisely, for $\tilde{X}^n \sim P_{\tilde{X}^n}$, we set $Y^n = \varpi(\tilde{X}^n) = (\varpi(\tilde{X}_1),\ldots,\varpi(\tilde{X}_n))$
and 
\begin{eqnarray}
P_{X^n|Y^n}(x^n|y^n) := \frac{P_{\tilde{X}^n}(\iota_{y^n}^{-1}(\vartheta_{y^n}(x^n)))}{P_{\tilde{X}^n}(\rom{Orb}(y^n))} \frac{1}{|\rom{Stb}(0_{y^n})|},
\end{eqnarray} 
where 
\begin{eqnarray}
\rom{Orb}(y^n) &:=& \rom{Orb}(y_1)\times \cdots \times \rom{Orb}(y_n), \\
\vartheta_{y^n}(x^n) &:=& (\vartheta_{y_1}(x_1),\ldots,\vartheta_{y_n}(x_n)), \\
\iota_{y^n}(b^n) &:=& (\iota_{y_1}(b_1),\ldots,\iota_{y_n}(b_n)), \\
\rom{Stb}(0_{y^n}) &:=& \rom{Stb}(0_{y_1}) \times \cdots \times \rom{Stb}(0_{y_n}).
\end{eqnarray}

Since the conversion to the virtual channel in \eqref{eq:definition-virtual-channel} is reversible,
we can assume that the channel is a conditional additive from the beginning without loss of generality.

\subsection{Conversion of BPSK-AWGN Channel to Conditional Additive Channel} \label{subsection:conversion-bpsk-awgn}

Although we only considered finite input/output sources and channels throughout the paper,
in order to demonstrate the utility of the conditional additive channel framework, let us consider 
the additive white Gaussian noise (AWGN) channel with binary phase shift keying (BPSK) in this section. 
Let ${\cal A} = \{0,1\}$ be the input alphabet of the channel, and let ${\cal B} = \mathbb{R}$ be the output alphabet of the channel. 
For an input $a \in {\cal A}$ and Gaussian noise $Z$ with mean $0$ and variance $\sigma^2$, the output of channel is given by $B = (-1)^a + Z$. 
Then, the conditional probability density function of this channel is given as
\begin{align}
P_{B|A}(b|a)=\frac{1}{\sqrt{2\pi }\sigma}
e^{-\frac{(b-(-1)^{a})^2}{\sigma^2}}.\label{10-2-6}
\end{align}

Now, to define a conditional additive channel,
we choose ${\cal Y}:=\mathbb{R}_+$
and define the probability density function
$p_Y $ on ${\cal Y}$ with respect to the Lebesgue measure
and the conditional distribution $P_{X|Y}(x|y)$ as
\begin{align}
p_Y(y)&:= 
\frac{1}{\sqrt{2\pi }\sigma}
(e^{-\frac{(y-1)^2}{\sigma^2}}+ e^{-\frac{(y+1)^2}{\sigma^2}})
\\
P_{X|Y}(0|y)&:= 
\frac{e^{-\frac{(y-1)^2}{\sigma^2}}}{e^{-\frac{(y-1)^2}{\sigma^2}}+ e^{-\frac{(y+1)^2}{\sigma^2}}}
\\
P_{X|Y}(1|y)&:= 
\frac{e^{-\frac{(y+1)^2}{\sigma^2}}}{e^{-\frac{(y-1)^2}{\sigma^2}}+ e^{-\frac{(y+1)^2}{\sigma^2}}}
\end{align}
for $y \in \mathbb{R}_+$.
When we define $b := (-1)^x y \in \mathbb{R}$ for 
$x \in \{0,1\}$ and $y \in \mathbb{R}_{+}$,
we have
\begin{align}
p_{XY|A}(y,x|a)
=
\frac{1}{\sqrt{2\pi }\sigma}
e^{-\frac{(y-(-1)^{a+x})^2}{\sigma^2}}
=
\frac{1}{\sqrt{2\pi }\sigma}
e^{-\frac{((-1)^x y-(-1)^{a})^2}{\sigma^2}}
=
\frac{1}{\sqrt{2\pi }\sigma}
e^{-\frac{(b-(-1)^{a})^2}{\sigma^2}}.\label{10-2-7}
\end{align}
The relations \eqref{10-2-6} and \eqref{10-2-7}
show that the AWGN channel with BPSK is given as a conditional additive channel in the above sense.

By noting this observation, 
as explained in Remark \ref{R10-2-4},
the single-shot achievability bounds in Section \ref{subsection:multi-source-one-shot} 
are also valid for continuous $Y$, 
Also, the discussions for the single-shot converse bounds 
in Subsection \ref{subsection:channel-converse} hold even for continuous $Y$.
So, the bounds in Subsections \ref{subsection:channel-single-achievability}
and \ref{subsection:channel-converse} are also applicable to the BPSK-AWGN channel.

In particular, in the $n$ memoryless extension of the BPSK-AWGN channel,
the information measures for the noise distribution
are given as $n$ times of the single-shot information measures
for the noise distribution.
Even in this case, the upper and lower bounds in Subsections \ref{subsection:channel-single-achievability} and \ref{subsection:channel-converse}
are also applicable by replacing the information measures by $n$ times of the single-shot information measures.
Therefore, we obtain finite-length upper and lower bounds of the optimal coding length for the memoryless BPSK-AWGN channel.
Furthermore, 
even though the additive noise is not Gaussian,
when the probability density function $p_Z$ of the additive noise $Z$ satisfies the symmetry $p_Z(z)=p_Z(-z)$,
the BPSK channel with the additive noise $Z$ can be converted to a conditional additive channel in the same way.

\subsection{Achievability Bound Derived by Source Coding with Side-Information} \label{subsection:channel-single-achievability}
In this subsection, we give a code for a conditional additive channel from a code of source coding with side-information in a canonical way.
In this construction, we see that the decoding error probability of the channel code equals that of the source code.

When the channel is given as 
the conditional additive channel with conditional additive noise distribution 
$P_{X^nY^n}$ as \eqref{10-2-2}
and ${\cal X}={\cal A}$ is the finite field $\mathbb{F}_q$,
we can construct a linear channel code from a source coder
with full side-information whose encoder and decoder are $f_n$ and $\san{d}_n$ 
as follows. 
That is, we assume the linearity for the source encoder $f_n$.
Let ${\cal C}_n(f_n)$ be the kernel of the liner encoder $f_n$ of the source coder. 
Suppose that the sender sends a codeword $c_n \in {\cal C}_n(f_n)$
and $(c_n + X^n,Y^n)$ is received. 
Then, the receiver computes the syndrome $f_n(c_n+X^n) = f_n(X^n)$,
estimates $X^n$ from $f_n(X^n)$ and $Y^n$, and subtracts the estimate from $c_n + X^n$.
That is, 
we choose the channel decoder 
$\tilde{\san{d}}_n$ as
\begin{align}
\tilde{\san{d}}_n(x^{\prime n},y^n) := x^{\prime n} - \san{d}_n(f_n(x^{\prime n}),y^n).
\end{align}
We succeeded in decoding in this channel coding
if and only if $\san{d}_n (f_n(X^n),Y^n)$ equals $X^n$.
Thus, the error probability of this channel code coincides with that of the source code for the 
correlated source $(X^n,Y^n)$. 
In summary, we have the following lemma, which was first pointed out in \cite{chen:09c}.
\begin{lemma}[\protect{\cite[(19)]{chen:09c}}]
Given a linear encoder $f_n$ and a decoder $\san{d}_n$ for source coding with side-information with distribution $P_{X^nY^n}$,
let $I_{\rom{Ker}(f_n)}$ and $\tilde{\san{d}}_n$ be channel encoder and decoder induced from $(f_n,\san{d}_n)$. 
Then, the error probability of channel coding for conditionally additive channel with noise distribution $P_{X^nY^n}$ satisfies
\begin{align}
\Pce^{(n)}[(I_{\rom{Ker}(f_n)}, \tilde{\san{d}}_n)|P_{X^nY^n}] 
= \Pse^{(n)}[(f_n,\san{d}_n)|P_{X^nY^n}].
\end{align}
Furthermore,\footnote{In fact, when we additionally impose the linearity to the random function $F$ in the definition
\eqref{eq:definition-Pe-bar-source-side-info-2} for the definition of $\barPse(M|P_{X^nY^n}) $,
the result in \cite{chen:09c} implies that the equality in \eqref{eq:relation-bar-error-source-channel} holds.} 
taking the infimum for $F_n$ chosen to be a linear two-universal hash function, we also have
\begin{align} 
& \barPce(n,k) 
=
\sup_{F_n \in {\cal F}_l}  \mathbb{E}_{F_n}
\left[ \Pce^{(n)} [\Psi(F_n)] \right]
\le \sup_{F_n \in {\cal F}_l}  \mathbb{E}_{F_n}
\left[ \Pce^{(n)} [(I_{\rom{Ker}(F_n)}, \tilde{\san{d}}_n)] \right] \nonumber\\
=& \sup_{F_n \in {\cal F}_l}  \mathbb{E}_{F_n}
 \Pse^{(n)}[(F_n,\san{d}_n)] 
\le \sup_{F_n \in {\cal F}}  \mathbb{E}_{F_n}
 \Pse^{(n)}[(F_n,\san{d}_n)]
= \barPse^{(n)}(|{\cal A}^{n-k}| ).
\label{eq:relation-bar-error-source-channel}
\end{align}
\end{lemma}

By using this observation and the results in Section \ref{subsection:multi-source-one-shot},
we can derive the achievability bounds. 
By using the conversion argument in Section \ref{subsection:conversion-regular-g-additive},
we can also construct a channel code for a regular channel from a source code with full side-information.
Although the following bounds are just specialization of known bounds for conditional additive channels, 
we review these bounds here to clarify correspondence between the bounds in source coding
with side-information and channel coding. 

From Lemma \ref{lemma:multi-source-han-direct} and \eqref{eq:relation-bar-error-source-channel}, we have the following.
\begin{lemma}[\cite{verdu:94}] \label{lemma:channel-spectrum-direct}
The following bound holds:
\begin{eqnarray}
\barPce(n,k) \le \inf_{\gamma \ge 0} \left[ P_{X^nY^n}\left\{ \log\frac{1}{P_{X^n|Y^n}(x^n|y^n)} > \gamma \right\} + \frac{e^\gamma}{|{\cal A}|^{n-k}} \right].
\end{eqnarray}
\end{lemma}

From Lemma \ref{lemma:multi-source-tight-bound} and \eqref{eq:relation-bar-error-source-channel}, 
we have the following exponential type bound.
\begin{lemma}[\cite{gallager:65}] \label{lemma:channel-gallager-bound}
The following bound holds:
\begin{eqnarray}
\barPce(n,k) \le \inf_{-\frac{1}{2} \le \theta \le 0} |{\cal A}|^{\frac{\theta(n-k)}{1+\theta}} 
  e^{-\frac{\theta}{1+\theta} H_{1+\theta}^\uparrow(X^n|Y^n)}.
\end{eqnarray}
\end{lemma}

From Lemma \ref{lemma:multi-source-loose-bound} and \eqref{eq:relation-bar-error-source-channel}, 
we have the following slightly loose exponential bound.
\begin{lemma}[\cite{han:book, hayashi:07b}] \label{lemma:channel-loose-bound}
The following bound holds:\footnote{The bound \eqref{eq:channel-loose-exponential-bound} was derived
in the original Japanese edition of \cite{han:book}, but it is not written in the English edition \cite{han:book}.
The quantum analogue was derived in \cite{hayashi:07b}.}
\begin{eqnarray} \label{eq:channel-loose-exponential-bound}
\barPce(n,k) \le
 \inf_{-1 \le \theta \le 0} |{\cal A}|^{\theta (n-k)} e^{- \theta H^{\downarrow}_{1+\theta}(X^n|Y^n)}.
\end{eqnarray}
\end{lemma}

When ${\cal Y}$ is singleton, i.e., the virtual channel is additive, we have the following special case of Lemma \ref{lemma:channel-gallager-bound}.
\begin{lemma}[\cite{gallager:65}] \label{lemma:channel-additive-channel-bound}
Suppose that ${\cal Y}$ is singleton. Then, the following bound holds:
\begin{eqnarray}
\barPce(n,k) \le
\inf_{- \frac{1}{2} \le \theta \le 0} |{\cal A}|^{\frac{\theta (n-k)}{1+\theta}} e^{- \frac{\theta}{1+\theta} H_{1+\theta}(X^n)}.
\end{eqnarray}
\end{lemma}

\subsection{Converse Bound} \label{subsection:channel-converse}

In this subsection, we show some converse bounds.
The following is the information spectrum type converse shown in \cite{hayashi:03}.

\begin{lemma}[\protect{\cite[Lemma 4]{hayashi:03}}] \label{lemma:channel-one-shot-converse}
For any code $\Psi_n = (\san{e}_n,\san{d}_n)$ and any output distribution $Q_{B^n} \in {\cal P}({\cal B}^n)$, we have
\begin{eqnarray} \label{eq:hayashi-nagaoka-bound}
\Pce^{(n)}[\Psi_n] \ge \sup_{\gamma \ge 0}\left[ \sum_{m=1}^{M_n} \frac{1}{M_n} 
 P_{B^n|A^n}\left\{ \log \frac{P_{B^n|A^n}(b^n|\san{e}_n(m))}{Q_{B^n}(b^n)} < \gamma \right\} - \frac{e^\gamma}{M_n} \right].
\end{eqnarray}
\end{lemma} 

When a channel is a conditional additive channel, we have
\begin{align} \label{eq:property-of-conditional-additive}
P_{B^n|A^n}(a^n + x^n, y^n|a^n) = P_{X^nY^n}(x^n,y^n).
\end{align} 
By taking the output distribution $Q_{B^n}$ as 
\begin{eqnarray} \label{eq:choice-of-Q}
Q_{B^n}(a^n + x^n, y^n) = \frac{1}{|{\cal A}|^n} Q_{Y^n}(y^n)
\end{eqnarray}
for some $Q_{Y^n} \in {\cal P}({\cal Y}^n)$, we have the following bound.
\begin{lemma} \label{lemma:channel-converse-general-additive-specialized}
When a channel is a conditional additive channel, 
for any distribution $Q_{Y^n} \in {\cal P}({\cal Y}^n)$, we have
\begin{eqnarray} \label{eq:channel-converse-general-additive}
\Pce^{(n)}(M_n) 
\ge \sup_{\gamma \ge 0}\left[ P_{X^nY^n}\left\{ \log \frac{Q_{Y^n}(y^n)}{P_{X^nY^n}(x^n,y^n)} > n\log |{\cal A}| - \gamma \right\} - \frac{e^\gamma}{M_n} \right].
\end{eqnarray}
\end{lemma}
\begin{proof}
By noting \eqref{eq:property-of-conditional-additive} and \eqref{eq:choice-of-Q}, the first term of the right hand side 
of \eqref{eq:hayashi-nagaoka-bound} can be rewritten as 
\begin{align}
& \sum_{m=1}^{M_n} \frac{1}{M_n} 
 P_{B^n|A^n}\left\{ \log \frac{P_{B^n|A^n}(b^n|\san{e}_n(m))}{Q_{B^n}(b^n)} < \gamma \right\} \\
&= \sum_{m=1}^{M_n} \frac{1}{M_n}
 P_{X^n Y^n}\left\{ \log \frac{P_{B^n|A^n}(\san{e}_n(m)+x^n,y^n|\san{e}_n(m))}{Q_{B^n}(\san{e}_n(m) + x^n,y^n)} \right\} \\
&= P_{X^nY^n}\left\{ \log \frac{Q_{Y^n}(y^n)}{P_{X^nY^n}(x^n,y^n)} > n\log |{\cal A}| - \gamma \right\},
\end{align}
which implies the statement of the lemma.
\end{proof}

By a similar argument as in Theorem \ref{theorem:multi-source-converse}, we can also derive the following converse bound.
\begin{theorem} \label{theorem:channel-exponential-converse}
For any $Q_{Y^n} \in {\cal P}({\cal Y}^n)$, we have
\begin{eqnarray}
\lefteqn{ -\log \Pce^{(n)}(M_n) } \\
&\le& \inf_{s > 0 \atop \tilde{\theta} \in \mathbb{R}, \vartheta \ge 0}
\bigg[ (1+s) \tilde{\theta} \left\{ H_{1+\tilde{\theta}}(P_{X^nY^n}|Q_{Y^n}) - H_{1+(1+s)\tilde{\theta}}(P_{X^nY^n}|Q_{Y^n}) \right\}  \\
&& - (1+s) \log \left( 1- 2 e^{- \frac{-\vartheta R 
  + (\tilde{\theta}+\vartheta(1+\tilde{\theta})) H_{1+\tilde{\theta} + \vartheta(1+\tilde{\theta})}(P_{X^nY^n}|Q_{Y^n})
    - (1+\vartheta)\tilde{\theta}H_{1+\tilde{\theta}}(P_{X^nY^n}|Q_{Y^n})}{1+\vartheta}} \right) \bigg] / s
  \label{eq:bound-channel-exponential-converse-1} \\
&\le& \inf_{s > 0 \atop -1 < \tilde{\theta} < \theta(a(R )) }
\bigg[
 (1+s)\tilde{\theta} \left\{ H_{1+\tilde{\theta}}(P_{X^nY^n}|Q_{Y^n}) - H_{1+(1+s)\tilde{\theta}}(P_{X^nY^n}|Q_{Y^n})  \right\} \\
&&  - (1+s) \log \left( 1 - 2 e^{ (\theta(a(R ))-\tilde{\theta}) a(R ) - \theta(a(R )) H_{1+\theta(a(R ))}(P_{X^nY^n}|Q_{Y^n})
  + \tilde{\theta} H_{1+\tilde{\theta}}(P_{X^nY^n}|Q_{Y^n})  } \right)
\bigg] / s, 
 \label{eq:bound-channel-exponential-converse-2}
\end{eqnarray}
where $R = n \log |{\cal A}| - \log M_n$, 
and $\theta(a)$ and $a(R )$ are the inverse functions defined in \eqref{eq:definition-inverse-theta-multi-one-shot} 
and \eqref{eq:definition-inverse-a-multi-one-shot} respectively.
\end{theorem}
\begin{proof}
See Appendix \ref{appendix:theorem:channel-exponential-converse}.
\end{proof}

\subsection{Finite-Length Bound for Markov Noise Channel} \label{subsection:channel-finite-markov}

From this section, we address conditional additive channel
whose conditional additive noise us subject to Markovian chain.
Here, the input alphabet ${\cal A}^n$ equals the additive group ${\cal X}^n=\mathbb{F}_q^n$
and the output alphabet ${\cal B}^n$ is ${\cal X}^\times {\cal Y}^n$.
That is, the transition matrix describing the channel is given 
by using a transition matrix ${W}$ on ${\cal X}^\times {\cal Y}^n$
and an initial distribution $Q$ as
\begin{align}
P_{B^n|A^n}(x^n+a^n,y^n|a^n)
= Q({x}_1,{y}_1) 
\prod_{i=2}^n {W}({x}_i,{y}_i|{x}_{i-1},{y}_{i-1}).
\end{align}

As in Section \ref{subsection:multi-terminal-measures-markov}, we consider two assumptions on the transition matrix $W$
of the noise process $(\mathbf{X},\mathbf{Y})$, i.e., Assumption \ref{assumption-Y-marginal-markov}
and Assumption \ref{assumption-memory-through-Y}. We also use the same notations as in 
Section \ref{subsection:multi-terminal-measures-markov}.

\begin{example}[Gilbert-Elliot channel with state-information available at the receiver] \label{example:gilbert-elliot}
The Gilbert-Elliot channel \cite{gilbert:60,elliott:63} is characterized by a channel state $Y^n$ on ${\cal Y}^n = \{0,1\}^n$,
and an additive noise $X^n$ on ${\cal X}^n = \{0,1\}^n$. 
The noise process $(X^n,Y^n)$ is a Markov chain induced by the transition matrix $W$ introduced in
Example \ref{example:gilbert-elliot-noise}.
For the channel input $a^n$, the channel output is given by $(a^n + X^n, Y^n)$
when the state-information is available at the receiver. Thus, this channel can be regarded as 
a conditional additive channel, and the transition matrix of the noise process satisfies 
Assumption \ref{assumption-memory-through-Y}.
\end{example}

Proofs of the following bounds are almost the same as those in Section \ref{subsection:multi-source-finite-markov}, and thus omitted.
From Lemma \ref{lemma:channel-loose-bound} and 
Lemma \ref{lemma:mult-terminal-finite-evaluation-down-conditional-renyi},
we can derive the following achievability bound.
\begin{theorem}[Direct, Ass. 1] \label{theorem:channel-finite-marko-direct-assumptiton-1}
Suppose that the transition matrix $W$ of the conditional additive noise
satisfies Assumption \ref{assumption-Y-marginal-markov}.
Let $R := \frac{n-k}{n}\log |{\cal A}|$. Then we have
\begin{eqnarray}
- \log \barPce(n,k) &\ge& \sup_{-1 \le \theta \le 0} \left[ -\theta n R + (n-1) \theta H_{1+\theta}^{\downarrow,W}(X|Y) + \underline{\delta}(\theta) \right].
\end{eqnarray}
\end{theorem}

From Theorem \ref{theorem:channel-exponential-converse} for $Q_{Y^n} = P_{Y^n}$
and Lemma \ref{lemma:mult-terminal-finite-evaluation-down-conditional-renyi}, we have the following converse bound.
\begin{theorem}[Converse, Ass. 1] \label{theorem:channel-finite-markov-converse-assumptiotn-1}
Suppose that transition matrix $W$ of the conditional additive noise satisfies Assumption \ref{assumption-Y-marginal-markov}.
Let $R := \log |{\cal A}| - \frac{1}{n} \log M_n$. If $H^W(X|Y) < R < H_0^{\downarrow,W}(X|Y)$, then we have
\begin{eqnarray}
\lefteqn{ - \log \Pce^{(n)}(M_n) } \\&\le&
 \inf_{s > 0 \atop -1 < \tilde{\theta} < \theta(a(R )) } \bigg[ (n-1) (1+s) \tilde{\theta}
  \left\{ H_{1+\tilde{\theta}}^{\downarrow,W}(X|Y) - H_{1+(1+s)\tilde{\theta}}^{\downarrow,W}(X|Y) \right\} + \delta_1 \\
  && - (1+s) \log \left(1 - 2e^{(n-1)[ (\theta(a(R )) - \tilde{\theta}) a(R ) - \theta(a(R )) H_{1+\theta(a(R ))}^{\downarrow,W}(X|Y)
   + \tilde{\theta} H_{1+\tilde{\theta}}^{\downarrow,W}(X|Y)] + \delta_2} \right) \bigg] / s,
\end{eqnarray}
where 
$\theta(a) = \theta^\downarrow(a)$ and $a(R ) = a^\downarrow(R )$ 
are the inverse functions defined by \eqref{eq:definition-theta-inverse-multi-markov} 
and \eqref{eq:definition-a-inverse-multi-markov} respectively, and
\begin{eqnarray}
\delta_1 &:=& (1+s) \overline{\delta}(\tilde{\theta}) - \underline{\delta}((1+s)\tilde{\theta}), \\
\delta_2 &:=& \frac{ (\theta(a(R ))-\tilde{\theta}) R - (1+\tilde{\theta}) \underline{\delta}(\theta(a(R ))) + (1+\theta(a(R ))) \overline{\delta}(\tilde{\theta})}{1 + \theta(a(R ))}.
\end{eqnarray}
\end{theorem}

Next, we derive tighter bounds under Assumption \ref{assumption-memory-through-Y}.
From Lemma \ref{lemma:channel-gallager-bound} and 
Lemma \ref{lemma:multi-terminal-finite-evaluation-upper-conditional-renyi}, we have the following achievability bound.
\begin{theorem}[Direct, Ass. 2] \label{theorem:channel-finite-marko-direct-assumptiton-2}
Suppose that the transition matrix $W$ of the conditional additive noise satisfies Assumption \ref{assumption-memory-through-Y}. 
Let $R := \frac{n-k}{n} \log |{\cal A}| $. Then we have
\begin{eqnarray}
- \log \barPce(n,k) \ge \sup_{-\frac{1}{2} \le \theta \le 0} 
 \frac{-\theta n R + (n-1) \theta H_{1+\theta}^{\uparrow,W}(X|Y)}{1+\theta} + \underline{\xi}(\theta).
\end{eqnarray}
\end{theorem}

By using Theorem \ref{theorem:channel-exponential-converse} for $Q_{Y^n} = P_{Y^n}^{(1+\theta(a(R )))}$
and Lemma \ref{lemma:multi-terminal-finite-evaluation-two-parameter-conditional-renyi}, we can derive the following converse bound.
\begin{theorem}[Converse, Ass. 2] \label{theorem:channel-finite-marko-converse-assumptiton-2}
Suppose that the transition matrix $W$ of the conditional additive noise satisfies Assumption \ref{assumption-memory-through-Y}. 
Let $R := \log|{\cal A}| - \frac{1}{n} \log M_n$. If $H^W(X|Y) < R < H_0^{\uparrow,W}(X|Y)$, we have
\begin{eqnarray}
\lefteqn{ - \log \Pce^{(n)}(M_n) } \\
&\le& \inf_{s > 0 \atop -1 < \tilde{\theta} < \theta(a(R ))} \bigg[
 (n-1) (1+s) \tilde{\theta} 
  \left\{ H_{1+\tilde{\theta},1+\theta(a(R ))}^W(X|Y) - H_{1+(1+s)\tilde{\theta},1+\theta(a(R ))}^W(X|Y) \right\} + \delta_1 \\
 && - (1+s)\log \left( 1 - 2e^{(n-1)[ (\theta(a(R ))-\tilde{\theta}) a(R ) - \theta(a(R )) H_{1+\theta(a(R ))}^{\uparrow,W}(X|Y)
  + \tilde{\theta} H_{1+\tilde{\theta},1+\theta(a(R ))}^W(X|Y) ] + \delta_2} \right)
\bigg] / s,
\end{eqnarray}
where 
$\theta(a) = \theta^\uparrow(a)$ and $a(R ) = a^\uparrow(R )$ 
are the inverse functions defined by \eqref{eq:definition-theta-inverse-markov-optimal-Q} 
and \eqref{eq:definition-a-inverse-markov-optimal-Q} respectively, and
\begin{eqnarray}
\delta_1 &:=& (1+s) \overline{\zeta}(\tilde{\theta}, \theta(a(R ))) - \underline{\zeta}((1+s)\tilde{\theta},\theta(a(R ))), \\
\delta_2 &:=& \frac{ (\theta(a(R ))-\tilde{\theta}) R - (1+\tilde{\theta}) \underline{\zeta}(\theta(a(R )),\theta(a(R )))
 + (1+\theta(a(R ))) \overline{\zeta}(\tilde{\theta},\theta(a(R ))) }{1+\theta(a(R ))}.
\end{eqnarray}
\end{theorem}

Finally, when ${\cal Y}$ is singleton, i.e., the channel is additive, we can derive the 
following achievability bound from Lemma \ref{lemma:channel-additive-channel-bound}.
\begin{theorem}[Direct, Singleton] \label{theorem:channel-finite-markov-additive}
Let $R := \frac{n-k}{n} \log |{\cal A}|$. Then we have
\begin{eqnarray}
- \log \barPce(n,k) \ge \sup_{-\frac{1}{2} \le \theta \le 0} 
 \frac{ - \theta n R + (n-1) \theta H_{1+\theta}^W(X) + \underline{\delta}(\theta)}{1 + \theta}.
\end{eqnarray}
\end{theorem}

\begin{remark}
Our treatment for Markovian conditional additive channel covers Markovian regular channels because Markovian regular channel can be reduced to Markovian conditional additive channel as follows.
Let $\tilde{\mathbf{X}} = \{ \tilde{X}^n \}_{n=1}^\infty$ be a Markov chain on ${\cal B}$
whose distribution is given by
\begin{eqnarray}
P_{\tilde{X}^n}(\tilde{x}^n) = Q(\tilde{x}_1) \prod_{i=2}^n \tilde{W}(\tilde{x}_i|\tilde{x}_{i-1})
\end{eqnarray}
for a transition matrix $\tilde{W}$ and an initial distribution $Q$.
Let $(\mathbf{X},\mathbf{Y}) = \{ (X^n,Y^n) \}_{n=1}^\infty$ be the noise process 
of the conditional additive channel derived from the noise process $\tilde{\mathbf{X}}$ of
the regular channel  by the argument of Section \ref{subsection:conversion-regular-g-additive}.
Since we can write 
\begin{eqnarray}
P_{X^nY^n}(x^n,y^n) = Q(\iota_{y_1}^{-1}(\vartheta_{y_1}(x_1))) \frac{1}{|\rom{Stb}(0_{y_1})|}
 \prod_{i=2}^n \tilde{W}(\iota_{y_i}^{-1}(\vartheta_{y_i}(x_i))| \iota_{y_{i-1}}^{-1}(\vartheta_{y_{i-1}}(x_{i-1}))) 
  \frac{1}{|\rom{Stb}(0_{y_i})|},
\end{eqnarray}
the process $(\mathbf{X},\mathbf{Y})$ is also a Markov chain. 
Thus, the regular channel given by $\tilde{\mathbf{X}}$ is reduced to the conditional 
additive channel given by $(\mathbf{X},\mathbf{Y})$.
\end{remark}

\subsection{Second Order} \label{subsection:multi-random-second-order}

To discuss the asymptotic performance, we introduce the quantity
\begin{eqnarray}
C := \log |{\cal A}| - H^W(X|Y).
\end{eqnarray}
By applying the central limit theorem 
(cf.~\cite[Theorem 27.4, Example 27.6]{billingsley-book}) to
Lemma \ref{lemma:channel-spectrum-direct} and 
Lemma \ref{lemma:channel-converse-general-additive-specialized} for $Q_{Y^n} = P_{Y^n}$, 
and by using Theorem \ref{theorem:multi-markov-variance}, we have the following.
\begin{theorem} \label{theorem:channel-second-order}
Suppose that the transition matrix $W$ of the conditional additive noise satisfies Assumption \ref{assumption-Y-marginal-markov}.
For arbitrary $\varepsilon \in (0,1)$, we have
\begin{eqnarray}
\log M(n,\varepsilon)=
k(n,\varepsilon) \log |{\cal A}|
= C n + \sqrt{\san{V}^W(X|Y)} \Phi^{-1}(\varepsilon)\sqrt{n}+ 
o(\sqrt{n}).
\end{eqnarray}
\end{theorem} 
\begin{proof}
It can be proved exactly in the same manner as Theorem \ref{theorem:source-coding-second-order}.
\end{proof}

From the above theorem, the (first-order) capacity of the conditional additive channel 
under Assumption \ref{assumption-Y-marginal-markov} is given by
\begin{eqnarray}
\lim_{n\to\infty} \frac{1}{n} \log M(n,\varepsilon)
= \lim_{n\to \infty} \frac{1}{n} \log \frac{k(n,\varepsilon) \log |{\cal A}|}{n} 
= C 
\end{eqnarray}
for every $0 < \varepsilon < 1$.
In the next subsections, we consider the asymptotic behavior of the error probability when the rate
is smaller than the capacity in the moderate deviation regime
and the large deviation regime, respectively.

\subsection{Moderate Deviation} \label{subsection:channel-mdp}

From Theorem \ref{theorem:channel-finite-marko-direct-assumptiton-1} and
Theorem \ref{theorem:channel-finite-markov-converse-assumptiotn-1}, we have the following.
\begin{theorem} \label{theorem:channel-mdp}
Suppose that the transition matrix $W$ of the conditional additive noise satisfies Assumption \ref{assumption-Y-marginal-markov}.
For arbitrary $t \in (0,1/2)$ and $\delta > 0$, we have 
\begin{eqnarray}
\lim_{n\to\infty} - \frac{1}{n^{1-2t}} \log \Pce^{(n)}\left(e^{n C + n^{1-t}\delta} \right)
&=& \lim_{n\to\infty} - \frac{1}{n^{1-2t}} \log \barPce^{(n)}\left(n, \frac{nC - n^{1-t}\delta}{\log|{\cal A}|} \right) \\
&=& \frac{\delta^2}{2 \san{V}^W(X|Y)}.
\end{eqnarray}
\end{theorem}
\begin{proof}
It can be proved exactly in the same manner as Theorem \ref{theorem:multi-moderate-deviation}.
\end{proof}

\subsection{Large Deviation} \label{subsection:channel-ldp}

From Theorem \ref{theorem:channel-finite-marko-direct-assumptiton-1} and
Theorem \ref{theorem:channel-finite-markov-converse-assumptiotn-1}, we have the following.
\begin{theorem} \label{theorem:large-deviation-source-coding-assumption-1}
Suppose that the transition matrix $W$ of the conditional additive noise satisfies Assumption \ref{assumption-Y-marginal-markov}.
For $H^W(X|Y) < R$, we have
\begin{eqnarray}
\liminf_{n\to\infty} - \frac{1}{n} \log \barPce^{(n)}\left(n, n\left(1-\frac{R}{\log|{\cal A}|}\right)\right) 
 \ge \sup_{-1 \le \theta \le 0}\left[ -\theta R + \theta H_{1+\theta}^{\downarrow,W}(X|Y) \right].
\end{eqnarray}
On the other hand, for $H^W(X|Y) < R < H_0^{\downarrow,W}(X|Y)$, we have
\begin{eqnarray}
\limsup_{n\to\infty} - \frac{1}{n} \log \Pce^{(n)}\left(e^{n(\log|{\cal A}|-R)}\right) 
 &\le& - \theta(a(R )) a(R ) + \theta(a(R )) H_{1+\theta(a(R ))}^{\downarrow,W}(X|Y) \\
 &=& \sup_{-1 < \theta \le 0} \frac{- \theta R + \theta H_{1+\theta}^{\downarrow,W}(X|Y)}{1+\theta}.
 \label{eq:converse-channel-assumption-1}
\end{eqnarray}
\end{theorem} 
\begin{proof}
It can be proved exactly in the same manner as Theorem \ref{theorem:source-coding-ldp-assumption-1}.
\end{proof}

Under Assumption \ref{assumption-memory-through-Y}, from 
Theorem \ref{theorem:channel-finite-marko-direct-assumptiton-2} and 
Theorem \ref{theorem:channel-finite-marko-converse-assumptiton-2}, we have the following tighter bound.
\begin{theorem} \label{theorem:channel-ldp-2}
Suppose that the transition matrix $W$ of the conditional additive noise satisfies Assumption \ref{assumption-memory-through-Y}. 
For $H^W(X|Y) < R$, we have
\begin{eqnarray} \label{eq:channel-ldp-assumption-2-direct}
\liminf_{n\to\infty} - \frac{1}{n} \log \barPce^{(n)}\left(n, n\left(1-\frac{R}{\log|{\cal A}|}\right)\right) 
\ge \sup_{-\frac{1}{2} \le \theta \le 0} \frac{- \theta R + \theta H_{1+\theta}^{\uparrow,W}(X|Y)}{1+\theta}.
\end{eqnarray}
On the other hand, for $H^W(X|Y) < R < H_0^{\uparrow,W}(X|Y)$, we have
\begin{eqnarray} \label{eq:multi-source-ldp-assumption-2-converse}
\limsup_{n\to\infty} - \frac{1}{n} \log \Pce^{(n)}\left(e^{n(\log|{\cal A}|-R)}\right) 
 &\le& - \theta(a(R )) a(R ) + \theta(a(R )) H_{1+\theta(a(R ))}^{\uparrow,W}(X|Y) \\
 &=& \sup_{-1 < \theta \le 0} \frac{- \theta R + \theta H_{1+\theta}^{\uparrow,W}(X|Y)}{1+\theta}.
\end{eqnarray}
\end{theorem}
\begin{proof}
It can be proved exactly in the same manner as Theorem \ref{theorem:source-coding-ldp-assumption-2}.
\end{proof}

When ${\cal Y}$ is singleton, i.e., the channel is additive, 
from Theorem \ref{theorem:channel-finite-markov-additive} and 
\eqref{eq:converse-channel-assumption-1}, we have the following.
\begin{theorem} \label{theorem:channel-ldp-additive}
For $H^W(X) < R$, we have
\begin{eqnarray}
\liminf_{n\to\infty} - \frac{1}{n} \log \barPce^{(n)}\left(n, n\left(1-\frac{R}{\log|{\cal A}|}\right)\right) 
\ge \sup_{-\frac{1}{2} \le \theta \le 0} \frac{ - \theta R + \theta H_{1+\theta}^W(X) }{1 + \theta}.
\end{eqnarray}
On the other hand, for $H^W(X) < R < H_0^W(X)$, we have
\begin{eqnarray}
\limsup_{n\to\infty} - \frac{1}{n} \log \Pce^{(n)}\left(e^{n(\log|{\cal A}|-R)}\right) 
 \le \sup_{-1 < \theta \le 0} \frac{- \theta R + \theta H_{1+\theta}^W(X)}{1+\theta}.
 \end{eqnarray}
\end{theorem}
\begin{proof}
It can be proved in the same manner as Remark \ref{remark:singleton}.
\end{proof}

%% file: conclusion.tex
\section{Discussion and Conclusion}

In this paper, we have developed a unified approach to
source coding with side information and channel coding for conditional additive channel 
for finite-length and asymptotic analyses of Markov chains.
In our approach, the conditional R\'enyi entropies defined for transition
matrices play important roles. Although we only illustrated the source coding
with side-information and the channel coding for conditional additive channel 
as applications of our approach, 
it can be applied to some other problems in information theory such as random number generation problems,
as shown in another paper \cite{HayWat16}.

Our obtained results for 
the source coding with side information 
and the channel coding of the conditional additive channel 
has been extended to the case 
when the side information is continuous and the joint distribution $X$ and $Y$ is memoryless.
Since this case covers the BPSK-AWGN channel,
it can be expected that it covers the MPSK-AWGN channel.
Since such channels are often employed in the real channel coding,
it is an interesting future topic to investigate the finite-length bound for these channels.
Further, we could not define the conditional R\'{e}nyi entropy for transition matrices of continuous $Y$.
Hence, our result could not extended to such a continuous case.
It is another interesting future topic to extend the obtained result to the case with 
continuous $Y$.

%% file: Appendix-Preparation.tex
\subsection{Preparation for Proofs} \label{Appendix:preparation}

When we prove some properties of R\'enyi entropies or derive converse bounds,
some properties of cumulant generating functions (CGFs) become useful.
For this purpose, we introduce some terminologies in statistics
from \cite{hayashi-watanabe:13,hayashi-watanabe:13b}.
Then, in  Appendix \ref{Appendix:preparetion-multi-terminal}, 
we show relation between terminologies in statistics and those in
information theory.
For proofs, see \cite{hayashi-watanabe:13,hayashi-watanabe:13b}.

\subsubsection{Single-Shot Setting}

Let $Z$ be a random variable with distribution $P$. Let 
\begin{eqnarray} \label{eq:definition-single-shot-cgf}
\phi(\rho) &:=& \log \mathsf{E}\left[ e^{\rho Z}\right] \\
&=& \log \sum_z P(z) e^{\rho Z}
\end{eqnarray}
be the cumulant generating function (CGF). 
Let us introduce an exponential family
\begin{eqnarray} \label{eq:one-shot-tail-exponential-family}
P_\rho(z) := P(z)e^{\rho z - \phi(\rho)}.
\end{eqnarray}
By differentiating the CGF, we find that 
\begin{eqnarray}
\phi^\prime(\rho) &=& \mathsf{E}_\rho[Z] \\
&:=& \sum_z P_\rho(z) z. 
\end{eqnarray}
We also find that 
\begin{eqnarray} \label{eq:one-shot-tail-second-derivative}
\phi^{\prime\prime}(\rho) = \sum_z P_\rho(z) \left( z - \mathsf{E}_\rho[Z] \right)^2.
\end{eqnarray}
We assume that $Z$ is not constant. Then, \eqref{eq:one-shot-tail-second-derivative} implies that $\phi(\rho)$ is a strict convex function
and $\phi^\prime(\rho)$ is monotonically increasing. Thus, we can define the inverse function $\rho(a)$ of $\phi^\prime(\rho)$
by 
\begin{eqnarray} \label{eq:definition-cgf-single-inverse-function}
\phi^\prime(\rho(a)) = a.
\end{eqnarray}

Let 
\begin{eqnarray}
D_{1+s}(P\|Q) := \frac{1}{s} \log \sum_z P(z)^{1+s} Q(z)^{-s}
\end{eqnarray}
be the R\'enyi divergence. Then, we have the following relation:
\begin{eqnarray} \label{eq:relation-renyi-divergence-cgf}
s D_{1+s}(P_{\tilde{\rho}}\|P_{\rho}) =  \phi((1+s) \tilde{\rho} - s \rho) - (1+s) \phi(\tilde{\rho}) + s \phi(\rho).
\end{eqnarray}

\subsubsection{Transition Matrix}

Let $\{ W(z|z^\prime) \}_{(z,z^\prime) \in {\cal Z}^2}$ be 
an ergodic and irreducible transition matrix, and let $\tilde{P}$ be its
stationary distribution.
For a function $g:{\cal Z} \times {\cal Z} \to \mathbb{R}$,  let
\begin{eqnarray}
\mathsf{E}[g] := \sum_{z,z^\prime} \tilde{P}(z^\prime) W(z|z^\prime) g(z,z^\prime).
\end{eqnarray}
We also introduce the following tilted matrix:
\begin{eqnarray}
W_\rho(z|z^\prime) := W(z|z^\prime) e^{\rho g(z,z^\prime)}.
\end{eqnarray}
Let $\lambda_\rho$ be the Perron-Frobenius eigenvalue of $W_\rho$. 
Then, the CGF for $W$ with generator $g$ is defined by
\begin{eqnarray} \label{eq:definition-transition-cgf}
\phi(\rho) := \log \lambda_\rho.
\end{eqnarray}

\begin{lemma} \label{lemma:general-markov-cgf-strict-convexity}
The function $\phi(\rho)$ is a convex function of $\rho$, and it is strict convex iff. $\phi^{\prime\prime}(0) > 0$.
\end{lemma}
From Lemma \ref{lemma:general-markov-cgf-strict-convexity}, $\phi^\prime(\rho)$ is monotone increasing function.
Thus, we can define the inverse function $\rho(a)$ of $\phi^\prime(\rho)$ by
\begin{eqnarray} \label{eq:definition-inverse-general-markov}
\phi^\prime(\rho(a)) = a.
\end{eqnarray}

\subsubsection{Markov Chain}

Let $\mathbf{Z} = \{ Z^n \}_{n=1}^\infty$ be the Markov chain induced by $W(z|z^\prime)$ and 
an initial distribution $P_{Z_1}$. For functions $g:{\cal Z} \times {\cal Z} \to \mathbb{R}$ 
and $\tilde{g}:{\cal Z} \to \mathbb{R}$, let $S_n := \sum_{i=2}^n g(Z_i,Z_{i-1}) + \tilde{g}(Z_1)$. Then,
the CGF for $S_n$ is given by
\begin{eqnarray}
\phi_n(\rho) := \log \mathsf{E}\left[ e^{\rho S_n} \right].
\end{eqnarray}
We will use the following finite evaluation for $\phi_n(\rho)$. 
\begin{lemma} \label{lemma:finite-evaluation-of-cgf}
Let $v_\rho$ be the eigenvector of $W_\rho^T$ with respect to the Perron-Frobenius eigenvalue $\lambda_\rho$
such that $\min_{z} v_\rho(z) =1$. Let $w_\rho(z) := P_{Z_1}(z) e^{\rho \tilde{g}(z)}$. Then, we have
\begin{eqnarray}
(n-1) \phi(\rho) + \underline{\delta}_\phi(\rho) 
\le \phi_n(\rho) \le (n-1) \phi(\rho) + \overline{\delta}_\phi(\rho),
\end{eqnarray}
where 
\begin{eqnarray}
\overline{\delta}_\phi(\rho) &:=& \log \langle v_\rho | w_\rho \rangle, \\
\underline{\delta}_\phi(\rho) &:=& \log \langle v_\rho | w_\rho \rangle - \log \max_z v_\rho(z).
\end{eqnarray}
\end{lemma}

From this lemma, we have the following.
\begin{corollary}
For any initial distribution and $\rho \in \mathbb{R}$, we have
\begin{eqnarray}
\lim_{n\to \infty } \phi_n(\rho) = \phi(\rho).
\end{eqnarray}
\end{corollary}

The relation
\begin{eqnarray}
\lim_{n\to\infty} \frac{1}{n} \san{E}[S_n] &=& \phi^\prime(0) \\
&=& \san{E}[g]
\end{eqnarray}
is well known.
Furthermore, we also have the following.
\begin{lemma} \label{lemma:appendix-variance-limit}
For any initial distribution, we have
\begin{eqnarray}
\lim_{n \to \infty} \frac{1}{n} \mathrm{Var}\left[ S_n \right] = \phi^{\prime\prime}(0).
\end{eqnarray}
\end{lemma}

\subsection{Relation Between CGF and Conditional R\'enyi Entropies} \label{Appendix:preparetion-multi-terminal}

\subsubsection{Single-Shot Setting}

For correlated random variable $(X,Y)$, let us consider $Z = \log \frac{Q_Y(Y)}{P_{XY}(X,Y)}$. Then,
the relation between the CGF and conditional R\'enyi entropy relative to $Q_Y$ is given by
\begin{eqnarray} \label{eq:relation-single-shot-q-phi}
\theta H_{1+\theta}(P_{XY}|Q_Y) = - \phi(-\theta; P_{XY}|Q_Y).
\end{eqnarray}
From this, we can also find that the relationship between the inverse functions
(cf.~\eqref{eq:definition-inverse-theta-multi-one-shot} and \eqref{eq:definition-cgf-single-inverse-function}):
\begin{eqnarray}
\theta(a) = - \rho(a).
\end{eqnarray}
Thus, the inverse function defined in \eqref{eq:definition-inverse-a-multi-one-shot} also satisfies 
\begin{eqnarray} \label{eq:definition-inverse-R-multi-source-cgf-version}
(1- \rho(a(R )) a(R ) + \phi(\rho(a(R )); P_{XY}|Q_Y) = R.
\end{eqnarray}

Similarly, by setting $Z = \log \frac{1}{P_{X|Y}(X|Y)}$, we have
\begin{eqnarray} \label{eq:relation-single-shot-downarrow-phi}
\theta H_{1+\theta}^\downarrow(X|Y) = - \phi(-\theta; P_{XY}|P_Y).
\end{eqnarray}
Then, the variance (cf.~\eqref{eq:multi-single-shot-variance-1}) satisfies 
\begin{eqnarray} \label{eq:variance-second-derivative-cgf}
\san{V}(X|Y) = \phi^{\prime\prime}(0; P_{XY}|P_Y).
\end{eqnarray}

Let $\phi(\rho,\rho^\prime)$ be the CGF of $Z = \log \frac{P_Y^{(1-\rho^\prime)}(Y)}{P_{XY}(X,Y)}$
(cf.~\eqref{eq:single-shot-optimal-conditioning-distribution} for the definition of $P_Y^{(1-\rho^\prime)}$).
Then, we have
\begin{eqnarray} \label{eq:relation-single-shot-two-parameter-phi}
\theta H_{1+\theta,1+\theta^\prime}(X|Y) = - \phi(-\theta,-\theta^\prime).
\end{eqnarray}
It should be noted that $\phi(\rho,\rho^\prime)$ is a CGF for fixed $\rho^\prime$, 
but $\phi(\rho,\rho)$ cannot be treated as a CGF.

\subsubsection{Transition Matrix}

For transition matrix $W(x,y|x^\prime,y^\prime)$, we consider the function given by
\begin{eqnarray}
g((x,y),(x^\prime,y^\prime)) := \log \frac{W(y|y^\prime)}{W(x,y|x^\prime,y^\prime)}.
\end{eqnarray}
Then, the relation between the CGF and the lower conditional R\'enyi entropy is given by
\begin{eqnarray} \label{eq:relation-down-conditional-renyi-moment-transition-matrix}
\theta H_{1+\theta}^{\downarrow, W}(X|Y) = - \phi(-\theta).
\end{eqnarray}
Then, the variance defined in \eqref{eq:lower-conditional-renyi-markov-theta-0-derivative} satisfies
\begin{eqnarray} \label{eq:relation-variance-cgf-transition-matrix}
\san{V}^W(X|Y) = \phi^{\prime\prime}(0).
\end{eqnarray}

%% file: Appendix-Multi-Entropy.tex
\subsection{Proof of Lemma \ref{lemma:property-upper-conditional-renyi-single-shot}}
\label{Appendix:lemma:property-upper-conditional-renyi-single-shot}

We use the following lemma.
\begin{lemma} \label{lemma:appendix-upper-conditional-renyi-upper-and-lower}
For $\theta \in (-1,0) \cup (0,1)$, we have
\begin{eqnarray} \label{eq:appendix-upper-conditional-renyi-upper-and-lower}
H_{\frac{1}{1-\theta}}^\downarrow(X|Y) \le H_{\frac{1}{1-\theta}}^\uparrow(X|Y) \le H_{1+\theta}^\downarrow(X|Y).
\end{eqnarray}
\end{lemma}
\begin{proof}
The left hand side inequality of \eqref{eq:appendix-upper-conditional-renyi-upper-and-lower} is obvious from
the definition of two R\'enyi entropies (the latter is defined by taking maximum). 
The right hand side inequality was proved in \cite[Lemma 6]{hayashi:13}.
\end{proof}

Now, we go back to the proof of Lemma \ref{lemma:property-upper-conditional-renyi-single-shot}.
From \eqref{eq:down-conditional-renyi-theta-0} and \eqref{eq:multi-single-shot-variance-1}, 
by the Taylor approximation, we have
\begin{eqnarray}
H_{1+\theta}^\downarrow(X|Y) = H(X|Y) - \frac{1}{2} \san{V}(X|Y) \theta + o(\theta).
\end{eqnarray}
Furthermore, since $\frac{1}{1-\theta} = 1 + \theta + o(\theta)$, we also have
\begin{eqnarray}
H_{\frac{1}{1-\theta}}^\downarrow(X|Y) = H(X|Y) - \frac{1}{2} \san{V}(X|Y) \theta + o(\theta).
\end{eqnarray}
Thus, from Lemma \ref{lemma:appendix-upper-conditional-renyi-upper-and-lower}, 
we can derive \eqref{eq:up-conditional-renyi-theta-0} and \eqref{eq:multi-single-shot-variance-2}. \qed

\subsection{Proof of Lemma \ref{lemma:multi-terminal-single-shot-property}} 
\label{Appendix:lemma:multi-terminal-single-shot-property}

Statements \ref{item:multi-terminal-single-shot-property-1} and \ref{item:multi-terminal-single-shot-property-2}
follow from the relationships in \eqref{eq:relation-single-shot-q-phi}
and \eqref{eq:relation-single-shot-downarrow-phi} and strict convexity of
the CGFs. 

To prove Statement \ref{item:multi-terminal-single-shot-property-3}, we first prove 
strict convexity of the Gallager function 
\begin{eqnarray}
E_0(\tau; P_{XY}) := \log \sum_y P_Y(y) \left( \sum_x P_{X|Y}(x|y)^{\frac{1}{1+\tau}} \right)^{1+\tau}
\end{eqnarray}
for $\tau > -1$. We use the H\"older inequality:
\begin{eqnarray} \label{eq:Holder}
\sum_i a_i^{\alpha} b_i^\beta \le \left( \sum_i a_i \right)^\alpha \left( \sum_i b_i \right)^\beta
\end{eqnarray}
for  $\alpha,\beta > 0$ such that $\alpha + \beta =1$, where 
the equality holds iff. $a_i = c b_i$ for some constant $c$.
For $\lambda \in (0,1)$, let $1+\tau_3 = \lambda (1+\tau_1) + (1-\lambda) (1+\tau_2)$,
which implies 
\begin{eqnarray}
\frac{1}{1+\tau_3} = \frac{1}{1+\tau_1} \frac{\lambda (1+\tau_1)}{1+\tau_3} + \frac{1}{1+\tau_2} \frac{(1-\lambda)(1+\tau_2)}{1+\tau_3}
\end{eqnarray}
and
\begin{eqnarray}
\frac{\lambda (1+\tau_1)}{1+\tau_3} + \frac{(1-\lambda)(1+\tau_2)}{1+\tau_3} = 1.
\end{eqnarray}
Then, by applying the H\"older inequality twice, we have
\begin{eqnarray}
\lefteqn{ \sum_y P_Y(y) \left( \sum_x P_{X|Y}(x|y)^{\frac{1}{1+\tau_3}} \right)^{1+\tau_3} } \\
&=& \sum_y P_Y(y) \left( \sum_x P_{X|Y}(x|y)^{\frac{1}{1+\tau_1} \frac{\lambda(1+\tau_1)}{1+\tau_3}}
 P_{X|Y}(x|y)^{\frac{1}{1+\tau_2} \frac{(1-\lambda)(1+\tau_2)}{1+\tau_3}} \right)^{1+\tau_3} \\
&\le& \sum_y P_Y(y) \left[ \left( \sum_x P_{X|Y}(x|y)^{\frac{1}{1+\tau_1}} \right)^{\frac{\lambda(1+\tau_1)}{1+\tau_3}}
 \left( \sum_x P_{X|Y}(x|y)^{\frac{1}{1+\tau_2}} \right)^{\frac{(1-\lambda)(1+\tau_2)}{1+\tau_3}} \right]^{1+\tau_3} 
  \label{eq:equality-condition-Gallager-1} \\
 &=&  \sum_y P_Y(y) \left( \sum_x P_{X|Y}(x|y)^{\frac{1}{1+\tau_1}} \right)^{\lambda (1+\tau_1)}
  \left( \sum_x P_{X|Y}(x|y)^{\frac{1}{1+\tau_2}} \right)^{(1-\lambda)(1+\tau_2)} \\
 &=&  \sum_y P_Y(y)^\lambda \left( \sum_x P_{X|Y}(x|y)^{\frac{1}{1+\tau_1}} \right)^{\lambda (1+\tau_1)}
  P_Y(y)^{1-\lambda} \left( \sum_x P_{X|Y}(x|y)^{\frac{1}{1+\tau_2}} \right)^{(1-\lambda)(1+\tau_2)} \\
 &\le& \left[ \sum_y P_Y(y) \left( \sum_x P_{X|Y}(x|y)^{\frac{1}{1+\tau_1}} \right)^{(1+\tau_1)} \right]^{\lambda}
   \left[ \sum_y P_Y(y) \left( \sum_x P_{X|Y}(x|y)^{\frac{1}{1+\tau_2}} \right)^{(1+\tau_2)} \right]^{1-\lambda}.
   \label{eq:equality-condition-Gallager-2}
\end{eqnarray}
The equality in the second inequality holds iff. 
\begin{eqnarray} \label{eq:condition-Gallager-function-strict-inequality}
\left( \sum_x P_{X|Y}(x|y)^{\frac{1}{1+\tau_1}} \right)^{1+\tau_1}
 = c \left( \sum_x P_{X|Y}(x|y)^{\frac{1}{1+\tau_2}} \right)^{1+\tau_2}~~\forall y \in {\cal Y} 
\end{eqnarray}
for some constant $c$. Futhermore, the equality in the first inequality holds iff. 
$P_{X|Y}(x|y) = \frac{1}{|\rom{supp}(P_{X|Y}(\cdot|y))|}$. Substituting this into \eqref{eq:condition-Gallager-function-strict-inequality},
we find that $|\rom{supp}(P_{X|Y}(\cdot|y))|$ is irrespective of $y$. Thus, both the equalities hold 
simultaneously iff. $\san{V}(X|Y) = 0$.
Now, since 
\begin{eqnarray} \label{eq:relation-upper-renyi-gallager}
\theta H_{1+\theta}^\uparrow(X|Y) = - (1+\theta) E_0\left( \frac{-\theta}{1+\theta} ; P_{XY} \right),
\end{eqnarray}
we have
\begin{eqnarray} \label{eq:relation-up-conditional-renyi-gallager-second-derivative}
\frac{d^2 [\theta H_{1+\theta}^\uparrow(X|Y)]}{d\theta^2} &=&
 - \frac{1}{(1+\theta)^4} E_0^{\prime\prime}\left( \frac{-\theta}{1+\theta} ; P_{XY} \right) \\
 &\le& 0
\end{eqnarray}
for $\theta \in (-1, \infty)$, where the equality holds iff. $\san{V}(X|Y) = 0$.

Statement \ref{item:multi-terminal-single-shot-property-4} is obvious from the definitions of the two measures.
The first part of Statement \ref{item:multi-terminal-single-shot-property-5} follows from
\eqref{eq:relation-single-shot-two-parameter-phi} and convexity of the CGF, but we need 
another argument to check the conditions for strict concavity.  
Since the second term of
\begin{eqnarray}
\theta H_{1+\theta,1+\theta^\prime}(X|Y) = 
- \log \sum_y P_Y(y) \left[ \sum_x P_{X|Y}(x|y)^{1+\theta} \right] \left[\sum_x P_{X|Y}(x|y)^{1+\theta^\prime} \right]^{\frac{\theta}{1+\theta^\prime}}
 + \frac{\theta \theta^\prime}{1+\theta^\prime} H_{1+\theta^\prime}^\uparrow(X|Y)
 \label{eq:strict-concavity-of-two-parameter-renyi-single-shot}
\end{eqnarray}
is linear with respect to $\theta$, it suffice to show strict concavity of the first term.
By using the H\"older inequality twice, for $\theta_3 = \lambda \theta_1 + (1-\lambda) \theta_2$, we have
\begin{eqnarray}
\lefteqn{
\sum_y P_Y(y) \left[ \sum_x P_{X|Y}(x|y)^{1+\theta_3} \right] \left[\sum_x P_{X|Y}(x|y)^{1+\theta^\prime} \right]^{\frac{\theta_3}{1+\theta^\prime}}
} \\
&\le& \sum_y P_Y(y) \left[ \sum_x P_{X|Y}(x|y)^{1+\theta_1} \right]^\lambda
 \left[ \sum_x P_{X|Y}(x|y)^{1+\theta_2} \right]^{1-\lambda}
  \left[\sum_x P_{X|Y}(x|y)^{1+\theta^\prime} \right]^{\frac{\lambda \theta_1 + (1-\lambda) \theta_2}{1+\theta^\prime}} \\
 &\le& \left[ \sum_y P_Y(y) \left[ \sum_x P_{X|Y}(x|y)^{1+\theta_1} \right] 
   \left[\sum_x P_{X|Y}(x|y)^{1+\theta^\prime} \right]^{\frac{\theta_1}{1+\theta^\prime}} \right]^\lambda \\
&& ~~~~  \left[ \sum_y P_Y(y) \left[ \sum_x P_{X|Y}(x|y)^{1+\theta_2} \right] 
   \left[\sum_x P_{X|Y}(x|y)^{1+\theta^\prime} \right]^{\frac{\theta_2}{1+\theta^\prime}} \right]^{1-\lambda},
\end{eqnarray}
where both the equalities hold simultaneously iff. $\san{V}(X|Y) = 0$, which can be proved 
in a similar manner as the equality conditions in \eqref{eq:equality-condition-Gallager-1}
and \eqref{eq:equality-condition-Gallager-2}. Thus we have the latter part of Statement \ref{item:multi-terminal-single-shot-property-5}.

Statements \ref{item:multi-terminal-single-shot-property-6}-\ref{item:multi-terminal-single-shot-property-8}
are also obvious from the definitions. 
Statements \ref{item:multi-terminal-single-shot-property-1-b}, \ref{item:multi-terminal-single-shot-property-2-b},
\ref{item:multi-terminal-single-shot-property-3-b}, \ref{item:multi-terminal-single-shot-property-5-b},
follows from
Statements \ref{item:multi-terminal-single-shot-property-1}, \ref{item:multi-terminal-single-shot-property-2},
\ref{item:multi-terminal-single-shot-property-3}, \ref{item:multi-terminal-single-shot-property-5},
(cf.~\cite[Lemma 1]{hayashi:13}).
\qed

\subsection{Proof of Lemma \ref{lemma:limit-renyi-max-entropy}}
\label{appendix:lemma:limit-renyi-max-entropy}

Since \eqref{eq:limit-renyi-max-entropy-0} and \eqref{eq:limit-renyi-max-entropy-2} are obvious from the definitions,
we only prove \eqref{eq:limit-renyi-max-entropy-1}. We note that 
\begin{eqnarray}
\lefteqn{
\left[  \sum_y P_Y(y) \left[ \sum_x P_{X|Y}(x|y)^{1+\theta} \right]^{\frac{1}{1+\theta}} \right]^{1+\theta}
} \\
&\le& \left[ \sum_y P_Y(y) | \rom{supp}(P_{X|Y}(\cdot|y))|^{\frac{1}{1+\theta}} \right]^{1+\theta} \\
&\le& \max_{y \in \rom{supp}(P_Y)} | \rom{supp}(P_{X|Y}(\cdot|y))|
\end{eqnarray}
and
\begin{eqnarray}
\lefteqn{
\left[  \sum_y P_Y(y) \left[ \sum_x P_{X|Y}(x|y)^{1+\theta} \right]^{\frac{1}{1+\theta}} \right]^{1+\theta}
} \\
&\ge& P_Y(y^*)^{1+\theta} \left[ \sum_x P_{X|Y}(x|y^*)^{1+\theta} \right] \\
&\stackrel{\theta \to -1}{\to}& |\rom{supp}(P_{X|Y}(\cdot|y^*))|,
\end{eqnarray}
where 
\begin{eqnarray}
y^* := \argmax_{y \in \rom{supp}(P_Y)} | \rom{supp}(P_{X|Y}(\cdot|y))|.
\end{eqnarray}
\qed

\subsection{Proof of Lemma \ref{lemma:properties-upper-conditional-renyi-transition-matrix}}
\label{appendix:lemma:properties-upper-conditional-renyi-transition-matrix}

From Lemma \ref{lemma:appendix-upper-conditional-renyi-upper-and-lower},
Theorem \ref{theorem:asymptotic-down-conditional-renyi}, and Theorem \ref{theorem:asymptotic-up-conditional-renyi},
we have
\begin{eqnarray} 
H_{\frac{1}{1-\theta}}^{\downarrow,W}(X|Y) \le 
H_{\frac{1}{1-\theta}}^{\uparrow,W}(X|Y) \le 
H_{1+\theta}^{\downarrow,W}(X|Y)
\end{eqnarray}
for $\theta \in (-1,0) \cup (0,1)$. Thus, we can prove Lemma \ref{lemma:properties-upper-conditional-renyi-transition-matrix}
in the same manner as Lemma \ref{lemma:property-upper-conditional-renyi-single-shot}. \qed

\subsection{Proof of \eqref{eq:alternative-definition-of-upper-conditional-W}} 
\label{appendix:eq:alternative-definition-of-upper-conditional-W}

First, in the same manner as Theorem \ref{theorem:asymptotic-down-conditional-renyi}, we can show
\begin{eqnarray}
\lim_{n \to \infty} \frac{1}{n} H_{1+\theta}(P_{X^nY^n}|Q_{Y^n}) = H_{1+\theta}^{W|V}(X|Y),
\end{eqnarray}
where $Q_{Y^n}$ is a Markov chain induced by $V$ for some initial distribution.
Then, since $H_{1+\theta}(P_{X^nY^n}|Q_{Y^n}) \le H_{1+\theta}^\uparrow(X^n|Y^n)$ for each $n$, 
by using Theorem \ref{theorem:asymptotic-up-conditional-renyi}, we have
\begin{eqnarray}
H_{1+\theta}^{W|V}(X|Y) \le H_{1+\theta}^{\uparrow,W}(X|Y).
\end{eqnarray}
Thus, the rest of the proof is to show that $H_{1+\theta}^{\uparrow,W}(X|Y)$ is attainable by some $V$.

Let $\hat{Q}_\theta$ be the normalized left eigenvector of $K_\theta$, and let 
\begin{eqnarray}
V_\theta(y|y^\prime) := \frac{\hat{Q}_\theta(y)}{\kappa_\theta \hat{Q}_\theta(y^\prime)} K_\theta(y|y^\prime).
\end{eqnarray}
Then, $V_\theta$ attains the maximum. To prove this, we will show that $\kappa_\theta^{1+\theta}$ is the 
Perron-Frobenius eigenvalue of 
\begin{eqnarray} \label{eq:w-v-for-optimal-v}
W(x,y|x^\prime,y^\prime)^{1+\theta} V_\theta(y|y^\prime)^{-\theta}.
\end{eqnarray}
We first confirm that $(\hat{Q}_\theta(y)^{1+\theta} : (x,y) \in {\cal X} \times {\cal Y})$ is an eigenvector of
\eqref{eq:w-v-for-optimal-v} as follows:
\begin{eqnarray}
\lefteqn{
\sum_{x,y} \hat{Q}_\theta(y)^{1+\theta} W(x,y|x^\prime,y^\prime)^{1+\theta} V_\theta(y|y^\prime)^{-\theta} 
} \\
&=& \sum_y \hat{Q}_\theta(y)^{1+\theta} W_\theta(y|y^\prime) 
 \left[ \frac{\hat{Q}_\theta(y)}{\kappa_\theta \hat{Q}_\theta(y^\prime)} W_\theta(y|y^\prime)^{\frac{1}{1+\theta}} \right]^{-\theta} \\
&=& \kappa_\theta^\theta \hat{Q}_\theta(y^\prime)^{\theta} \sum_y \hat{Q}_\theta(y) W_\theta(y|y^\prime)^{\frac{1}{1+\theta}} \\
&=& \kappa_\theta^{1+\theta} \hat{Q}_{\theta}(y^\prime)^{1+\theta}.
\end{eqnarray}
Since $(\hat{Q}_\theta(y)^{1+\theta} : (x,y) \in {\cal X} \times {\cal Y})$ is a positive vector
and the Perron-Frobenius eigenvector is the unique positive eigenvector, we find that 
$\kappa_\theta^{1+\theta}$ is the Perron-Frobenius eigenvalue. Thus, we have
\begin{eqnarray}
H_{1+\theta}^{W|V_\theta}(X|Y) &=& - \frac{1+\theta}{\theta} \log \kappa_\theta \\
&=& H_{1+\theta}^{\uparrow,W}(X|Y).
\end{eqnarray}
\qed

\subsection{Proof of Lemma \ref{lemma:multi-terminal-markov-property}}
\label{appendix:lemma:multi-terminal-markov-property}

Statement \ref{item:multi-terminal-markov-property-1} follows from \eqref{eq:relation-down-conditional-renyi-moment-transition-matrix}
and strict convexity of the CGF.
Statements \ref{item:multi-terminal-markov-property-3}, \ref{item:multi-terminal-markov-property-5},
\ref{item:multi-terminal-markov-property-6}, and \ref{item:multi-terminal-markov-property-7}, 
follow from the corresponding statements 
in Lemma \ref{lemma:multi-terminal-single-shot-property},
Theorem \ref{theorem:asymptotic-down-conditional-renyi},
Theorem \ref{theorem:asymptotic-up-conditional-renyi},
and Theorem \ref{theorem:asymptotic-two-parameter-renyi-markov}.

Now, we prove\footnote{The concavity of $\theta H_{1+\theta}^{\uparrow,W}(X|Y)$ 
follows from the limiting argument, i.e.,
the concavity of $\theta H_{1+\theta}^{\uparrow}(X^n|Y^n)$ (cf.~Lemma \ref{lemma:multi-terminal-single-shot-property})
and Theorem \ref{theorem:asymptotic-up-conditional-renyi}.
However, the strict concavity does not follows from the limiting argument.} 
Statement \ref{item:multi-terminal-markov-property-2}.
For this purpose, we introduce transition matrix counterpart of the Gallager function as follows.
Let
\begin{eqnarray}
\bar{K}_\tau(y|y^\prime) := W(y|y^\prime) \left[ \sum_x W(x|x^\prime,y^\prime,y)^{\frac{1}{1+\tau}} \right]^{1+\tau}
\end{eqnarray}
for $\tau > -1$, which is well defined under Assumption \ref{assumption-memory-through-Y}.
Let $\bar{\kappa}_\tau$ be the Perron-Frobenius eigenvalue of $\bar{K}_\tau$, and let
$\tilde{Q}_\tau$ and $\hat{Q}_\tau$ be its normalized right and left eigenvectors. 
Then, let 
\begin{eqnarray}
L_\tau(y|y^\prime) := \frac{\hat{Q}_\tau(y)}{\bar{\kappa}_\tau \hat{Q}_\tau(y^\prime)} \bar{K}_\tau(y|y^\prime)
\end{eqnarray}
be a parametrized transition matrix. The stationary distribution of $L_\tau$ is given by
\begin{eqnarray}
Q_\tau(y^\prime) := \frac{\hat{Q}_\tau(y^\prime) \tilde{Q}_\tau(y^\prime)}{\sum_{y^{\prime\prime}} \hat{Q}_\tau(y^{\prime\prime}) \tilde{Q}_\tau(y^{\prime\prime})}.
\end{eqnarray}
We prove strict convexity of $E_0^W(\tau) := \log \bar{\kappa}_\tau$ for $\tau > -1$.
Then, by the same reason as \eqref{eq:relation-up-conditional-renyi-gallager-second-derivative}, 
we can show Statement \ref{item:multi-terminal-markov-property-2}.
Let $Q_\tau(y,y^\prime) := L_\tau(y|y^\prime) Q_\tau(y^\prime)$. 
By the same calculation as \cite[Proof of Lemma 13 and Lemma 14]{hayashi-watanabe:13},
we have
\begin{eqnarray}
\sum_{y,y^\prime} Q_\tau(y,y^\prime) \left[ \frac{d}{d\tau} \log L_\tau(y|y^\prime) \right]^2
 = - \sum_{y,y^\prime} Q_\tau(y,y^\prime) \left[ \frac{d^2}{d\tau^2} \log L_\tau(y|y^\prime) \right].
\end{eqnarray} 
Furthermore, from the definition of $L_\tau$, we have
\begin{eqnarray}
\lefteqn{ - \sum_{y,y^\prime} Q_\tau(y,y^\prime) \left[ \frac{d^2}{d\tau^2} \log L_\tau(y|y^\prime) \right] }  \\
&=& - \sum_{y,y^\prime} Q_\tau(y,y^\prime) \left[ \frac{d^2}{d \tau^2} \log \frac{1}{\kappa_\tau} 
  + \frac{d^2}{d\tau^2} \log \frac{\hat{Q}_\tau(y)}{\hat{Q}_\tau(y^\prime)} + \frac{d^2}{d\tau^2} \log K_\tau(y|y^\prime) \right] \\
&=& \frac{d^2}{d\tau^2} \log \kappa_\tau -  \sum_{y,y^\prime} Q_\tau(y,y^\prime) \frac{d^2}{d\tau^2} \log K_\tau(y|y^\prime).
\end{eqnarray}
Now, we show convexity of $\log \bar{K}_\tau(y|y^\prime)$ for each $(y,y^\prime)$.
By using the H\"older inequality (cf.~Appendix \ref{Appendix:lemma:multi-terminal-single-shot-property}),
for $\tau_3 = \lambda \tau_1 + (1-\lambda) \tau_2$, we have
\begin{eqnarray} \label{eq:convexity-markov-Gallager}
\left[ \sum_x W(x|x^\prime,y^\prime,y)^{\frac{1}{1+\tau_3}} \right]^{1+\tau_3}
\le \left[ \sum_x W(x|x^\prime,y^\prime,y)^{\frac{1}{1+\tau_1}} \right]^{\lambda(1+\tau_1)}
 \left[ \sum_x W(x|x^\prime,y^\prime,y)^{\frac{1}{1+\tau_2}} \right]^{(1-\lambda)(1+\tau_2)}.
\end{eqnarray}
Thus, $E_0^W(\tau)$ is convex. To check strict convexity, we note that the equality in \eqref{eq:convexity-markov-Gallager}
holds iff. $W(x|x^\prime,y^\prime,y) = \frac{1}{|\rom{supp}(W(\cdot|x^\prime,y^\prime,y))|}$.
Since
\begin{eqnarray}
\sum_x W(x|x^\prime,y^\prime,y)^{1+\theta} = \frac{1}{|\rom{supp}(W(\cdot|x^\prime,y^\prime,y))|^\theta}
\end{eqnarray}
does not depend on $x^\prime$ from Assumption \ref{assumption-memory-through-Y}, we have 
$|\rom{supp}(W(\cdot|x^\prime,y^\prime,y))| = C_{yy^\prime}$ for some integer $C_{yy^\prime}$.
By substituting this into $\bar{K}_\tau$, we have
\begin{eqnarray}
\bar{K}_\tau(y|y^\prime) = W(y|y^\prime) C_{yy^\prime}^\tau.
\end{eqnarray}
On the other hand, we note that the CGF $\phi(\rho)$ is defined as the logarithm of the Perron-Frobenius eigenvalue of
\begin{eqnarray} \label{eq:proof-strict-convexity-two-parameter}
W(x,y|w^\prime,y^\prime)^{1-\rho} W(y|y^\prime)^\rho 
= W(y|y^\prime) \frac{1}{C_{yy^\prime}^{1-\rho}} \bol{1}[x \in \rom{supp}(W(\cdot|x^\prime,y^\prime,y)) ].
\end{eqnarray}
Since 
\begin{eqnarray}
\lefteqn{ \sum_{x,y} \hat{Q}_\tau(y) W(y|y^\prime) \frac{1}{C_{yy^\prime}^{1-\tau}} \bol{1}[x \in \rom{supp}(W(\cdot|x^\prime,y^\prime,y)) ] } \\
&=& \sum_y  \hat{Q}_\tau(y) W(y|y^\prime) C_{yy^\prime}^\tau \\
&=& \bar{\kappa}_\tau \hat{Q}_\tau(y^\prime),
\end{eqnarray}
$\bar{\kappa}_\tau$ is the Perro-Frobenius eigenvalue of \eqref{eq:proof-strict-convexity-two-parameter}, and thus
we have $E_0^W(\tau) = \phi(\tau)$ when the equality in \eqref{eq:convexity-markov-Gallager} 
holds for every $(y,y^\prime)$ such that $W(y|y^\prime) > 0$. 
Since $\phi(\tau)$ is strict convex if $\san{V}^W(X|Y) > 0$,
$E_0^W(\tau)$ is strict convex if $\san{V}^W(X|Y) > 0$. 
Thus, $\theta H_{1+\theta}^{\uparrow,W}(X|Y)$ is strict concave if $\san{V}^W(X|Y) > 0$.
On the other hand, from
\eqref{eq:upper-conditional-renyi-transition-variance}, 
$\theta H_{1+\theta}^{\uparrow,W}(X|Y)$ is strict concave only if $\san{V}^W(X|Y) > 0$.

Statement \ref{item:multi-terminal-markov-property-4} can be proved by modifying 
the proof of Statement \ref{item:multi-terminal-single-shot-property-5} 
of Lemma \ref{lemma:multi-terminal-single-shot-property} to a transition matrix
in a similar manner as Statement \ref{item:multi-terminal-markov-property-2} of the present lemma.

Finally, Statements \ref{item:multi-terminal-markov-property-1-b}, \ref{item:multi-terminal-markov-property-2-b},
\ref{item:multi-terminal-markov-property-4-b} follows from
Statements \ref{item:multi-terminal-markov-property-1}, \ref{item:multi-terminal-markov-property-2},
\ref{item:multi-terminal-markov-property-4}
(cf.~\cite[Lemma 1]{hayashi:13}).
\qed

\subsection{Proof of Lemma \ref{lemma:legendra-transform-2}}
\label{subsection:proof-lemma:legendra-transform-2}

We only prove \eqref{eq:legendra-transform-2-lower} since we can prove \eqref{eq:legendra-transform-2-upper}
exactly in the same manner by replacing $H_{1+\theta}^{\downarrow,W}(X|Y)$, $\theta^\downarrow(a)$, and $a^\downarrow(R )$
by $H_{1+\theta}^{\uparrow,W}(X|Y)$, $\theta^\uparrow(a)$, and $a^\uparrow(R )$. 
Let 
\begin{align}
f(\theta) := \frac{- \theta R + \theta H_{1+\theta}^{\downarrow,W}(X|Y)}{1+\theta}.
\end{align}
Then, we have
\begin{align}
f^\prime(\theta) &= \frac{-R + (1+\theta) \frac{d[\theta H_{1+\theta}^{\downarrow,W}(X|Y)]}{d\theta} - \theta H_{1+\theta}^{\downarrow,W}(X|Y)}{(1+\theta)^2} \\
&= \frac{- R + R\left( \frac{d[\theta H_{1+\theta}^{\downarrow,W}(X|Y)]}{d\theta} \right)}{(1+\theta)^2}.
\end{align}
Since $R(a )$ is monotonically increasing and $\frac{d[\theta H_{1+\theta}^{\downarrow,W}(X|Y)]}{d\theta}$ 
is monotonically decreasing, we have
$f^\prime(\theta) \ge 0$ for $\theta \le \theta(a(R ))$ and 
$f^\prime(\theta) \le 0$ for $\theta \ge \theta(a(R ))$. Thus, $f(\theta)$ takes its maximum at $\theta(a( R))$.
Furthermore, since $-1 \le \theta(a(R )) \le 0$ for $H^W(X|Y) \le R \le H_0^{\downarrow,W}(X|Y)$, we have
\begin{align}
& \sup_{-1\le \theta \le 0} \frac{-\theta R + \theta H_{1+\theta}^{\downarrow,W}(X|Y)}{1+\theta} \\
 &= \frac{- \theta(a(R )) R + \theta(a(R )) H_{1+\theta(a(R ))}^{\downarrow,W}(X|Y)}{1+\theta(a(R ))} \\
 &= \frac{- \theta(a(R )) [(1+\theta(a(R ))) a(R ) - \theta(a(R )) H_{1+\theta(a(R ))}^{\downarrow,W}(X|Y)] + \theta(a(R )) H_{1+\theta(a(R ))}^{\downarrow,W}(X|Y)}{1+\theta(a(R ))} \\
 &= - \theta(a(R )) a(R ) + \theta(a(R )) H_{1+\theta(a(R ))}^{\downarrow,W}(X|Y),
\end{align}
where we substituted $R = R(a(R ))$ in the second equality. \qed

\subsection{Proof of Lemma \ref{lemma:multi-terminal-finite-evaluation-upper-conditional-renyi}}
\label{appendix:proof-lemma:multi-terminal-finite-evaluation-upper-conditional-renyi}

Let $u$ be the vector such that $u(y) =1$ for every $y \in {\cal Y}$. 
From the definition of $H_{1+\theta}^{\uparrow}(X^n|Y^n)$, we have
the following sequence of calculations:
\begin{eqnarray}
\lefteqn{ e^{- \frac{\theta}{1+\theta} \theta H_{1+\theta}^{\uparrow}(X^n|Y^n)} } \\
&=& \sum_{y_1,\ldots,y_n} \left[ \sum_{x_n,\ldots,x_1} P(x_1,y_1)^{1+\theta} 
 \prod_{i=2}^n W(x_i,y_i|x_{i-1},y_{i-1})^{1+\theta} \right]^{\frac{1}{1+\theta}} \\
&\stackrel{\rom{(a)}}{=}& \sum_{y_n,\ldots,y_1} \left[ \sum_{x_1} P(x_1,y_1)^{1+\theta}\right]^{\frac{1}{1+\theta}}
  \prod_{i=2}^n W_{\theta}(y_i|y_{i-1})^{\frac{1}{1+\theta}} \\
&=& \inner{u}{K_\theta^{n-1}w_\theta} \\
&\le& \inner{v_\tau}{K_\theta^{n-1} w_\theta} \\
&=& \inner{(K_\theta^T)^{n-1} v_\theta}{w_\theta} \\
&=& \kappa_\theta^{n-1} \inner{v_\theta}{w_\theta} \\
&=& e^{-(n-1)\frac{\theta}{1+\theta} H_{1+\theta}^{\uparrow,W}(X|Y)} \inner{v_\theta}{w_\theta},  
\end{eqnarray}
which implies the left hand side inequality,  where we used Assumption \ref{assumption-memory-through-Y} in $\rom{(a)}$.
On the other hand, we have the following sequence of calculations:
\begin{eqnarray}
\lefteqn{ e^{- \frac{\theta}{1+\theta} \theta H_{1+\theta}^{\uparrow}(X^n|Y^n)} } \\
&=& \inner{u}{K_\theta^{n-1}w_\theta} \\
&\ge& \frac{1}{\max_y v_\theta(y)} \inner{v_\theta}{K_\theta^{n-1} w_\theta} \\
&=& \frac{1}{\max_y v_\theta(y)} \inner{(K_\theta^T)^{n-1} v_\theta}{w_\theta} \\
&=& \kappa_\theta^{n-1} \frac{\inner{v_\theta}{w_\theta}}{\max_y v_\theta(y)} \\
&=& e^{-(n-1)\frac{\theta}{1+\theta} H_{1+\theta}^{\uparrow,W}(X|Y)} \frac{\inner{v_\theta}{w_\theta}}{\max_y v_\theta(y)},
\end{eqnarray}
which implies the right hand side inequality. \qed

%% file: Appendix-Multi-Source.tex
\subsection{Proof of Theorem \ref{theorem:multi-source-converse}}
\label{appendix:theorem:multi-source-converse}

For arbitrary $\tilde{\rho} \in \mathbb{R}$, we set $\alpha := P_{XY}\{ X \neq \san{d}(\san{e}(X),Y) \}$ and
$\beta: = P_{XY,\tilde{\rho}}\{ X \neq \san{d}(\san{e}(X),Y) \}$, where 
\begin{eqnarray}
P_{XY,\rho}(x,y) := P_{XY}(x,y)^{1-\rho} Q_Y(y)^\rho e^{-\phi(\rho;P_{XY}|Q_Y)}.
\end{eqnarray}
Then, by the monotonicity of the R\'enyi divergence, we have
\begin{eqnarray}
s D_{1+s}(P_{XY,\tilde{\rho}} \| P_{XY}) &\ge& \log \left[ \beta^{1+s} \alpha^{-s} + (1-\beta)^{1+s} (1-\alpha)^{-s} \right] \\
&\ge& \log \beta^{1+s} \alpha^{-s}.
\end{eqnarray}
Thus, we have
\begin{eqnarray}
- \log \alpha \le \frac{\phi((1+s) \tilde{\rho};P_{XY}|Q_Y) - (1+s) \phi(\tilde{\rho};P_{XY}|Q_Y) - (1+s) \log \beta}{s}.
\end{eqnarray}
Now, by using Lemma \ref{lemma:multi-source-han-converse}, we have
\begin{eqnarray}
1 - \beta &\le& P_{XY,\tilde{\rho}}\left\{ \log \frac{Q_Y(y)}{P_{XY,\tilde{\rho}}(x,y)} \le \gamma \right\} + \frac{M}{e^\gamma}.
\end{eqnarray}
We also have, for any $\sigma \le 0$, 
\begin{eqnarray}
\lefteqn{ P_{XY,\tilde{\rho}}\left\{ \log \frac{Q_Y(y)}{P_{XY,\tilde{\rho}}(x,y)} \le \gamma \right\} } \label{eq:multi-source-exponential-converse-proof-1} \\
&\le& \sum_{x,y} P_{XY,\tilde{\rho}}(x,y) e^{\sigma \left(\log \frac{Q_Y(y)}{P_{XY,\tilde{\rho}}(x,y)} - \gamma \right)} \\
&=& e^{- [ \sigma \gamma - \phi(\sigma; P_{XY,\tilde{\rho}}|Q_Y)]}.
\end{eqnarray}
Thus, by setting $\gamma$ so that
\begin{eqnarray}
\sigma \gamma - \phi(\sigma; P_{XY,\tilde{\rho}}|Q_Y) = \gamma - R,
\end{eqnarray}
we have
\begin{eqnarray}
1 - \beta \le 2 e^{- \frac{\sigma R - \phi(\sigma; P_{XY,\tilde{\rho}}|Q_Y)}{1-\sigma}}.
\label{eq:multi-source-exponential-converse-proof-1-5}
\end{eqnarray}
Furthermore, we have the relation
\begin{eqnarray}
\phi(\sigma; P_{XY,\tilde{\rho}}|Q_Y)
&=& \log \sum_{x,y} P_{XY,\tilde{\rho}}(x,y)^{1-\sigma} Q_Y(y)^\sigma \\
&=& \log \sum_{x,y} \left( P_{XY}(x,y)^{1-\tilde{\rho}} Q_Y(y)^{\tilde{\rho}} e^{- \phi(\tilde{\rho}; P_{XY}|Q_Y)}\right)^{1-\sigma} Q_Y(y)^\sigma \\
&=& - (1-\sigma) \phi(\tilde{\rho}; P_{XY}|Q_Y) + \log \sum_{x,y} P_{XY}(x,y)^{1-\tilde{\rho} - \sigma(1-\tilde{\rho})} Q_Y(y)^{\tilde{\rho}+\sigma(1-\tilde{\rho})} \\
&=& \phi(\tilde{\rho}+\sigma(1-\tilde{\rho}); P_{XY}|Q_Y) - (1-\sigma) \phi(\tilde{\rho}; P_{XY}|Q_Y).
\label{eq:multi-source-exponential-converse-proof-2}
\end{eqnarray}
Thus, by substituting $\tilde{\rho} = - \tilde{\theta}$ and $\sigma = - \vartheta$, and by using \eqref{eq:relation-single-shot-q-phi}, 
we can derive \eqref{eq:bound-multi-source-exponential-converse-1}. 

Now, we restrict the range of $\tilde{\rho}$ so that $\rho(a(R )) < \tilde{\rho} < 1$, and take 
\begin{eqnarray}
\sigma = \frac{\rho(a(R )) - \tilde{\rho}}{1-\tilde{\rho}}.
\label{eq:multi-source-exponential-converse-proof-3}
\end{eqnarray}
Then, by substituting this into \eqref{eq:multi-source-exponential-converse-proof-2}
and \eqref{eq:multi-source-exponential-converse-proof-2} into \eqref{eq:multi-source-exponential-converse-proof-1-5},
we have ($\phi(\rho;P_{XY}|Q_Y)$ is omitted as $\phi(\rho)$)
\begin{eqnarray} 
\lefteqn{ \frac{\sigma R - \phi(\tilde{\rho} 
 + \sigma(1-\tilde{\rho})) + (1-\sigma)\phi(\tilde{\rho})}{1-\sigma} }  \label{eq:multi-source-proof-converting-exponent-4} \\
&=& \frac{(\rho(a(R )) -  \tilde{\rho}) R - (1-\tilde{\rho}) \phi(\rho(a(R )) ) + (1-\rho(a(R ))) \phi(\tilde{\rho})}{1 - \rho(a(R ))} \\
&=& \frac{(\rho(a(R )) - \tilde{\rho}) \left\{ (1-\rho(a(R ))) a(R ) + \phi(\rho(a(R ))) \right\}  - (1-\tilde{\rho}) \phi(\rho(a(R ))) + (1-\rho(a(R ))) \phi(\tilde{\rho})}{1 - \rho(a(R ))} \nonumber \\ \\
&=& (\rho(a(R )) - \tilde{\rho}) a(R ) - \phi(\rho(a(R ))) + \phi(\tilde{\rho}), 
\label{eq:multi-source-proof-converting-exponent-5}
\end{eqnarray} 
where we used \eqref{eq:definition-inverse-R-multi-source-cgf-version} in the second equality.
Thus, by substituting $\tilde{\rho} = - \tilde{\theta}$ and by using \eqref{eq:relation-single-shot-q-phi} again,
we have \eqref{eq:bound-multi-source-exponential-converse-2}. \qed

%% file: Appendix-Channel.tex
\subsection{Proof of Theorem \ref{theorem:channel-exponential-converse}}
\label{appendix:theorem:channel-exponential-converse}

Let
\begin{eqnarray}
P_{X^n Y^n, \rho}(x^n,y^n) := P_{X^nY^n}(x^n,y^n)^{1-\rho} Q_{Y^n}(y^n)^{\rho} e^{-\phi(\rho;P_{X^nY^n}|Q_{Y^n})},
\end{eqnarray}
and let $P_{B^n|A^n,\rho}$ be a conditional additive channel defined by
\begin{eqnarray}
P_{B^n|A^n,\rho}(a^n+x^n|a^n) = P_{X^nY^n,\rho}(x^n,y^n).
\end{eqnarray}
We also define the joint distribution of the message, the input, the output, and the decoded message for each channel:
\begin{eqnarray}
P_{M_n A^n B^n\hat{M}_n}(m,a^n,b^n,\hat{m}) &:=& \frac{1}{M_n} \bol{1}[\san{e}_n(m) = a^n] P_{B^n|A^n}(b^n|a^n) \bol{1}[\san{d}_n(b^n) = \hat{m}], \\
P_{M_n A^n B^n\hat{M}_n,\rho}(m,a^n,b^n,\hat{m}) &:=& 
 \frac{1}{M_n} \bol{1}[\san{e}_n(m) = a^n] P_{B^n|A^n,\rho}(b^n|a^n) \bol{1}[\san{d}_n(b^n) = \hat{m}].
\end{eqnarray}
For arbitrary $\tilde{\rho} \in \mathbb{R}$, let 
$\alpha := P_{M_n\hat{M}_n}\{ m \neq \hat{m} \}$ and $\beta := P_{M_n \hat{M}_n,\tilde{\rho}}\{ m \neq \hat{m} \}$.
Then, by the monotonicity of the R\'enyi divergence, we have
\begin{eqnarray}
s D_{1+s}(P_{A^nB^n,\tilde{\rho}} \| P_{A^nB^n}) 
&\ge& s D_{1+s}(P_{M_n \hat{M}_n,\tilde{\rho}} \| P_{M_n\hat{M}_n}) \\
&\ge& \log \left[ \beta^{1+s} \alpha^{-s} + (1-\beta)^{1+s} (1-\alpha)^{-s} \right] \\
&\ge& \log \beta^{1+s} \alpha^{-s}.
\end{eqnarray}
Thus, we have
\begin{eqnarray}
- \log \alpha \le \frac{s D_{1+s}(P_{A^nB^n,\tilde{\rho}} \| P_{A^nB^n}) - (1+s) \log \beta}{s}.
\end{eqnarray}
Here, we have
\begin{eqnarray}
D_{1+s}(P_{A^nB^n,\tilde{\rho}} \| P_{A^nB^n}) = D_{1+s}(P_{X^nY^n,\tilde{\rho}} \| P_{X^nY^n}).
\end{eqnarray}
On the other hand, from Lemma \ref{lemma:channel-converse-general-additive-specialized}, we have
\begin{eqnarray}
1 - \beta \le P_{X^nY^n,\tilde{\rho}}\left\{ \log \frac{Q_{Y^n}(y^n)}{P_{X^nY^n,\tilde{\rho}}(x^n,y^n)} \le n \log |{\cal A}| - \gamma \right\} + \frac{e^R}{e^{n \log|{\cal A}| - \gamma}}.
\end{eqnarray}
Thus, by the same argument as in \eqref{eq:multi-source-exponential-converse-proof-1}-\eqref{eq:multi-source-exponential-converse-proof-2} and by noting \eqref{eq:relation-single-shot-q-phi}, we can derive 
\eqref{eq:bound-channel-exponential-converse-1}. 

Now, we restrict the range of $\tilde{\rho}$ so that $\rho(a(R )) < \tilde{\rho} < 1$, and take 
\begin{eqnarray}
\sigma = \frac{\rho(a(R )) - \tilde{\rho}}{1-\tilde{\rho}}.
\end{eqnarray}
Then, by noting \eqref{eq:relation-single-shot-q-phi},
we have \eqref{eq:bound-channel-exponential-converse-2}. 
\qed